\documentclass[arxiv]{imsart}

\RequirePackage{amsthm,amsmath,amsfonts,amssymb}
\RequirePackage[numbers]{natbib}
\RequirePackage[colorlinks,citecolor=blue,urlcolor=blue]{hyperref}
\RequirePackage{graphicx}

\doi{10.1214/154957804100000000}
\pubyear{0000}
\volume{0}
\firstpage{1}
\lastpage{1}
\startlocaldefs
\numberwithin{equation}{subsection}
\newtheorem{theorem}{Theorem}[section]
\newtheorem{lemma}[theorem]{Lemma}
\theoremstyle{remark}
\newtheorem{definition}[theorem]{Definition}

\usepackage{xr}
\usepackage{enumerate}

\input{math_commands.tex}

\newcommand{\A}{\ensuremath{{\bf A}}} %  adjacency matrix
\newcommand{\Kt}{\ensuremath{ \bf{K}}} % normalised adjacency matrix
\newcommand{\At}{\ensuremath{\mathbb A}} % normalised adjacency matrix
\newcommand{\Mt}{\ensuremath{\mathbb M}} % normalised adjacency matrix
\newcommand{\An}{\ensuremath{\overline{\bf A}}} % normalised adjacency matrix
\newcommand{\Ubf}{\ensuremath{{\bf U}}} % normalised adjacency matrix
\newcommand{\Abar}{\ensuremath{\overline{\mathbb A}}} % normalised adjacency matrix
\newcommand{\Atilde}{\ensuremath{\widetilde{\mathbb A}}} % normalised adjacency matrix
\newcommand{\Sn}{\ensuremath{\overline{\bf S}}} % normalised adjacency matrix
\newcommand{\Kn}{\ensuremath{\overline{\bf K}}} % normalised adjacency matrix
\newcommand{\Id}{\ensuremath{{\bf I}_{d \times d}}} %d x d dimensional identity matrix
\newcommand{\Q}{\ensuremath{\textbf{Q}}} % parameter matrix in OU
\newcommand{\MdR}{\ensuremath{{\mathcal{M}_{d}(\R)}}} % d x d real-valued matrix
\newcommand{\MdB}{\ensuremath{{\mathcal{M}_{d}(\{0,1\})}}} % d x d real-valued matrix
\newcommand{\vectorise}{\ensuremath{\mathrm{vec}}}
\newcommand{\stableconv}{\ensuremath{\xrightarrow{\ \mathcal{D}\text{-st}\ }}}
\newcommand{\probconv}{\ensuremath{\xrightarrow{\ P \ }}}

\newcommand{\N}{\ensuremath{\mathbb{N}}}

\newcommand{\AL}{\ensuremath{\textnormal{AL}}}
\newcommand{\supp}{\ensuremath{\textnormal{supp}}}

\newcommand{\diag}{\ensuremath{\mathrm{diag}}}

\newcommand{\trace}{\ensuremath{\mathrm{tr}}}

\newcommand{\RV}{\ensuremath{\mathrm{RV}}}

\newcommand{\indicator}{\ensuremath{\mathbb{I}}}
\newcommand{\proba}{\ensuremath{{\text{P}}}} % d x d complex-valued matrix
\newcommand{\determinant}{\ensuremath{{\text{det}}}} % d x d complex-valued matrix
\newcommand{\YY}{\ensuremath{\mathbb{Y}}} %Y vector
\newcommand{\DD}{\ensuremath{\mathbb{D}}} %Y vector
\newcommand{\HH}{\ensuremath{\mathbb{H}}} %Y vector
\newcommand{\XX}{\ensuremath{\mathbb{X}}} %X vector (sampled process)
\newcommand{\ZZ}{\ensuremath{\mathbb{Z}}} %Z vector (noise vector in the sampled process)
\newcommand{\JJ}{\ensuremath{\mathbb{J}}} %J vector jump process
\newcommand{\Ss}{\ensuremath{\mathbb{S}_d}} % space of ? definite matrix (^+ for positive semidefinite ^++ positive definite, etc).

\newcommand{\Si}{\ensuremath{ {\boldsymbol \Sigma}}}% NxN-dim stochastic volatility matrix 

\newcommand{\boldTheta}{\ensuremath{ {\boldsymbol{\theta}}}}% NxN-dim stochastic volatility matrix 
\newcommand{\boldPsi}{\ensuremath{ {\boldsymbol{\psi}}}}% NxN-dim stochastic volatility matrix 
\newcommand{\boldB}{\ensuremath{ {\boldsymbol{b}}}}% NxN-dim stochastic volatility matrix 

\newcommand{\LL}{\ensuremath{\mathbb{L}}}% N-dim Levy process
\newcommand{\WW}{\ensuremath{\mathbb{W}}}% N-dim Brownian

\DeclareFontFamily{U}{mathx}{\hyphenchar\font45}
\DeclareFontShape{U}{mathx}{m}{n}{
      <5> <6> <7> <8> <9> <10>
      <10.95> <12> <14.4> <17.28> <20.74> <24.88>
      mathx10
      }{}
\DeclareSymbolFont{mathx}{U}{mathx}{m}{n}
\DeclareFontSubstitution{U}{mathx}{m}{n}

\newcommand{\Levy}{L\'{e}vy\ }
\newcommand{\levy}{L\'{e}vy}
\newcommand{\Ito}{It\^{o}\ }
\newcommand{\ito}{It\^{o}}

\newtheorem{assumption}{Assumption}

\newtheorem{corollary}[theorem]{Corollary}

\newtheorem{notation}[theorem]{Notation}

\newtheorem{remark}[theorem]{Remark}

\endlocaldefs

\begin{document}

\begin{frontmatter}
\title{High-frequency Estimation of the \levy-driven Graph Ornstein-Uhlenbeck process}
\runtitle{High-frequency Graph Ornstein-Uhlenbeck process}

% indicate corresponding author with \corref{}
% \author{\fnms{Valentin} \snm{Courgeau}\thanksref{t1}\corref{qweqwe}\ead[label=e1]{valentin.courgeau15@imperial.ac.uk}
% \thankstext{t1}{Correspond author} 
% \address{Department of Mathematics,
%Imperial College London,\\ London, United Kingdom.\\ \printead{e1}\\ \printead{e2}}
%\and
%\author{\fnms{Almut E.D.} \snm{Veraart}\ead[label=e3]{a.veraart@imperial.ac.uk}}
%\address{\printead{e1}}

% indicate corresponding author with \corref{}
\author{\fnms{Valentin} \snm{Courgeau}\thanksref{t1}\corref{}\ead[label=e1]{valentin.courgeau15@imperial.ac.uk}}
 \thankstext{t1}{Correspond author} 
\address{Department of Mathematics,
Imperial College London,\\ London, United Kingdom.\\ \printead{e1}}
\and
\author{\fnms{Almut E.D.} \snm{Veraart}\ead[label=e2]{a.veraart@imperial.ac.uk}}
 \address{Department of Mathematics,
Imperial College London,\\ London, United Kingdom.\\ \printead{e2}}

\runauthor{CV}

%\begin{abstract}
%\end{abstract}

%\begin{keyword}
%\kwd{}
%\kwd{}
%\end{keyword}

% history:
% \received{\smonth{1} \syear{0000}}

\begin{abstract}
We consider the Graph Ornstein-Uhlenbeck (GrOU) process observed on a non-uniform discrete timegrid and introduce discretised maximum likelihood estimators with parameters specific to the whole graph or specific to each component of the graph. Under a high-frequency sampling scheme, we study the asymptotic behaviour of those estimators as the mesh size of the observation grid goes to zero. We prove two stable central limit theorems to the same distribution as in the continuously-observed case under both finite and infinite jump activity for the \Levy driving noise. In addition to providing the consistency of the estimators, the stable convergence allows us to consider probabilistic sparse inference procedures on the edges themselves when a graph structure is not explicitly available. It also preserves its asymptotic properties. In particular, we also show the asymptotic normality and consistency of an Adaptive Lasso scheme. We apply the new estimators to wind capacity factor measurements, i.e.\ the ratio between the wind power produced locally compared to its rated peak power, across fifty locations in Northern Spain and Portugal. We compare those estimators to the standard least squares estimator through a simulation study extending known univariate results across graph configurations, noise types and amplitudes. 
\end{abstract}



\end{frontmatter}

\section{Introduction}

Ornstein-Uhlenbeck (OU) models, and in particular those driven by \Levy processes, form a class of continuous-time processes with a broad range of applications, e.g.\ in finance for pairs trading \citep{hol2018estimation, endres2019tradingLevyDrivenOU} and volatility modelling \citep{barndorff2002bnsmodel, pigorsch2009DefinitionSemiPositiveMultOU}, in electricity management \citep{Longoria2019OULevyElectricityPortfolios} or neuroscience \citep{Melanson2019DataDrivenStationaryJumpDiffusions}. On the other hand, the availability of high-dimensional time series datasets gave rise to sparse inference for OU-type processes \citep{Boninsegna2018SparseLearningDynamical, Gaiffas2019SparseOUProcess, matulewicz2017statistical} as a way to control interactions within complex systems.

In this article, we consider the multivariate \levy-driven Graph Ornstein-Uhlenbeck (GrOU) process on a given graph structure and the maximum likelihood estimators (MLE) introduced in \cite{courgeau2020likelihood}. This process was originally designed to quantify OU-type relationships between \emph{nodes} of a continuously-observed graph. However, in this article, we suppose that observations are placed on a \emph{discrete} time grid as an extension of \cite{zhu2017network} and we derive the asymptotic distributions of those discretised MLEs under a high-frequency sampling scheme \citep{mai2014efficient}. Using a jump-filtering approach, we also discretise the continuous-time results from \cite{courgeau2020likelihood} to prove the stable convergence (in distribution) of the estimators to the same Gaussian distribution as in the continuous-time case when the mesh of the observation grid shrinks to zero. The stable convergence allows to use a complementary inference procedure to obtain a sparse structure, if not directly available (as in \cite{matulewicz2017statistical}), in parallel to the GrOU inference.

OU processes observed on a discrete uniformly-spaced time grid are multivariate AR(1) processes and have been well studied in this case: a least squares (LS) estimator \citep{fasen2013, Hu2009LeastSquares} as well as other inference strategies such as moment-based estimators \citep{Wu2019MoMCompoundPoissonOU} or specific to nonnegative noise increments \citep{Brockwell2007} have been considered. Similarly to \cite{mai2014efficient}, we consider the case of a non-uniform time grid with a double asymptotic assumption on the sample size, the mesh size and the time horizon. This extends the results to event-based datasets \citep{Simonov2017EventBased}.
As suggested by empirical analysis \citep{aitsahalia2011testingFiniteInfinite}, we ought to consider noise types that allow, in short, for both finitely and infinitely many jumps in finite time intervals---called finite and infinite jump activities respectively---and we prove that asymptotic distribution result holds in both cases. As a result, the stable asymptotic convergence holds for a large collection of \Levy driving noises, may it be the sum of Wiener and compound Poisson processes for the former, or say, generalised hyperbolic \Levy motions \citep{Eberlein2001Ghyp} for the latter. By property of the stable convergence, this also opens the door to the stochastic and/or time-dependent modelling of the graph structure.

To broaden the applicability of the GrOU process on high-dimensional but non-graphical datasets, we extend the regularisation approach from \cite{Gaiffas2019SparseOUProcess, courgeau2020likelihood} to the high-frequency framework. We show the asymptotic normality and consistency in variable selection  of their Adaptive Lasso scheme, a sparse inference approach that infers a lean graph structure on which the GrOU model can be applied. 

We also compare quantitively the discretised MLEs with the LS estimator on a network-based electricity dataset, namely, the RE-Europe dataset \citep{jensen2017re}. A simulation study shows that the MLEs are stable across noise types, graph configurations as well as superior to the LS estimator across a variety of graph configurations and driving noises of reasonable amplitudes which corroborates and extends the results from \cite{mai2014efficient}.

The high-frequency stable convergence is proved with a driftless \Levy noise and under second (resp.\ fourth) moment integrability for finite (resp.\ infinite) jump activity. The infinite case also requires another stability assumption: the second moment of the \Levy measure restricted to the hypercube $[-v,v]^d$ should be $O(v^{2-\alpha})$ for some small $v > 0$ where $\alpha$ is the Blumenthal-Getoor index \citep{BlumenthalGetoor1961} of the driving \Levy noise.

We validate the inference power of the MLEs on a subset of the RE-Europe data set \cite{jensen2017re} where hourly electricity-related data has been collected in 2012--2014 and projected on close to 1,500 locations of the electricity network across Europe. More precisely, we focus on the hourly wind capacity factor measurements at 50 of those locations where 24 are located in Portugal and 26 are in Northern Spain along the Atlantic coast. We recall that the wind capacity factor is the proportion of energy produced compared to the peak power production at this location hence between 0 and 1. The data represents the local supply of wind energy that ought to be distributed across the networks: the GrOU model captures the geographical relationships between nodes while maintaining the graphical structure of the electricity grid on which the energy transfer take place. We will infer the model parameters and noise distributions on this dataset and use those parameters in the simulation study. This study focuses on testing the reliability of the estimation across noisy types, graph structures and mesh sizes by recovering the model parameters obtained from parametrically bootstrapped time series.

In Section \ref{section:network-ou}, we recall the definition of the GrOU process when observed on a discrete time grid and introduce the corresponding discretised MLEs. In Section \ref{section:discretised-high-frequency-data}, we describe the high-frequency sampling scheme which allows the asymptotic stable convergence of the discretised GrOU MLEs given therein---with separate treatments of the finite and infinite jump activity cases. Section \ref{section:adaptive-lasso} is devoted the definition of the Adaptive Lasso scheme and the study of its asymptotic properties. Finally, in Section \ref{section:applications}, we apply the discretised MLEs on the subset of the RE-Europe data set \citep{jensen2017re}. Additionally, we perform the recovery of \Levy increments \citep{brockwell2013carma-recovery} to infer the noise distribution via likelihood maximisation or the generalised method of moments. Those fitted distributions are used in the simulation study presented in Section \ref{section:simulation-study}. Finally, the technical lemmas and propositions are stated in Appendix \ref{section:proofs} whilst the proofs are relegated to Appendix \ref{appendix:proofs-finite} \& \ref{appendix:proofs-infinite} for a clearer layout of the argument. Finally, implementation details, R code and reproducible plots are available on GitHub (numerical routines\footnote{\url{https://github.com/valcourgeau/ntwk}}, end-to-end scripts \footnote{\url{https://github.com/valcourgeau/r-continuous-network}}).

\section{The Graph Ornstein-Uhlenbeck process and estimators}
\label{section:network-ou}
In this section, we recall the continuous-time framework of the Graph Ornstein-Uhlenbeck (GrOU) process with two parametrisations as defined in Sections 2--4, \cite{courgeau2020likelihood}. 

\subsection{Notations}
We consider a filtered probability space $(\Omega, \mathcal{F}, (\mathcal{F}_t,\ t \geq 0), \proba_0)$ to which all stochastic processes are adapted. We consider a \Levy process $(\LL_t,\ t \geq 0)$  (i.e.\ a stochastic process with stationary and independent increments and continuous in probability and $\LL_0 = \boldsymbol{0}_{\boldsymbol{d}}, \ \proba_0-a.s.$)  \citep[Remark 1]{Brockwell2009} without loss of generality assumed to be c{\`a}dl{\`a}g. For a stochastic process $(\YY_t,\ t \geq 0)$, we write $\YY_{t-}:= \lim_{s \uparrow t} \YY_s$ for any $t \geq 0$. For any probability measure $\proba$, we denote by $\proba_t$ its restriction to the $\sigma$-field $\mathcal{F}_t$ for any $t \geq 0$. Regarding the convergence of random variables, we write $\xrightarrow{\ a.s. \ }$ the almost-sure convergence, $\xrightarrow{\ L^p\ }$ for the convergence in $L^p$, $\probconv$ for the convergence in probability and $\xrightarrow{\ \mathcal{D} \ }$ for the convergence in distribution and $\stableconv$ the $\mathcal{F}-$stable convergence in distribution.

We denote the matrix determinant by $\determinant$ , the space of $\{0,1\}$-valued $d \times d$ matrices by $\MdB$, the space of real-valued $d \times d$ matrices by $\MdR$, the linear subspace of $d\times d$ symmetric matrices by $\Ss$, the (closed in $\Ss$) positive semidefinite cone (i.e.\ with the real parts of their eigenvalues non-negative) by $\Ss^+$ and the (open in $\Ss$) positive definite cone (i.e.\ with the real parts of their eigenvalues positive) by $\Ss^{++}$. In particular, $\Id \in \MdR$ denotes the $d\times d$ identity matrix.

We denote by $\lambda^{leb}$ the one-dimensional Lebesgue measure. For a non-empty topological space, $\mathcal{B}(S)$ is the Borel $\sigma$-algebra on $S$ and $\pi$ is some probability measure on $(S,\mathcal{B}(S))$. The collection of all Borel sets in $S \times \R$ with finite $\pi \otimes \lambda^{leb}$-measure is written as $\mathcal{B}_b(S \times \R)$. Also, the norms of vectors and matrices are denoted by $\|\cdot \|$. We usually take the Euclidean (or Frobenius) norm but due to the equivalence between norms, our results are not norm-specific and are valid under any norm in $\R^d$ or $\MdR$.
%Also, for $\boldsymbol{x} = (x_1,\dots,x_d)^\top \in \R^d$, we denote by $|x|$ (respectively, $\|x\|$) its $L^1$-norm (respectively, $L^2$-norm) in the sense that $\|\boldsymbol{x}\| :=  \sum_{i=1}^d |x_i|$ (respectively, $\|\boldsymbol{x}\| :=  \boldsymbol{x}^\top \boldsymbol{x}$). For matrices, $\|\cdot \|$ denotes the Frobenius norm.
In addition, for an invertible matrix $\boldsymbol{M} \in \MdR$, we define $\langle \boldsymbol{x}, \boldsymbol{y} \rangle_{\boldsymbol{M}} := \boldsymbol{x}^\top \boldsymbol{M^{-1}} \boldsymbol{y}$ for $\boldsymbol{x}, \boldsymbol{y} \in \R^d$. Finally, for a process $\XX_t= (X^{(1)}_t,\dots,X^{(d)}_t)^\top\in\R^d$ we denote by $[\XX]_t$ the matrix $([X^{(i)},X^{(j)}]_t)$ of quadratic co-variations on $[0,t]$ for $t \geq 0$.

 In this article, $\otimes$ denotes the Kronecker matrix product, $\odot$ is for the Hadamard (element-wise) matrix product and $\vectorise$ is the vectorisation transformation where columns are stacked on one another. We denote by $\vectorise^{-1}(\boldsymbol{x}) := (\vectorise(\Id)^\top \otimes \Id)(\Id \otimes \boldsymbol{x}) \in \MdR$ for $\boldsymbol{x}\in \R^{d^2}$ the inverse vectorisation transformation.

\subsection{The \levy-driven Ornstein-Uhlenbeck process}
We consider a $d$-dimensional Ornstein-Uhlenbeck (OU) process $\YY_t = (Y^{(1)}_t,\dots,Y^{(d)}_t)^\top$ for $t \geq 0$ satisfying the stochastic differential equation (SDE) for a \emph{dynamics} matrix $\Q\in\Ss^{++}$
\begin{equation}
	\label{eq:sde}
	d\YY_t = -\rmQ \YY_{t-}dt + d\LL_t,
\end{equation}
where $\LL_t = (L^{(1)}_t, \dots, L^{(d)}_t)^\top$ is a $d$-dimensional \Levy process independent of $\YY_0$ \citep[Section 1]{Masuda2004}. As presented in Section 2.1.2, \cite{courgeau2020likelihood}, the \Levy process $\LL$ in Equation \eqref{eq:sde} is defined by the \levy-Khintchine triplet $(\boldB, \Si, \nu)$ with respect to the truncation function  $\tau(\boldsymbol{z}) := \mathbb{I}_{\{\boldsymbol{x} \in \R^d : \|\boldsymbol{x}\| \leq 1\}}(\boldsymbol{z})$ where $\mathbb{I}$ denotes the indicator function and $\boldB \in \R^d$, $\Si \in \Ss^{++}$. Also, $\nu$ is a \Levy measure on $\R^d$ satisfying $\int_{\R^d\symbol{92}\{\boldsymbol{0}\}} (1 \wedge \|\boldsymbol{z}\|^2)\nu(d\boldsymbol{z}) < \infty$, see additional details in Appendix \ref{appendix:levy-ito-decomposition}. We consider \emph{strictly stationary} OU processes \citep{Masuda2004, masuda2007ergodicity, Brockwell2007} and we recall a standard assumption in this context:
\begin{assumption}
\label{assumption:strong-solution-vector}
Suppose that $\rmQ \in \Ss^{++}$ and that the \Levy measure $\nu(\cdot)$ satisfies the log moment condition:
\begin{equation*}
    \int_{\|\boldsymbol{z}\|>1}\ln \|\boldsymbol{z}\|\nu(d\boldsymbol{z}) < \infty.
\end{equation*}
\end{assumption}
Then, under Assumption \ref{assumption:strong-solution-vector}, there exists a unique strictly stationary solution to Equation \eqref{eq:sde} given by
\begin{equation}
\label{eq:stationary-solution-vec}
\YY_t = e^{-(t-s)\Q}\YY_s + \int_s^t e^{-(t-u)\Q}d\LL_u, \qquad \text{for any $t \geq s \geq 0$,}
\end{equation}
%We refer to Section 2.2, \cite{courgeau2020likelihood}, to obtain additional information on the corresponding characteristic triplet of $(\YY_t, \ t \geq 0)$.
 Similarly to \cite{fasen2013}, we consider a \levy-driven OU process $\YY$ observed on a discrete time grid $0=t_{1,N} < \cdots < t_{N,N}:=T_N$ where $N \in \N$ which leads to multivariate AR(1) process representation given by
  $$\YY_{t_{k+1,N}} = e^{-\Q(t_{k+1,N}-t_{k,N})}\YY_{t_k,N} + \int_{t_{k,N}}^{t_{k+1,N}}e^{-\Q(s-t_{k,N})}d\LL_s, \quad \text{for} \ k \in \{1,\dots,N\}.$$
  
  \begin{notation}
When there is no ambiguity, we use $t_k$ for $t_{k,N}$ and $\YY_{k}$ for $\YY_{t_k}$.
\end{notation}

We study the asymptotic properties of a couple of discretised estimators for $\Q$ under two parametrisations specific to graphs which we describe in the following section.

\subsection{A graphical interpretation of the OU process}
\label{section:graphical-interpretation}
The components of $\YY$ are interpreted as the \emph{nodes} of a graph structure and the undirected links, or \emph{edges}, between them are given by a so-called \emph{adjacency} matrix (or \emph{graph topology} matrix) $\A = (a_{ij}) \in \MdB$ where $a_{ii} = 0$ for $i \in \{1,\dots,d\}$ and, for $i\neq j$, $a_{ij} = 1$ if nodes $i$ and $j$ are linked and $0$ otherwise. We define the degree of any node $i$ as $n_i := 1 \vee \sum_{j\neq i}a_{ij}$ and we also define the row-normalised adjacency matrix 
$$\An := \diag(n_1^{-1},\dots, n_d^{-1})\A.$$

\begin{assumption}
	\label{assumption:A-known}
	We assume that $\A$ is deterministic, constant in time and known.
\end{assumption}
 Following \cite{courgeau2020likelihood}, the \emph{dynamics} matrix $\Q$ is parametrised in two different ways specific to graphs, i.e.\ they include the adjacency matrix $\A$ and model interactions within the graph. 
 \subsubsection{The two-parameter formulation}
 The $\boldTheta$-GrOU parameter $\boldTheta = (\theta_1, \theta_2)^\top \in \R^2$ give way to the parametrisation
\begin{equation}
\label{eq:Q-using-theta}
\rmQ  = \Q(\boldTheta) :=\theta_2 \Id + \theta_1 \An.
\end{equation}
This yields single \emph{network} and \emph{momentum} effects for the whole graph respectively parametrised by $\theta_1$ and $\theta_2$. %The resulting SDE is given by $$dY^{(i)}_t = - \theta_2 Y^{(i)}_{t-}dt - \theta_1 n_i^{-1}\sum_{j \neq i} a_{ij}Y^{(j)}_{t-} dt + dL^{(i)}_t, \quad t \geq 0,$$
The resulting autoregressive form is given by
$$Y_{k+1}^{(i)} = e^{-\theta_2 (t_{k+1}-t_{k})} Y^{(i)}_k + e^{-\theta_1 \An (t_{k+1}-t_{k})}\YY_{k}+ \int_{t_{k,N}}^{t_{k+1,N}}e^{-\Q(\boldTheta)(s-t_{k,N})}d\LL_s,$$
for $k \in \{1,\dots,N\}$ and $i\in \{1,\dots,d\}$. For this process to be well-defined, $\Q(\boldTheta)$ must be positive definite \citep{Masuda2004} and we consider the following sufficient condition \citep[Proposition 2.2.2]{courgeau2020likelihood}:
\begin{assumption}
\label{assumption:theta-grou}
	Suppose that $\boldTheta \in \boldsymbol{\Theta}:= \{ (\theta_1, \theta_2)^\top \in \R^2: \ \theta_2 > 0, \ \theta_2 > |\theta_1|\}$.
\end{assumption}

\subsubsection{A more detailed formulation}
We also consider the $\boldPsi$-GrOU parameter $\boldPsi \in \R^{d^2}$ which in turn yields the parametrisation
\begin{equation}
\label{eq:Q-using-psi}
\rmQ = \Q(\boldPsi) := (\Id + \An) \odot \vectorise^{-1}(\boldPsi),
\end{equation}
where $\odot$ is the Hadamard product. In this case, \emph{momentum} and \emph{network} effects are \emph{specific} to each node and its neighbours. %Similarly, the componentwise SDE is given by
%$$dY^{(i)}_t = -\psi_{ii}Y^{(i)}_{t-}dt - n_i^{-1}\sum_{j \neq i} a_{ij}\psi_{d(i-1)+j}Y^{(j)}_{t-} dt + dL^{(i)}_t, \quad t \geq 0,$$
Similarly, the autoregressive form is given by
$$Y_{k+1}^{(i)} = e^{-\psi_{ii} (t_{k+1}-t_{k})} Y^{(i)}_k + e^{- \An \odot \vectorise^{-1}(\boldPsi) (t_{k+1}-t_{k})}\YY_{k}+ \int_{t_{k,N}}^{t_{k+1,N}}e^{-\Q(\boldPsi)(s-t_{k,N})}d\LL_s,$$
 for any $i \in \{1,\dots,d\}$. For this process to be well-defined, $\Q(\boldPsi)$ ought to be positive definite \citep{Masuda2004} and we consider the following sufficient condition \citep[Proposition 2.2.3]{courgeau2020likelihood}:
\begin{assumption}
\label{assumption:psi-grou}
Suppose that $$\boldPsi \in \boldsymbol{\Psi} := \Big\{ \phi \in \R^{d^2}: \phi_{d(i-1)+i} > 0, \ \phi_{d(i-1)+i} > n_i^{-1}\sum_{j\neq i}|\phi_{d(j-1)+i}|, \  \forall i \in \{1,\dots,d\}\Big\}.$$
\end{assumption}

\subsection{The Graph Ornstein-Uhlenbeck process}
 
Here, we recall the definition of the  Graph Ornstein-Ulhenbeck (GrOU) process in its two different configurations.
\begin{definition}
	\label{definition:grou}
	The \emph{Graph Ornstein-Uhlenbeck (GrOU)} process is a c{\`a}dl{\`a}g process $(\YY_t,\ t \geq 0)$ satisfying Equation \eqref{eq:sde} with a \Levy process $(\LL_t,\ t \geq 0)$ where $\Q$ is given by either Equation \eqref{eq:Q-using-theta} or by Equation \eqref{eq:Q-using-psi} such that $\Q$ is positive definite. This process is then called a $\boldTheta$-GrOU process or a $\boldPsi$-GrOU process, respectively.
\end{definition}
Since we consider the inference of GrOU processes observed on a discrete time grid, we first introduce the corresponding continuous-time estimators.
\subsubsection{Continuous-time estimators}
To estimate the parameters $\boldTheta$ and $\boldPsi$, \cite{courgeau2020likelihood} propose a likelihood framework under Assumptions \ref{assumption:strong-solution-vector} \& \ref{assumption:A-known} and study the corresponding maximum likelihood estimators (MLE). More precisely, for a $\boldTheta$-GrOU process which is observed in continuous time, has finite second moments and satisfies Assumption \ref{assumption:theta-grou}, the MLE for $\boldTheta = (\theta_1, \theta_2)^\top\in \R^2$ is given by
	$$\widehat{\boldTheta}_t := [\HH^\Si]_t^{-1} \cdot \HH^\Si_t, \quad \text{for any $t \geq 0$,}$$
	where we define
	$$
\HH^\Si_t := - \begin{pmatrix}
\int_{0}^t\langle \An\YY_{s}, d\YY^c_s \rangle_\Si \\
\int_{0}^t\langle \YY_{s}, d\YY^c_s \rangle_\Si
\end{pmatrix} \ \text{such that} \ [\HH^\Si]_t = \begin{pmatrix}
		\int_0^t\langle \An \YY_{s},\An\YY_{s} \rangle_\Si ds & \int_0^t\langle \An\YY_{s},\YY_{s} \rangle_\Si ds\\
		 \int_0^t\langle \An\YY_{s},\YY_{s} \rangle_\Si ds & \int_0^t\langle \YY_{s},\YY_{s} \rangle_\Si ds
	\end{pmatrix}
,$$
with $\YY^c$, the continuous $\proba_0$-martingale part of $\YY$, defined by $d\YY^c_t = -\rmQ \YY_{t-}dt + d\WW_t$ for a $d$-dimensional Brownian motion with covariance matrix $\Si$.
Similarly, for a $\boldPsi$-GrOU process which is observed in continuous time, has finite second moments and satisfies Assumption \ref{assumption:psi-grou}, the MLE for $\boldPsi \in \R^{d^2}$ is given by
	\begin{equation*}
	\widehat{\boldsymbol{\psi}}_t := [\At^\Si]_t^{-1} \cdot \At^\Si_t, \quad \text{for any $t \geq 0$,}
\end{equation*}
where $\At^\Si_t := \: - \int_0^t \YY_s \otimes \Si^{-1}d\YY_s^c$ such that $[\At^\Si]_t = \: {\Kt}_t \otimes \Si^{-1}$ with ${\Kt}_t := \int_0^t \YY_s \YY_s^\top ds$.

%We define the $d\times d$ symmetric matrix ${\Kt}_t:=\int_0^t\YY_s \YY_s^\top ds$ which is $\proba_{0}$-almost surely nonsingular for $t$ large enough \citep[Lemma 3.3.1]{courgeau2020likelihood}. 

\subsubsection{Discretised estimators}
\label{section:jump-filtered-quantities}

Recall that we denote by $\YY^c$ the  $\proba_0$-martingale part of $\YY$. Note that the parameters sets $\boldsymbol{\Theta}$ and $\boldsymbol{\Psi}$ are equivalent for the $\boldTheta$-GrOU process and Assumption \ref{assumption:psi-grou} boils down to Assumption \ref{assumption:theta-grou} in this case. Therefore, in general, we only mention Assumption \ref{assumption:psi-grou}. We use the jump-filtering approach from \cite{mai2014efficient} and discretise the continuous-time estimators defined in \cite{courgeau2020likelihood} as follows:
\begin{definition}
\label{definition:discretised-estimators}
Assume that Assumptions \ref{assumption:strong-solution-vector}, \ref{assumption:A-known} \& \ref{assumption:psi-grou}
%  \& \ref{assumption:psi-grou}--\ref{assumption:mle-convergence} 
hold for a GrOU process $\YY$. Define the increments and their continuous martingale equivalent 
$\Delta_k \YY  := \YY_{{k+1}} - \YY_{{k}}$, and $\Delta_k \YY^{c}  := \YY^{c}_{{k+1}} - \YY^{c}_{{k}}$, componentwise, for any $k \in \{0,\dots,N-1\}$. Define the discretised matrices $$\Sn_N := \Kn_N \otimes \Id \quad \text{where} \quad \Kn_N := \sum_{k=0}^{N-1} \YY_k \YY^\top_k \cdot (t_{k+1} - t_{k}).$$
Also, for a collection $(v_N^{(i)},\ i \in \{1,\dots,d\}) \subset (0,\infty)$, we introduce the $\R^{d^2}$-valued random vectors given by:
\begin{align*}
\begin{cases}
(\Abar_N)_{d(j-1)+i} &= - \sum_{k=0}^{N-1} Y^{(i)}_{k} \cdot \Delta_{k}  Y^{(j),c},\\
(\Atilde_N)_{d(j-1)+i} &= - \sum_{k=0}^{N-1} Y^{(i)}_{k} \cdot \Delta_{k}  Y^{(j)} \cdot \mathbb{I}\{|\Delta Y^{(j)}_{k}| \leq v_N^{(j)}\} \label{eq:estimator-A-jump-filtered},
\end{cases}
\end{align*}
where $i,j \in \{1,\dots,d\}$.
We define the discretised unfiltered estimator as
    \begin{equation*}
    \label{eq:discretised-estimator}
    	\boldsymbol{\overline{\psi}}_N := \Sn_N^{-1} \cdot \Abar_N,
    \end{equation*}
    and the discretised jump-filtered estimator by
    \begin{equation*}
    \label{eq:discretised-filtered-estimator}
        \boldsymbol{\widetilde{\psi}}_N :=  \Sn_N^{-1} \cdot \Atilde_N.
    \end{equation*}
\end{definition}

\begin{remark}
\label{remark:discretised-estimators}
%	We discuss the definition of those estimators in Appendix \ref{proof:definition:discretised-estimators}.
	Both $\widehat{\boldTheta}_t$ and $\widehat{\boldPsi}_t$ feature the diffusion matrix $\Si \in \Ss^{++}$ which is generally not available and we have reformulated them  without $\Si$ as in the proof of Proposition 3.3.2, \cite{courgeau2020likelihood}. 
	Remark that
	\begin{align*}
		\int_0^t\langle \rmQ \YY_s , d\YY^{c}_s \rangle_\Si  &= \vectorise(\Q)^\top \cdot (\Id \otimes \Si^{-1}) \cdot \int_0^t \YY_{s} \otimes d\YY^{c}_s,\\
		\int_0^t\langle \rmQ \YY_s , \Q\YY \rangle_\Si ds  &= \vectorise(\Q)^\top \cdot (\Id \otimes \Si^{-1}) \cdot (\boldsymbol{K}_t \otimes \Id) \cdot \vectorise(\Q),
	\end{align*}
%$$\int_0^t\langle \rmQ \YY_s , d\YY^{c}_s \rangle_\Si  = \vectorise(\Q)^\top \cdot (\Id \otimes \Si^{-1}) \cdot \int_0^t \YY_{s} \otimes d\YY^{c}_s,$$
%and
%$$\int_0^t\langle \rmQ \YY_s , \Q\YY \rangle_\Si ds  = \vectorise(\Q)^\top \cdot (\Id \otimes \Si^{-1}) \cdot (\boldsymbol{K}_t \otimes \Id) \cdot \vectorise(\Q),$$
as given in the proof of Proposition 3.3.2, \cite{courgeau2020likelihood}.
We can rewrite 
$$\begin{cases}
	\At^\Si_t &= \ -(\Id \otimes \Si^{-1})\cdot(\int_0^t \YY_s \otimes d\YY_s^c)\\ [\At^\Si]_t &= \ (\Id \otimes \Si^{-1})\cdot (\boldsymbol{K}_t \otimes \Id)
\end{cases},$$
such that the estimator $\widehat{\boldsymbol{\psi}}_t$ is given by
\begin{equation*}
	\label{eq:reformulation-bold-psi-without-sigma}
	\widehat{\boldsymbol{\psi}}_t := \boldsymbol{S}_t^{-1} \cdot \At_t, \quad \text{where} \quad
\begin{cases}
\At_t &:= \: - \int_0^t \YY_s \otimes d\YY_s^c\\
\boldsymbol{S}_t &:= \: {\Kt}_t \otimes \Id
\end{cases}.
\end{equation*}
We proceed similarly for $\widehat{\boldsymbol{\theta}}_t$.
\end{remark}

\begin{remark}
\label{remark:jump-threshold-delta_n}
	The collection $(v_N^{(i)},\ i \in \{1,\dots,d\})$ consists of jump-filtering threshold values: it relies on the assumption that when noise increments are small, they are most likely due to the continuous martingale part of the \Levy process (i.e.\ the Brownian motion part). On the other hand, larger values would correspond to the contribution of the jump part $\JJ$ (see Appendix \ref{appendix:levy-ito-decomposition}) and should be discarded as such. However, as explored in Section \ref{section:discretised-high-frequency-data} and beyond, the exact values chosen for these thresholds are crucial as they balance the Brownian increments against the jumps. This simple methodology is instrumental for those estimators to be consistent and asymptotically normal as well as usable in practice (Sections \ref{section:adaptive-lasso} \& \ref{section:applications}).
	\end{remark}

Since the continuous martingale part is not available in practice, the discretised jump-filtered estimator $\widetilde{\boldPsi}_N$ is the only one usable for applications. Informally, the main result of this article is that the discretised jump-filtered estimator satisfies the \emph{stable} central limit theorem given by 
$$T_N^{1/2}(\boldsymbol{\widetilde{\psi}}_N - \boldsymbol{\psi}) \stableconv \mathcal{N}\left(\boldsymbol{0}_{d^2},\ \E \left(\YY_\infty \YY_\infty^\top \right)^{-1} \otimes \Si \right), \quad \text{as $N \rightarrow \infty$,}$$
under a collection of assumptions. We also obtain a similar estimator for $\boldTheta$.
\begin{definition}
\label{definition:theta-discretised-estimator}
Assume that Assumptions \ref{assumption:strong-solution-vector}--\ref{assumption:theta-grou} hold for a $\boldTheta$-GrOU process $\YY$. We define the $\boldTheta$-GrOU discretised jump-filtered estimator
$$\widetilde{\boldTheta}_N := d^{-1}\begin{pmatrix}
	\rho(\An) \cdot \vectorise(\A)^\top \\
	\vectorise(\Id)^\top
\end{pmatrix} \cdot  \Sn_N^{-1} \cdot \Atilde_N,$$
where $\rho(\An) := d \cdot (\sum_{i,j}a_{ij}/n_i)^{-1}\cdot \indicator_{\{\|\boldsymbol{X}\| \neq 0\}}(\An)$.
\end{definition}

This estimator is directly obtained from the formulation of $\widehat{\boldPsi}_t$ above
	since $\vectorise(\Q(\boldTheta)) = \theta_2\vectorise(\Id) +\theta_1\vectorise(\An)$. We have
\begin{equation*}
	\widehat{\boldTheta}_t = d^{-1}\begin{pmatrix}
	\rho(\An) \cdot \vectorise(\A)^\top \\
 \vectorise(\Id)^\top
\end{pmatrix} \cdot  \boldsymbol{S}_t^{-1} \cdot \At_t, \quad \text{for any $t \geq 0$,}
\end{equation*}
where $\rho(\An) := d \cdot (\sum_{i,j}a_{ij}/n_i)^{-1}\cdot \indicator_{\{\|\boldsymbol{X}\| \neq 0\}}(\An)$.
%
%\begin{proof}
%%	See Appendix \ref{proof:definition:theta-discretised-estimator}.
%
%	Since $\vectorise(\Q(\boldTheta)) = \theta_2\vectorise(\Id) +\theta_1\vectorise(\An)$, for $\widehat{\boldTheta}_t$ we have
%\begin{equation*}
%	\widehat{\boldTheta}_t = d^{-1}\begin{pmatrix}
%	\rho(\An) \cdot \vectorise(\A)^\top \\
% \vectorise(\Id)^\top
%\end{pmatrix} \cdot  \boldsymbol{S}_t^{-1} \cdot \At_t, \quad \text{for any $t \geq 0$,}
%\end{equation*}
%where $\rho(\An) := d \cdot (\sum_{i,j}a_{ij}/n_i)^{-1}\cdot \indicator_{\{\|\boldsymbol{X}\| \neq 0\}}(\An)$.
%
%\end{proof}

 Note that if each node has at least one neighbour, we have $\rho(\An) =\indicator_{\{\|\boldsymbol{X}\| \neq 0\}}(\An)$. 
We deduce that it suffices to prove the asymptotic convergence properties for the $\boldPsi$-GrOU and deduce the result for the $\boldTheta$-GrOU process using Definition \ref{definition:theta-discretised-estimator}.

To convert the results involving the continuous martingale part of $\YY$ to the practical jump-filtered estimator, we split the convergence study into two sections for both finite and infinite jump activities.

\section{Asymptotic theory with discretised high-frequency observations}
\label{section:discretised-high-frequency-data}
We extend the high-frequency sampling scheme from Section 1, \cite{mai2014efficient}, to discretise the observations of the GrOU process and study the asymptotic distributions of $\widetilde{\boldPsi}_N$ and $\widetilde{\boldTheta}_N$.
%In this section, we describe the practical framework in which the inference procedure is going to be used - that is, with discretised data sampled from a system evolving continuously in time. Since we allow for jumps, it is often accepted that we should be in a high-frequency sampling scheme which will be instrumental in proving that the discretised central limit theorem for $\boldsymbol{\psi}_t$ follows through. 

\subsection{High-frequency framework}
\label{section:high-frequency-setup}
The framework consists of \emph{two} long-span asymptotic assumptions:
\begin{assumption}
\label{assumption:high-frequency-asymptotics}
	Let $(\YY_t,\ t \geq 0)$ be a GrOU process. Suppose that we are given observations $\YY_{t_{1,N}},\dots,\YY_{t_{N,N}}$ at times $0=t_{1,N} < \dots < t_{N,N} := T_N$. Also, suppose that, as $N \rightarrow \infty$, we have $T_N \rightarrow \infty$, a smaller mesh size $\Delta_N := \max_i |t_{i+1,N}-t_{i,N}| \rightarrow 0$ and a relative growth such that 
	$$\Delta_N^{1/2}T_N = o(1) \quad \text{and} \quad N \Delta_N T_N^{-1} = O(1), \quad \text{as $N \rightarrow \infty$.}$$
\end{assumption}
The framework described in Assumption \ref{assumption:high-frequency-asymptotics} is said to be a \emph{high-frequency sampling scheme} because $\Delta_N$ goes faster to zero than $T_N^2$ goes to $\infty$. The number of samples $N$ is curbed since $N\Delta_N$ must grow as fast as $T_N$ as $N\rightarrow \infty$. Note that this assumption implies that data is synchronously available across all nodes or components of $\YY$. The asynchronous case is not relevant for our application of interest, but is important for other applications; see \cite[Chapters 3 \& 4]{cheang2018three} and \cite{AitSahalia2010CovarianceAsynchronous}  for respectively autoregressive and log-Gaussian statistical approaches whilst \cite{bikowski2017autoregressive} proposes a recurrent neural network approach. All of these models are applied on \emph{asynchronous} financial market closing prices. However, this is beyond the scope of this article and we leave this open problem for future research.

By extending Assumption 3.1, \cite{mai2014efficient} to the multivariate case, we provide the following additional sets of assumptions to obtain the asymptotic consistency and normality of the discretised estimators. First, regarding the GrOU process and the \Levy process $\LL$: 
\begin{assumption}
\label{assumption:mle-convergence}
Let $(\YY_t,\ t \geq 0)$ be a GrOU process. Assume that:
\begin{enumerate}[(i)]
    \item \label{assumption:mle-convergence:sq-integrability} the square integrability of the GrOU process \YY 
    %the random vector $\YY$ solving Equation \eqref{eq:sde} is componentwise square integrable
    , i.e. $\E\left(\|\YY_t\|^2\right) < \infty$\label{assumption:enum-sq-integrability}
    \item \label{assumption:mle-convergence:drift} the underlying \Levy noise drift w.r.t the truncation function $\tau(\boldsymbol{z}) := \mathbb{I}_{\{\boldsymbol{x} \in \R^d : \|\boldsymbol{x}\| \leq 1\}}(\boldsymbol{z})$
    is zero, that is $\boldB = \boldsymbol{0}_d$. % \label{assumption:enum-drift-zero}
    % \item We observe simultaneously all nodes at discrete time points $0=t_{1,N} < t_{2,N} < \cdots < t_{N,N} := T_N$, such that as $N \rightarrow \infty$:
    % \begin{enumerate}
    %     \item The time horizon gets infinitely larger, $T_N \rightarrow \infty$; 
    %     \item The mesh size srinks to zero $\Delta_N := \max_{0 \leq i \leq N-1}\left\{|t_{i+1,N}-t_{i,N}|\right\}\rightarrow 0$;
    %     \item $N$ grows at most proportionally to $\Delta_N^{-1}T_N$ in the sense that $N \cdot \Delta_N \cdot T_N^{-1} = O(1)$; 
    % \end{enumerate} 
\end{enumerate}
\end{assumption}

We impose an additional assumption on the high-frequency sampling scheme:
\begin{assumption}
\label{assumption:cv-rate-beta}
We assume that there exists $\beta^{(i)} \in (0,1/2)$ such that 
$$T_N \Delta_N^{(1-2\beta^{(i)}) \wedge 1/2} = o(1), \quad \text{for any $i \in \{1,\dots,d\}$.}$$ % \label{assumption:enum-jump-betas}
%\begin{enumerate}[(i)]
%\setcounter{enumi}{2}
%    \item  There exist $\beta^{(i)} \in (0,1/2)$ such that the mesh size satisfies $T_N \Delta_N^{(1-2\beta^{(i)}) \wedge 1/2} = o(1)$ for any $i \in \{1,\dots,d\}$ \label{assumption:enum-jump-betas}
%\end{enumerate}
\end{assumption}

\subsection{Asymptotic convergence of discretised estimators}
 A key result about the asymptotic distribution of the discretised \emph{unfiltered} estimator $\boldsymbol{\overline{\psi}}_N$ is given as follows:
 \begin{lemma}{(Discretised unfiltered estimator convergence)\\}
\label{lemma:cv-estimator-continuous-component}
Assume Assumptions \ref{assumption:strong-solution-vector}--\ref{assumption:A-known} \& \ref{assumption:psi-grou}-\ref{assumption:cv-rate-beta} hold. The discretised unfiltered estimator $\boldsymbol{\overline{\psi}}_N$ satisfies 
$$T_N^{1/2}\left(\boldsymbol{\overline{\psi}}_N - \boldsymbol{\psi}\right) \stableconv \mathcal{N}\left(\boldsymbol{0}_{d^2},\ \E\left(\YY_\infty \YY_\infty^\top \right)^{-1} \otimes \Si \right), \quad \text{as $N\rightarrow\infty$.}
$$
\end{lemma}
\begin{proof}
	See Appendix \ref{proof:lemma:cv-estimator-continuous-component}.
\end{proof}

From this, we obtain the consistency property:
\begin{corollary}
	\label{corollary:consistency-estimator-continuous-component}
	Assume Assumptions \ref{assumption:strong-solution-vector}--\ref{assumption:A-known} \& \ref{assumption:psi-grou}-\ref{assumption:cv-rate-beta} hold. Then, the discretised unfiltered estimator $\boldsymbol{\overline{\psi}}_N$ is consistent: we have $\boldsymbol{\overline{\psi}}_N \probconv \boldPsi$ as $N\rightarrow\infty$.
\end{corollary}
\begin{proof}
	Lemma \ref{lemma:cv-estimator-continuous-component} yields the stable convergence of $T_N^{1/2}(\boldsymbol{\overline{\psi}}_N -\boldPsi)$ towards of Gaussian distribution with a finite and constant covariance matrix. Thus, $\boldsymbol{\overline{\psi}}_N -\boldPsi$ has a $O(T_N^{-1/2})$ vanishing covariance matrix which becomes the Dirac delta distribution at zero as $N\rightarrow\infty$. From Corollary 3.6, \cite{hausler2015whystableconvergence}, this ensures the convergence in probability to zero of $\boldsymbol{\overline{\psi}}_N -\boldPsi$.
\end{proof}

\subsubsection{Jump thresholds}
\label{section:jump-thresholds}
The jump thresholds play an important role in controlling how much of the noise increments are attributed to the continuous martingale part or as jumps (Remark \ref{remark:jump-threshold-delta_n}). We use the thresholds defined by $$v_N^{(j)} := \Delta_N^{\beta^{(j)}}, \quad \text{for $\beta^{(j)} \in (0, 1/2)$ and $j \in \{1\dots, d\}$},$$
for which the parameter $\beta^{(j)}$ allows for larger increments (as $\beta^{(j)} \rightarrow 0$) or narrows down to the usual Gaussian scale (as $\beta^{(j)} \rightarrow 1/2$) as a power of the mesh size. In a sense, the asymptotic behaviour relies on $\Delta_N$ in the theoretical study ($\beta^{(j)}$ being given as a parameter) and on $\beta^{(j)}$ in practice as it needs to be estimated ($\Delta_N$ being imposed by the data at hand). The data-based aspects of $\beta^{(j)}$ and the mesh size $\Delta_N$ are discussed in Sections \ref{section:fitting-recovered-increments} \& \ref{section:comparing-least-squares} respectively.

 We generalise the central limit theorem above to jump-filtered discretised estimators in both the finite and infinite jump activity cases. The technical lemmas are presented in Appendix \ref{section:proofs} and proved in Appendix \ref{appendix:proofs-finite}.

\subsubsection{Finite jump activity case}
\label{section:finite-activity}

In this section, we consider the finite jump activity case, i.e.
$\int_{\R^d} \nu(d\boldsymbol{z}) < \infty.$
The finite jump activity case is the case where the \Levy noise corresponds to the sum of a Brownian Motion and a compound Poisson process, see Appendix \ref{appendix:levy-ito-decomposition}.
%\subsection{Main result}
An additional assumption regarding the distribution of smaller jumps is required and reads as follows: 
\begin{assumption}
	\label{assumption:jump-height-finite-activity} Assume that $$\int_{\R^d} \nu(d\boldsymbol{z}) < \infty.$$
	In addition, suppose that the marginal distributions of the jump heights are such that $$F^{(i)}\left(2\Delta_N^{\beta^{(i)}}\right) - F^{(i)}\left(-2\Delta_N^{\beta^{(i)}}\right) = o(T_N^{-1}),$$
    where $F^{(i)}$ is the $i$-th node jump marginal cumulative distribution function for $i \in \{1,\dots,d\}$,
\end{assumption}
%The sub-assumption  \eqref{assumption:jump-height-finite-activity} is quite informative on the order of $\Delta_N$ since we require that for $T_N \rightarrow \infty$, $\Delta_N$ shrinks fast enough towards zero such that $T_N \Delta^{1-2\beta^{(i)}}_N = o(1)$ for any $i$, and $T_N \Delta^{1/2}_N = o(1)$ (Remark 3.3, \cite{mai2014efficient}).

To obtain confidence intervals on the GrOU parameter estimation, we prove a central limit theorem for the discretised jump-filtered $\boldPsi$-GrOU estimator under finite jump activity as follows:
\begin{theorem}
\label{theorem:discrete-vec-finite-activity}
Assume Assumptions \ref{assumption:strong-solution-vector}, \ref{assumption:A-known} \& \ref{assumption:psi-grou}--\ref{assumption:jump-height-finite-activity} hold for a GrOU process $\YY$ and $(\beta^{(i)}, \ i \in \{1,\dots,d\}) \subset (0,1/2)$, then for $v^{(i)}_N = \Delta^{\beta^{(i)}}_N$, $i \in \{1,\dots,d\}$, the discretised jump-filtered estimator $\boldsymbol{\widetilde{\psi}}_N$ satisfies 
$$T_N^{1/2}(\boldsymbol{\widetilde{\psi}}_N - \boldsymbol{\psi}) \stableconv  \mathcal{N}\left(\boldsymbol{0}_{d^2},\ \E\left(\YY_\infty \YY_\infty^\top \right)^{-1} \otimes \Si \right), \quad \text{as $N \rightarrow \infty$.}$$
\end{theorem}
\begin{proof}
By Slutsky's lemma and Lemmas \ref{lemma:cv-estimator-continuous-component} \& \ref{lemma:consistency-conservation-jump-filtering-finite}.
\end{proof}

Next, from its definition, the $\boldTheta$-GrOU parameters follows a similar central limit theorem.
\begin{theorem}
\label{theorem:discrete-vec-finite-activity-2-params}
Assume Assumptions \ref{assumption:strong-solution-vector}--\ref{assumption:theta-grou} \& \ref{assumption:high-frequency-asymptotics}--\ref{assumption:jump-height-finite-activity} hold for a $\boldTheta$-GrOU process $\YY$ and $(\beta^{(i)}, \ i \in \{1,\dots,d\}) \subset (0,1/2)$, then for $v^{(i)}_N = \Delta^{\beta^{(i)}}_N$, $i \in \{1,\dots,d\}$, the discretised jump-filtered estimator $\boldsymbol{\widetilde{\psi}}_N$ satisfies 
$$T_N^{1/2}(\widetilde{\boldTheta}_N - \boldTheta) \stableconv \mathcal{N}\left(\boldsymbol{0}_{2},\ \boldsymbol{\widetilde{G}}_\infty \right), \quad \text{as $N \rightarrow \infty$,}$$
where $\boldsymbol{\widetilde{G}}_\infty$ is defined as
$$
d^{-2}\begin{pmatrix}
	\rho(\An) \cdot \vectorise(\A)^\top \\
	\vectorise(\Id)^\top 
\end{pmatrix}  \cdot \E\left(\YY_\infty \YY_\infty^\top \right)^{-1} \otimes \Si \cdot 
\begin{pmatrix}
	\rho(\An) \cdot \vectorise(\A)^\top \\
	\vectorise(\Id)^\top
\end{pmatrix}^\top,
$$
where $\rho(\An) := d \cdot (\sum_{i,j}a_{ij}/n_i)^{-1}\cdot \indicator_{\{\|\boldsymbol{X}\| \neq 0\}}(\A)$.

\end{theorem}
\begin{proof}
	This follows from Theorem \ref{theorem:discrete-vec-finite-activity} and the definition of $\widetilde{\boldTheta}_N$.
\end{proof}

\begin{remark}
\label{remark:finite-compound-poisson}
Under Assumptions \ref{assumption:mle-convergence} and \ref{assumption:jump-height-finite-activity}, the \levy-\Ito decomposition (see Appendix \ref{appendix:levy-ito-decomposition}) yields that there exists a centred Gaussian \Levy process $\WW_t$ with covariance matrix $\Si$ and almost-surely continuous paths as well as a compound Poisson process $\JJ_t$ independent of $\WW_t$ such that
$$\LL_t = \WW_t + \JJ_t.$$
 Indeed, there exists a family of Poisson processes $N_t^{(j)}$ with elementwise intensity $\boldsymbol{\lambda} = (\lambda^{(1)},\dots,\lambda^{(d)}) := \nu(\R^d) < \infty$ such that $\JJ_t^{(j)}=\sum_{k=0}^{N_t^{(j)}}Z^{(j)}_k$ where $Z^{(j)}_k$ are i.i.d.\ jump heights with distribution $F^{(i)}$ as in Assumption \ref{assumption:jump-height-finite-activity}.	
\end{remark}

To prove those afore-mentioned results, we rely on the consistency of the jump-filtered estimator $\widetilde{\boldPsi}_N$ with respect to the unfiltered estimator $\boldsymbol{\overline{\boldPsi}}_N$ given in the following lemma:

\begin{lemma}{(Consistency of the jump filtering with finite jump activity)\\}
\label{lemma:consistency-conservation-jump-filtering-finite}
Assume Assumptions \ref{assumption:strong-solution-vector}--\ref{assumption:A-known} \& \ref{assumption:psi-grou}--\ref{assumption:jump-height-finite-activity} hold for a GrOU process $(\YY_t,\ t \geq 0)$ and $(\beta^{(i)}: i \in \{1,\dots,d\})\subset (0,1/2)$.
If $v_N^{(i)} := \Delta_N^{\beta^{(i)}}$, then we have as $N \rightarrow \infty$:
$$T_N^{1/2}(\boldsymbol{\widetilde{\psi}}_N - \boldsymbol{\overline{\psi}}_N) \probconv \boldsymbol{0}_{d^2}, \quad \text{as $N\rightarrow\infty$.}
$$
\end{lemma}
\begin{proof}
	See Appendix \ref{proof:lemma:consistency-conservation-jump-filtering-finite}.
\end{proof}

From Corollary \ref{corollary:consistency-estimator-continuous-component}, this means that $\boldsymbol{\widetilde{\psi}}_N$ is also consistent.

\subsubsection{Infinite jump activity case}
\label{section:infinite-activity}
The \emph{infinite activity models} have been investigated numerous times as they provide realistic scenarios in financial applications, avoiding for instance to introduce a Brownian component as empirically determined in \cite{carr2002}. More precisely, it entails assuming that $\int_{\R^d} \nu(d\boldsymbol{z}) = \infty$.

The \levy-\Ito decomposition with respect to $\tau(\boldsymbol{z}) := \mathbb{I}_{\{\boldsymbol{x} \in \R^d : \|\boldsymbol{x}\| \leq 1\}}(\boldsymbol{z})$ describes the pure-jump process $\JJ$ as the sum of jumps larger than one $\JJ^1$ and jumps smaller than one $\JJ^2$, such that $\JJ \overset{d}{=} \JJ^1 + \JJ^2$. Moreover, the infinite jump activity case features an infinite number of those small jumps in any time interval. In this context, an additional compensation on the \Levy measure $\nu$ is provided to control $\JJ^{2}$, see Appendix \ref{appendix:levy-ito-decomposition}. We also extend Assumptions 4.1, \citep{mai2014efficient} to our framework as follows:

\begin{assumption}
\label{assumption:infinite_activity}
For a GrOU process $\YY$, assume that
$$\int_{\R^d}\nu(d\boldsymbol{z}) = \infty,$$
as well as the following properties:
\begin{enumerate}[(i)]
    \item \label{assumption:infinite_activity:cross-moments} $\YY$ has bounded fourth-order moments, i.e.\ for any $t \geq 0$, $\E\left(\|\YY_t\|^4\right) < \infty.$

    \item \label{assumption:infinite_activity:squared_activity} The squared magnitude of small jumps remains controlled in the sense that there exists $\alpha \in (0,2)$ such that, as $v \rightarrow 0$, we have
    $\int_{[-v,v]^d} \|\boldsymbol{x}\|^2 \nu(d\boldsymbol{x}) = O(v^{2-\alpha}).$ \label{assumption:enum-infinite-blumenthal}
    \item \label{assumption:infinite_activity:jumps_2_epsilon}there exist $\epsilon_0 > 0$ and $n_0 \in \mathbb{N}$ such that for all $\epsilon \leq \epsilon_0$ and $n \geq n_0$, one has
    $$\E\left(\Delta_k\JJ^{2,(i)}\indicator\{|\Delta_k\JJ^{2,(i)}| < \epsilon\}\right) = 0,$$
    for all $k \in \{0,\dots,N-1\}$ and $i\in\{1,\dots,d\}$.
\end{enumerate}
\end{assumption}
%\begin{remark}
%As pinpointed in \cite{mai2014efficient}, the quantity $\alpha$ discussed above is often referred to as the Blumenthal-Getoor index and requiring that its value shall be between zero and two is a standard assumption in the jump filtering methods \citep{macini2011}.
%\end{remark}
The relationship between the small jumps behaviour between the finite and infinite cases is given by Remark 4.3, \cite{mai2014efficient}: if the \Levy measure $\nu$ is finite, then Assumption \ref{assumption:jump-height-finite-activity} implies Assumption \ref{assumption:infinite_activity}-\eqref{assumption:enum-infinite-blumenthal} if $$F^{(i)}\left((-\Delta_N^{\beta^{(i)}},\Delta_N^{\beta^{(i)}})\right) = O(\Delta_N^{\beta^{(i)}}), \quad \text{for all }i \in \{1,\dots,d\}.$$ 
This means that a small-jump density that is regular enough around zero is sufficient to imply the jump-related regularity conditions. In addition, Assumption \ref{assumption:infinite_activity}-\eqref{assumption:infinite_activity:jumps_2_epsilon} yields that small jumps described by the \Levy measure of $\JJ^2$ (i.e.\ $\nu$ restricted to $[-1,1]$) are centred around zero. A related result is that if the \Levy measure of $\JJ^{2,(j)}$ is \emph{symmetric} (potentially restricted to an interval $(-\epsilon_0, \epsilon_0)$) then Assumption \ref{assumption:infinite_activity}-\eqref{assumption:infinite_activity:jumps_2_epsilon} holds \citep[Lemma 4.4]{mai2014efficient}. This includes symmetric $\alpha$-stable distributions or zero-mean generalised hyperbolic distributions \citep{Eberlein2001Ghyp} such as the Student's $t$-distribution, the Variance-Gamma distribution or zero-mean Normal-inverse Gaussian distribution (see Section \ref{section:applications}).

We hereby propose to adapt the proofs of Section 4, \cite{mai2014efficient} to the multivariate framework. As a stepping stone for the consistency and asymptotic normality of the jump-filtered discretised estimators, we prove the consistency of the jump-filtered estimator with respect to the unfiltered one similar to Lemma \ref{lemma:consistency-conservation-jump-filtering-finite} as follows:
\begin{lemma}{(Consistency of the jump filtering with infinite jump activity)\\}
\label{lemma:consistency-conservation-jump-filtering-infinite}
Suppose Assumptions \ref{assumption:strong-solution-vector}, \ref{assumption:A-known}, \ref{assumption:psi-grou}--\ref{assumption:cv-rate-beta} \& \ref{assumption:infinite_activity} hold for a $\boldPsi$-GrOU process $\YY$ and $(\beta^{(i)}, i\in \{1,\dots,d\})\subset(0,1/2)$, then for $v_N^{(i)}=\Delta_N^{\beta^{(i)}}$, $i\in \{1,\dots,d\}$, we have
$$T_N^{1/2}(\boldsymbol{\widetilde{\psi}}_N - \boldsymbol{\overline{\psi}}_N) \probconv \boldsymbol{0}_{d^2}, \quad \text{as $N\rightarrow\infty$.}
$$
\end{lemma}
\begin{proof}
	See Appendix \ref{proof:lemma:consistency-conservation-jump-filtering-infinite}.
\end{proof}

\begin{theorem}
\label{th:discrete_clt_infinite}
Suppose Assumptions \ref{assumption:strong-solution-vector}, \ref{assumption:A-known}, \ref{assumption:psi-grou}--\ref{assumption:cv-rate-beta} \& \ref{assumption:infinite_activity} hold for a $\boldPsi$-GrOU process $\YY$ and $(\beta^{(i)}, i\in \{1,\dots,d\})\subset(0,1/2)$, then for $v_N^{(i)}=\Delta_N^{\beta^{(i)}}$, $i\in \{1,\dots,d\}$, the discretised jump-filtered estimator $\boldsymbol{\widetilde{\psi}}_N$ satisfies
$$T_N^{1/2}(\boldsymbol{\widetilde{\psi}}_N - \boldsymbol{\psi}) \stableconv \mathcal{N}\left(\boldsymbol{0}_{d^2},\ \E\left(\YY_\infty \YY_\infty^\top \right)^{-1} \otimes \Si \right), \quad \text{as $N\rightarrow\infty$.}
$$
\end{theorem}
%\begin{proof}
%	See Appendix \ref{proof:th:discrete_clt_infinite}.
%\end{proof}
\begin{proof}
%\begin{proof}[Proof of Theorem \ref{th:discrete_clt_infinite}]
By Lemma \ref{lemma:cv-estimator-continuous-component}, we know that the discretised unfiltered estimator $\boldsymbol{\overline{\psi}}_N$ converges stably in distribution as follows:
$$T_N^{1/2}(\boldsymbol{\overline{\psi}}_N - \boldsymbol{\psi}) \stableconv \mathcal{N}\left(\boldsymbol{0}_{d^2},\ \mathbb{E}\left\{\YY_\infty \YY_\infty^\top \right\}^{-1} \otimes \Sigma \right).$$
By Lemma \ref{lemma:consistency-conservation-jump-filtering-infinite} we have that
$T_N^{1/2}(\boldsymbol{\overline{\psi}}_N- \boldsymbol{\widetilde{\psi}}_N) \probconv 0, \quad \text{as $N\rightarrow\infty$.}$
We conclude by Slutsky's lemma.	
\end{proof}

Similarly, for the $\boldTheta$-GrOU estimator, we obtain the following theorem
\begin{theorem}
\label{theorem:discrete-vec-infinite-activity-2-params}
Suppose Assumptions \ref{assumption:strong-solution-vector}--\ref{assumption:theta-grou}, \ref{assumption:high-frequency-asymptotics}--\ref{assumption:cv-rate-beta} \& \ref{assumption:infinite_activity} hold for a $\boldTheta$-GrOU process $\YY$ and $(\beta^{(i)}, \ i \in \{1,\dots,d\}) \subset (0,1/2)$, then for $v^{(i)}_N = \Delta^{\beta^{(i)}}_N$, $i \in \{1,\dots,d\}$, the $\boldTheta$-GrOU discretised jump-filtered estimator $\widetilde{\boldTheta}_N$ satisfies 
$$T_N^{1/2}(\widetilde{\boldTheta}_N - \boldTheta) \stableconv \mathcal{N}\left(\boldsymbol{0}_{2},\ \boldsymbol{\widetilde{G}}_\infty \right), \quad \text{as $N \rightarrow \infty$,}$$
where $\boldsymbol{\widetilde{G}}_\infty$ is defined as
$$
d^{-2}\begin{pmatrix}
	\rho(\An) \cdot \vectorise(\A)^\top \\
 \vectorise(\Id)^\top 
\end{pmatrix}  \cdot \E\left(\YY_\infty \YY_\infty^\top \right)^{-1} \otimes \Si \cdot 
\begin{pmatrix}
	\rho(\An) \cdot \vectorise(\A)^\top \\
	\vectorise(\Id)^\top
\end{pmatrix}^\top,
$$
where $\rho(\An) := d \cdot (\sum_{i,j}a_{ij}/n_i)^{-1}\cdot \indicator_{\{\|\boldsymbol{X}\| \neq 0\}}(\A)$.
\end{theorem}
\begin{proof}
	By Theorem \ref{th:discrete_clt_infinite} and the definition of $\widetilde{\boldTheta}_N$.
\end{proof}

As noted above in the finite activity case, from Corollary \ref{corollary:consistency-estimator-continuous-component}, this means that $\boldsymbol{\widetilde{\psi}}_N$ is also consistent.

We have now defined the discretised estimators in Section \ref{section:network-ou} and provided central limit theorems for both the $\boldTheta$-GrOU and $\boldPsi$-GrOU estimators in Section \ref{section:discretised-high-frequency-data} for \Levy driving noises for both finite and infinite jump activities. In the next section, we describe the different methodologies used to infer the GrOU parameters as well as recover the \Levy increments and estimate their distribution as means to perform parametric bootstrap.

\section{Asymptotic convergence of the Adaptive Lasso regularisation}
\label{section:adaptive-lasso}

Both the $\boldTheta$-GrOU and $\boldPsi$-GrOU demand the graph adjacency matrix $\A$ to be fully specified (Eqs.\ \eqref{eq:Q-using-theta} and \eqref{eq:Q-using-psi}), which restricts the applicability of those processes to data that is based solely on graphs. Sparse inference has risen as a way to control the dimensionality of the parameter vector by providing lean graph-like structure in high-dimensional problems \citep{chen2020community, ma2021sparseAutoregressions}.

In this section, we extend the $L^1$-penalised likelihood function from \cite{Gaiffas2019SparseOUProcess, courgeau2020likelihood} to discretely-observed data to estimate a candidate adjacency matrix: (i) which is fully determined by the information given in the dataset; (ii) whose sparsity can be tuned through an explicit parameter. This Lasso-type regularisation scheme is shown to estimate the adjacency matrix consistently and we provide stable asymptotic normality of the dynamics matrix estimator.
\subsection{Adaptive Lasso}
For a \levy-driven OU process which satisfies Eq.\ \eqref{eq:sde} with a true but \emph{unknown} dynamics matrix $\Q_0 \in \Ss^{++}$, a Lasso regularisation which penalises the likelihood with the $L^1$-norm of the dynamics matrix cannot achieve both afore-mentioned properties in, say, the Gaussian case \citep{zou2006adaptive}. An Adaptive Lasso (AL) procedure provides superior theoretical properties as shown in \citep{Gaiffas2019SparseOUProcess} with a Brownian motion as the driving noise and in \citep{courgeau2020likelihood} with a driving \Levy process with continuous-time observations. We solely extend the latter to the high-frequency framework of Section \ref{section:high-frequency-setup}. For fixed parameters $\lambda \geq 0$ and $\gamma > 0$, the AL scheme is defined as
$$\widetilde{\Q}_{\AL, N} := \argmax_{\Q} \ell_N(\Q) - \lambda \|\rmQ \odot |\widetilde{\Q}_N|^{-\gamma}\|_1,$$
where $\odot$ denotes the Hadamard product and where the absolute value in the denominator of the penalty is evaluated elementwise. The discretised log-likelihood $\ell_N(\cdot)$ (equivalent to Eq.\ 6, \cite{courgeau2020likelihood}, which is in continuous time) is given by 
$$\ell_N(\Q) := - \vectorise(\Q)^\top \cdot (\Id \otimes \Si^{-1}) \cdot \Atilde_N  - \frac{1}{2} \vectorise(\Q)^\top \cdot (\widebar{\Kt}_N \otimes \Si^{-1}) \cdot \vectorise(\Q),$$
where $\Kn_N := \sum_{k=0}^{N-1} \YY_k \YY^\top_k \cdot (t_{k+1} - t_{k})$ and the corresponding $d \times d$ jump-filtered MLE matrix $\widetilde{\rmQ}_{N} := - {\widebar{\Kt}_N}^{-1} \cdot \vectorise(\Atilde_N)$ is almost-surely non-zero entries which penalises more those that are expected to be zero. We also recall that for $i,j \in \{1,\dots,d\}$, we have:
$$(\Atilde_N)_{d(j-1)+i} = \vectorise(\Atilde_N)_{ij} = \sum_{k=0}^{N-1} Y^{(i)}_{k} \cdot \Delta_{k}  Y^{(j)} \cdot \mathbb{I}\{|\Delta Y^{(j)}_{k}| \leq v_N^{(j)}\}.$$

\subsection{Asymptotic properties}
We prove that the AL scheme is consistent in variable selection as well as asymptotically normal under the existence of an oracle that knows the true dynamics matrix $\rmQ_0$. 

Note that the penalty parameter $\lambda$ is implicitly a function of the time horizon $T_N$ which we write $\lambda = \lambda(T_N)$. The support of vector or matrix is defined as follows:
\begin{notation}
	The support of a vector or a matrix $x$ is denoted $\supp(x)$ and is defined as the set of indices of non-null coordinates of $x$. For a given set of indices $\gI$, we denote by $x_{|\gI}$, the restriction of $x$ to the indices in $\gI$ and $x_{|\gI\times\gI}$ to the indices in $\gI\times\gI$.
\end{notation}

We formalise the oracle asymptotic properties in the following theorem:

\begin{theorem}{(Adapted from Th.\ 5.3.1, \cite{courgeau2020likelihood})}
\label{theorem:adaptive-lasso}
	Suppose that Assumptions \ref{assumption:strong-solution-vector}, \ref{assumption:A-known}, 
%	\ref{assumption:high-frequency-asymptotics}, \ref{assumption:mle-convergence}, \ref{assumption:cv-rate-beta} 
\ref{assumption:psi-grou}--\ref{assumption:cv-rate-beta} 
%\ref{assumption:high-frequency-asymptotics}--\ref{assumption:cv-rate-beta} 
	hold for a \levy-driven OU process $\YY$ with a true but unknown dynamics matrix $\rmQ_0 \in \Ss^{++}$.  Also, assume that either Assumption \ref{assumption:jump-height-finite-activity} or that Assumption \ref{assumption:infinite_activity} hold. For a fixed $\gamma > 0$, assume that $\lambda = \lambda(T_N)$ satisfies $\lambda(T_N) N^{1/2} \rightarrow 0$ and $\lambda(T_N) N^{(1+\gamma)/2}\rightarrow \infty$ as $N \rightarrow \infty$. Then, we obtain the following properties: %under the assumption that $\rmQ_0$ is known, we obtain:
	\begin{enumerate}
		\item Consistency of the variable selection: $\proba\left(\supp(\widehat{\rmQ}_{\AL,N}) = \supp(\rmQ_0) \right) \rightarrow 1$ as $N \rightarrow \infty$.
		\item Asymptotic normality: 
		 $$T_N^{1/2}\left(\vectorise(\widetilde{\rmQ}_{\AL,N}) - \vectorise(\rmQ_0)\right)_{|\gQ_0} \xrightarrow{\ \mathcal{D} \ } \mathcal{N}\left(\boldsymbol{0}_{d^2}, \E\left( \YY_\infty \YY_\infty^\top\right)^{-1}_{|\gQ_0 \times \gQ_0} \otimes \Si_{|\gQ_0 \times \gQ_0}  \right),$$
		 as $N \rightarrow \infty$ and where $\gQ_0 := \supp(\rmQ_0)$.
	\end{enumerate}
\end{theorem}
\begin{proof}
	See Appendix \ref{proof:theorem:adaptive-lasso}.
\end{proof}

\section{Application to wind capacity factor measurements}
\label{section:applications}

In this section, we present the RE-Europe data set on which we apply the discretised GrOU process whose inference serves as a foundation for the simulation study.

\subsection{Data set and motivation}
The \emph{RE-Europe dataset} \citep{jensen2017re} offers an approximation of the hourly wind capacity factor (between 0 and 1) across 1,500 locations using the COSMO-REA6 data set \citep{bollmeyer2015CosmoRea6} for the years 2012-2014. This data set consists of a grid of wind speed measurements linearly interpolated to obtain wind speed at turbine height (80 m) with a $7\times 7$ $\text{kms}^2$ spatial resolution. It is then further processed to obtain wind capacity factors and projected on the 1,494 locations of the European electricity network \cite[page 7]{jensen2017re} of which we consider a small portion. Wind capacity factors represent the proportion of the maximum power generation that is generated at a given time. We focus on a subset of \emph{50} locations where 24 are located in Portugal and 26 are in Northern Spain along the Atlantic coast (see LHS, Figure \ref{figure:map-violin-node-level}).

This 50-node network reflects how distributed power generation capabilities would behave in reality with operational constraints: four subnetworks are disconnected from one another and power demand needs to be fulfilled by the local generation of power. The GrOU process describes the geographical dependencies between the nodes while maintaining the graphical information of neighbours and disconnected subset of nodes.

The 50-node network and the related wind capacity factors data will be referred to as the \texttt{RE-Europe 50} dataset which comprises of 25,000 hourly data points.

\begin{figure}[htp]
\centering
\begin{tabular}{cc}
	\includegraphics[height=8cm]{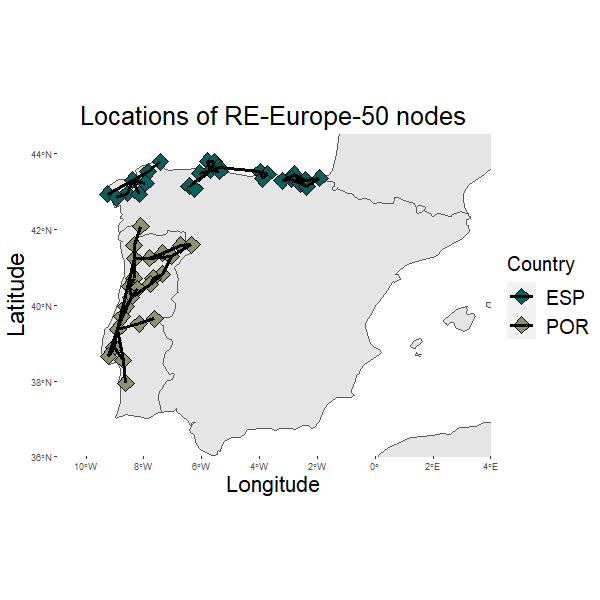} 
	\includegraphics[height=7cm]{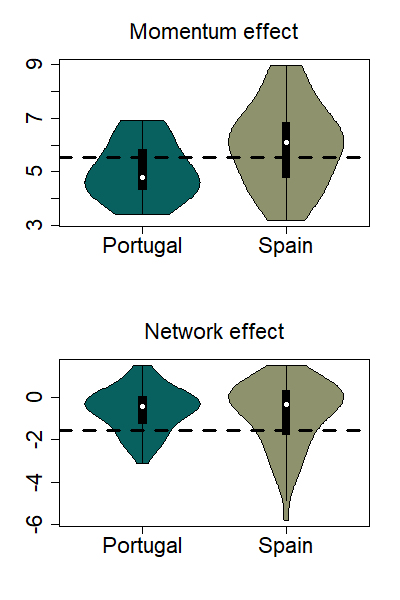}
\end{tabular}  	
  	\caption{LHS: Distribution of the wind power locations in Spain (ESP) and Portugal (POR). RHS: Violin plots of $\boldPsi$-GrOU parameters across Portugal and Spain with $\boldTheta$-GrOU parameters given by the dash lines.}
%\includegraphics[width=0.5\textwidth]{images/vioplot_nodes.png}
%\caption{RE-Europe 50 \Levy increments recovery. LHS: \Levy noise increment marginals. RHS: Q-Q plot of the standardised \Levy noise marginals against Gaussian quantiles.}
\label{figure:map-violin-node-level}
\end{figure}

\subsection{Methodology}
\label{section:methodology}

For the application on the wind capacity factors data set in this section and the simulation study presented respectively in Section \ref{section:simulation-study}, estimating the model parameters is only one of the necessary steps. We devote this section to clarifying the steps involved in preparing the data, obtaining the \Levy increments recovery and inferring the increments' distribution for parametric bootstrapping.

\subsubsection{Algorithm}
\label{section:algorithm}
For applications and the parametric bootstrap procedure, we proceed through the following steps
\begin{enumerate}
	\item Separate trend, seasonal and remainder time series using a LOESS regression or alternative methods \citep{Knight2016ModellingDA};
	\item Fit either the $\boldTheta$-GrOU or $\boldPsi$-GrOU model with the MLE to the remainder time series; \label{algo:step:grou}
	\item Recover the \Levy increments which we denote $(\Delta_k \widehat{\LL}: \ k \in \{1,\dots,N-1\})$ (Section \ref{section:levy-increments-recovery});
	\item Infer the \Levy driving noise distribution from those increments (Section \ref{section:levy-driving-noise}).
\end{enumerate}

In the following sections, we give details on the steps $2$ through $4$. Note that Step 4 relies on the knowledge of adjacency matrix $\A$. 

\subsubsection{With an unknown adjacency matrix}
If $\A$ is not known, one can proceed in two different ways: (i) by inferring it on the same dataset using a particular scheme, such as Lasso regularisation (see Section \ref{section:adaptive-lasso}) in lieu of Step \ref{algo:step:grou} above; (ii) by estimating a sparse adjacency matrix on a potentially different dataset with a consistent estimator $\widehat{\A}_N$ such that $\widehat{\A}_N \probconv \A$ as $N \rightarrow \infty$ \citep{Shojaie2010EstimationAcyclicGraph} where $\A$ is either a deterministic or random matrix. Then, we estimate the dynamics as follows
$$\widetilde{\Q}^\boldTheta_N = \boldsymbol{\widetilde{\boldTheta}}_{1,N} \Id + \boldsymbol{\widetilde{\boldTheta}}_{2,N} \widebar{\widehat{\A}}_N , \quad \text{and} \quad \widetilde{\Q}^\boldPsi_N = (\Id + \widebar{\widehat{\A}}_N) \odot \vectorise^{-1}\left(\boldsymbol{\widetilde{\psi}}_N\right),$$
where $(\widebar{\widehat{\A}}_N)_{ij}=\widehat{n}^{-1}_i \indicator_{\{x \neq 0\}}\left((\widehat{\A}_N)_{ij}\right)$ with $\widehat{n}_i = 1\vee \sum_{k \neq i} (\widehat{\A}_N)_{ik}$ the equivalent to the node degree defined in Section \ref{section:graphical-interpretation}.

Then, by \emph{stable convergence} if $\A$ and $(\widehat{\A}_N)$ are on $(\Omega, \mathcal{F})$, then we have $(\widehat{A}_N, \widetilde{\boldTheta}_N)\stableconv (\A, \mathcal{N}(\boldsymbol{0}_{2},\ \boldsymbol{\widetilde{G}}_\infty ))$ or $(\widehat{A}_N, \widetilde{\boldPsi}_N)\stableconv \left(\A, \mathcal{N}(\boldsymbol{0}_{d^2},\ \E(\YY_\infty \YY_\infty^\top )^{-1} \otimes \Si)\right)$, $\text{as } N \rightarrow \infty$ \citep[Eq.\ 2.3]{jacod1997continuous}. 
If $\A$ is deterministic then the pairwise convergences described above would hold in distribution by a standard argument \citep[p.\ 6--7]{hausler2015whystableconvergence}. On the other hand, if $\A$ is a random variable, the pairwise convergence in distribution will not even hold in general without the stability \citep[Example 1.2]{hausler2015whystableconvergence}. 

 In this sense, the stable convergence in distribution paves the way towards an extension to a random and/or time-dependent adjacency matrix (e.g.\ as an additional, potentially non-ergodic, process) \citep{zambon2019autoregressiveGraphs, ambroise2012new}. 

\subsubsection{\Levy increments recovery}
\label{section:levy-increments-recovery}
In this section, we describe the recovery of \Levy increments to infer their distribution. \cite{brockwell2013carma-recovery} approximate the noise increments from the data and the estimated matrix $\Q$, i.e.\
$$\Delta_k \LL = \Delta_k \YY + \rmQ \int_{t_k}^{t_{k+1}}\YY_{s-}ds, \quad \text{for $k \in \{1,\dots, N-1\}$.}$$
Since $\Q$ is not available in practice, we use the estimates $\Q(\widehat{\boldTheta}_N)$ or $\Q(\widehat{\boldPsi}_N)$ and, for the integral, the $k$-th absolute moment of the approximation error is asymptotically $O(\Delta_N^{2k})$ provided that the process is observed on a uniformly-spaced time grid with increment $\Delta_N$ and that the $k$-th moment of the \Levy process is finite \citep[Proposition 5.4]{brockwell2013carma-recovery}. We adapt to non-uniformly-spaced time intervals given our assumption of an asymptotically zero mesh size $\Delta_N$ as $N \rightarrow \infty$ and to estimate the $1$-step increment, we approximate it crudely by $\int_{t_k}^{t_{k+1}}\YY_{s} ds \approx \frac{1}{2}(t_{k+1} - t_{k}) \cdot \left(\YY_{t_k} + \YY_{t_{k+1}}\right)$, element-wise. Also, note that $\int_{t_k}^{t_{k+1}}\YY_{s}\YY_{s}^\top ds \approx \frac{1}{2}(t_{k+1} - t_{k}) \cdot \left(\YY_{t_k}\YY_{t_k}^{\top} + \YY_{t_{k+1}}\YY_{t_{k+1}}^{\top}\right)$. Those estimated, or \emph{recovered}, \Levy increments are denoted $\Delta_k \widehat{\LL}$, for $k \in \{1,\dots,N-1\}$.

\subsubsection{\Levy driving noise inference}
\label{section:levy-driving-noise}
Once the increments are recovered, we proceed to fit those to a finite and an infinite jump activity distribution. For the former, we first leverage a generalised method of moments inference, as inspired by \cite{Melanson2019DataDrivenStationaryJumpDiffusions}, to estimate the Brownian motion covariance matrix. Jumps are either modelled using a $d$-dimensional Gaussian distribution or a generalised hyperbolic distribution (GHYP), see below. For the latter, we directly model the noise as generalised hyperbolic \Levy motions \citep{Eberlein2001Ghyp}.
\paragraph{Generalised hyperbolic distribution} It is a Normal mean-variance mixture distribution which features semi-heavy tails along with an explicit likelihood function. Thus, it is particularly suitable for jump distribution modelling. We assume that the jumps can be decomposed as follows
$$ \JJ_t \overset{d}{=} \ermI\ermG_t \cdot \boldsymbol{\gamma} + \sqrt{\ermI\ermG_t}\cdot B \cdot \rmZ_t,$$
where $\gamma \in \R^d$, $B \in \MdR$, $\ermI\ermG_t$ is a scalar inverse Gaussian process such that $\E(\ermI\ermG_t) = 1$ independent from $\rmZ_t$---a $d$-dimensional zero-mean Gaussian process with covariance matrix $\Si^J \in \Ss^{+}$. Note that $\JJ_t\ |\ \ermI\ermG_t = \ervi \ervg_t \sim \mathcal{N}(\ervi \ervg_t \cdot \boldsymbol{\gamma}, \ervi \ervg_t \cdot \Si^J)$. We use an Expectation Maximisation algorithm called the \emph{Multi-Cycle Expectation Conditional Maximization (MCECM)} algorithm \citep{McNeil2005QuantRiskManagementEM}, as implemented in the R package \texttt{ghyp}, to fit the full 6-parameter GHYP density directly on increments.

\paragraph{Finite jump activity case}

Following Remark \ref{remark:finite-compound-poisson}, we assume that $\LL$ is the sum of a correlated d-dimensional Wiener process $\WW$ with covariance matrix $\Si$ and a $d$-dimensional compound Poisson process with finite intensity $\lambda > 0$ and zero-mean Gaussian or GHYP jumps. Gaussian jumps are modelled by a $d$-dimensional Gaussian random vector $\mathcal{N}(\boldsymbol{0}_d, \Si^J)$ where $\Si^J \in \Ss^{+}$, similarly to the GHYP notation. In the spirit of \cite{Melanson2019DataDrivenStationaryJumpDiffusions, Hanser2005RealizedVarianceHighFrequency, aitsahalia2011testingFiniteInfinite},   
we calibrate both the correlation matrix of the Wiener process and the compound Poisson process with Gaussian jumps using the realised (co)variance matrices. In this case, for $i,j \in \{1,\dots,d\}$, we define the $(i,j)th$ component of the realised covariance by $$\RV(\widehat{\LL})_{N,ij} := \sum_{0 \leq k \leq N-1} |\Delta_k \widehat{L}^{(i)}| \cdot |\Delta_k \widehat{L}^{(j)}|.$$ 
The jump-filtered equivalent is defined for a unique fixed $\beta \in (0, 1/2)$ as follows
\begin{equation}
	\label{eq:rv_c}
	\RV(\widehat{\LL})^C_{N,ij} := \sum_{0 \leq k \leq N-1} |\Delta_k \widehat{L}^{(i)}|  \cdot |\Delta_k \widehat{L}^{(j)}| \cdot \indicator\{||\Delta_k \widehat{\LL}\|_2 \leq \eta \cdot \Delta_N^{\beta}\},
\end{equation}
for some $\eta > 0$. That is, increments vectors are considered to be jumps if the $L^2$-norm of the recovered increments are above $\eta\cdot \Delta_N^{\beta}$: for instance, $\eta = 1$ implies that jumps include all increment vectors such that $|\Delta_k \widehat{L}^{(j)}| > \Delta^\beta$ for some $j\in\{1,\dots,d\}$.
\begin{remark}
	Note that \eqref{eq:rv_c} relies on the assumption of a fixed volatility (i.e.\ a fixed $\eta > 0$), e.g.\ a fixed Gaussian covariance structure. An extension of this estimator to a variable volatility (similarly to \citep{barndorff2004BiPowerVariation}) would be given by
		$$ \sum_{0 \leq k \leq N-1} |\Delta_k \widehat{L}^{(i)}| \indicator\{|\Delta_k \widehat{L}^{(i)}| \leq \widehat{\sigma}_k\cdot \Delta_N^{\beta^{(i)}}\} \cdot |\Delta_k \widehat{L}^{(j)}| \indicator\{|\Delta_k \widehat{L}^{(j)}| \leq \widehat{\sigma}_k \cdot \Delta_N^{\beta^{(j)}}\},
$$
where $\widehat{\sigma}_k$ is a spot volatility estimate at time $t_k$ for $k \in \{0,\dots,N-1\}$.
\end{remark}
 \begin{remark}
 	 Alternatively to estimate $\RV(\LL)^C_N$, one could have used the bipower variation \citep{barndorff2004BiPowerVariation}
$$\RV(\widehat{\LL})^C_{N,ij} := \sum_{0 \leq k \leq N-2} \frac{\pi}{2}|\Delta_k \widehat{L}^{(i)}| \cdot |\Delta_{k+1} \widehat{L}^{(j)}|.$$ 
 \end{remark}
Finally, the realised (co)variance of the jump part defined as $$\RV(\widehat{\LL})^J_N := \RV(\widehat{\LL})_{N} - \RV(\widehat{\LL})^C_{N}.$$
Since $\Delta_k \widehat{\LL} \probconv \Delta_k \LL$ as $\Delta_N \rightarrow 0$, those quantities approximate quadratic covariation matrices as $\Delta_N \rightarrow 0$ stably in distribution \citep{Hanser2005RealizedVarianceHighFrequency, Ait2008OutSampleQV}. Hence, we have for $N$ large enough
$$\Si \approx \RV(\widehat{\LL})^C_N/T_N,$$
whilst the jump intensity $\lambda$ is estimated by likelihood maximisation given that the probability of $n$ jumps in $[0, T_N]$ is proportional to $$\lambda^n e^{-\lambda T_N}.$$
 Recall that we assume that jumps occur when $\|\Delta_k \widehat{\LL}\|_2 > \eta \cdot \Delta_N^\beta$. Finally, for Gaussian jumps, we obtain directly
 \begin{equation}
 	\label{eq:sigma-jump}
 	\Si^J \approx\lambda^{-1}\RV(\widehat{\LL})^J_N/T_N.
 \end{equation}

\paragraph{Infinite jump activity case}
We consider the generalised hyperbolic \Levy motions (GHYP) for their important distributional flexibility to accommodate for fat tails \citep{Eberlein2001Ghyp} and use the likelihood maximisation directly on the recovered \Levy increments without filtering for jumps.
\begin{remark}
	The parametric bootstrap used in the following sections is supported by the theoretical guarantees presented in \cite{stute1993paramBootstrap, abdelrazeq2018paramboostrap}: consistent parameter estimators for the driving \Levy noise distribution are enough to replicate the true dynamics of the GrOU process. Although those results are proved for distributions characterised by their first two moments \citep[Th.\ 3.1]{abdelrazeq2018paramboostrap}, we assume it also holds for GHYP distributions. 
\end{remark}

\subsection{Exploratory analysis and driving noise inference}
We clean our data set using a LOESS regression (R package \texttt{stl}) with daily seasonality and a daily trend rolling window. The output has no unit root (via the  Dicker-Fuller test with $99.9\%$ significance on each marginal) and has zero mean. 
\begin{remark}
	The wind capacity factor values are restricted to $[0,1]$ whilst the processed time series are unrestricted real values. 
\end{remark}

Using the RE-Europe 50 data set, the $\boldTheta$-GrOU parameters is estimated by 
\begin{equation}
\label{eq:theta-exploratory-analysis}
\widehat{\boldTheta}_N =(\widehat{\theta}_{1,N}, \widehat{\theta}_{2,N})^\top= (-1.549 \ (.029), \ 5.525 \ (.015))^\top	
\end{equation}
 with (infinite activity) standard deviations in parenthesis on 500 samples of 12,500 observations each through a parametric bootstrap. Note that the finite activity standard deviations are similar up to the second decimal hence omitted. Equation \eqref{eq:theta-exploratory-analysis} yields that, \emph{on average}, each node's value acts negatively on its increments, i.e.\ they have a negative momentum, whilst the neighbouring nodes act positively on the increments as they sort of aggregate the neighbours' capacity towards the node. That is, if the capacity of a node spikes, it may indicate that neighbouring nodes might increase soon after that. This said, the momentum term in absolute value is four times as big as the network term which indicates the strong autoregressive nature of wind flows. The $\boldPsi$-GrOU parameters given on the right-hand side in Figure \ref{figure:map-violin-node-level} show that Portuguese locations seem to have lower momentum than Spanish ones (top violin plot) whilst the negative Spanish network effects can be larger in absolute values than the Portuguese counter parts. This said, both the momentum and networks effect means cannot be distinguished between the two countries with a two-sample t-test with a p-value larger than $90\%$.

\subsection{Fitting the recovered increments}
\label{section:fitting-recovered-increments}
 Noise marginals are symmetric and have fatter tails than a single Gaussian distributions (see Figure \ref{figure:noise-distribution}) hence the sum of two Gaussian marginals with a standard deviation larger for jumps or GHYP distributions are adapted to this data set. 
 
 \begin{figure}[htp]
  \begin{center}
\includegraphics[width=1.0\textwidth]{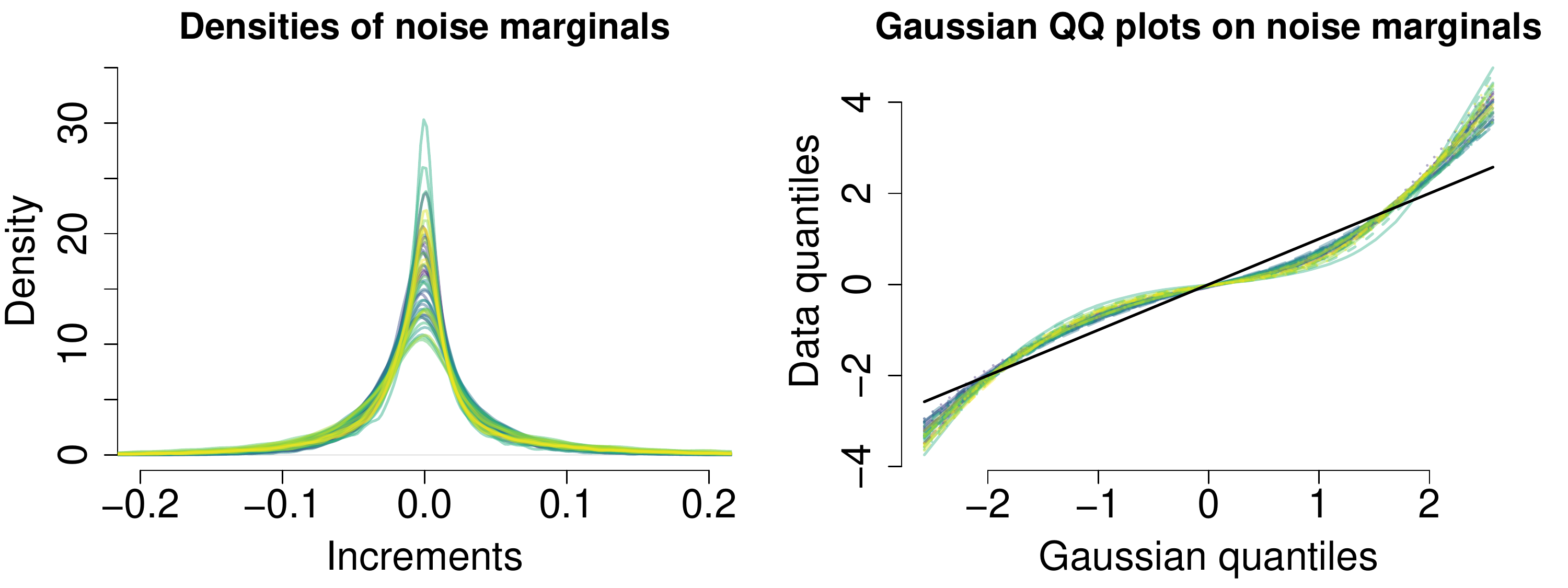}
\caption{RE-Europe 50 \Levy increments recovery. LHS: \Levy noise increment marginals. RHS: Q-Q plot of the standardised \Levy noise marginals against Gaussian quantiles.}
\label{figure:noise-distribution}
\end{center}
\end{figure}   

Regarding the finite activity modelling, we first ought to specific the jump cutoff value $\eta^*$ such that the recovery increments vectors $\Delta_k \widehat{\LL}$ are considered jumps if $\|\Delta_k \widehat{\LL} \|_2 > \eta^* \Delta_N^\beta$. We discuss the value of $\beta$ in Section \ref{section:beta} and use $\beta=0.499$. In Figure \ref{figure:jumps_stats}, we present the empirical probability of having a jump for different values of $\eta$ along with the variance captured by either the continuous part (through $\Si$) or the jump part (through $\lambda \cdot \Si^J$) under Gaussian jumps. Covariance matrices are estimated  as in Section \ref{section:levy-driving-noise} and we compute the average ratio between the diagonal elements of $\widehat{\Si}$ and $\widehat{\lambda} \cdot \widehat{\Si}^J$ and that of the empirical covariance matrix of $(\Delta_k \widehat{\LL}: k \in \{1,\dots, N\})$, respectively, for various values of $\eta$. We observe that we obtain the coverage equilibrium for $\eta = 1.7$, which corresponds to the $80\%$ quantile of the $L^2$-norm of the increments, and we set $\eta^* = 1.7$.

\begin{figure}[htp]
\begin{center}
\includegraphics[width=1.0\textwidth]{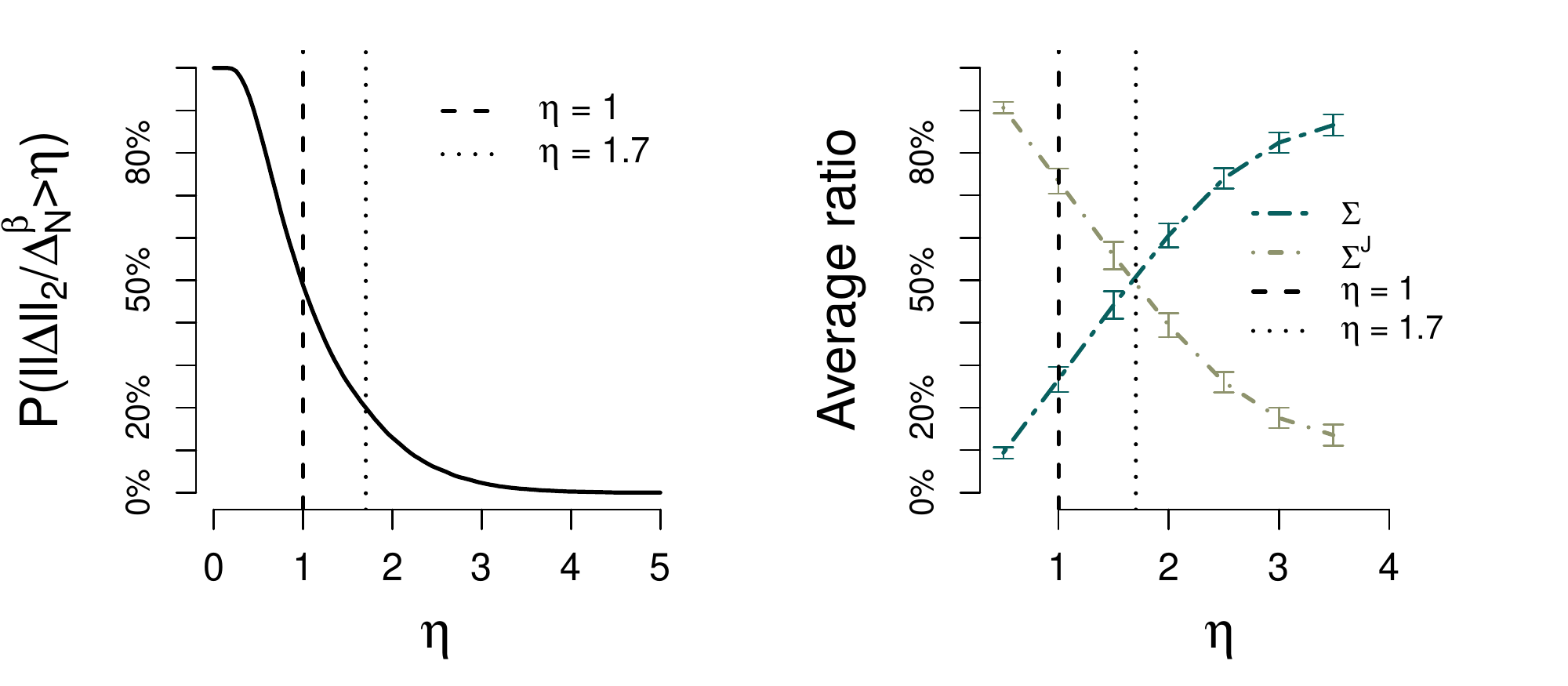}
\vspace*{-.8cm}
\caption{Choosing the jump cutoff value $\eta$: (left) empirical probability that $\|\Delta_k \widehat{\LL} \|_2 / \Delta_N^\beta > \eta$; (right) Variance coverage of the continuous part and jump parts with 95\% error bars. Lines corresponding to $\eta = 1$ (50\% quantile of jump norms) and $\eta =1.7$ (80\% quantile of jump norms) are drawn.}
\label{figure:jumps_stats}
\end{center}
\end{figure}  

We find that noise is mostly clustered by country (Figure \ref{figure:noise_corr}): we have higher correlations between nodes of the same country (top-left and bottom-right quadrants) and lower correlations for pairs of cross-country nodes. Correlations are positively skewed with minima of $-40\%$ and mostly positive (between $+20\%$ and $+60\%$) or close to zero on the second and third quadrants. The different topologies between Spain (multiple subgraphs) and Portugal (a single 24-node connected graph, see Fig. \ref{figure:network-configurations} below) may explain the higher density of the correlation matrix in the first quadrant (Portugal). 

\begin{figure}[htp]
\begin{center}
\includegraphics[width=1.0\textwidth]{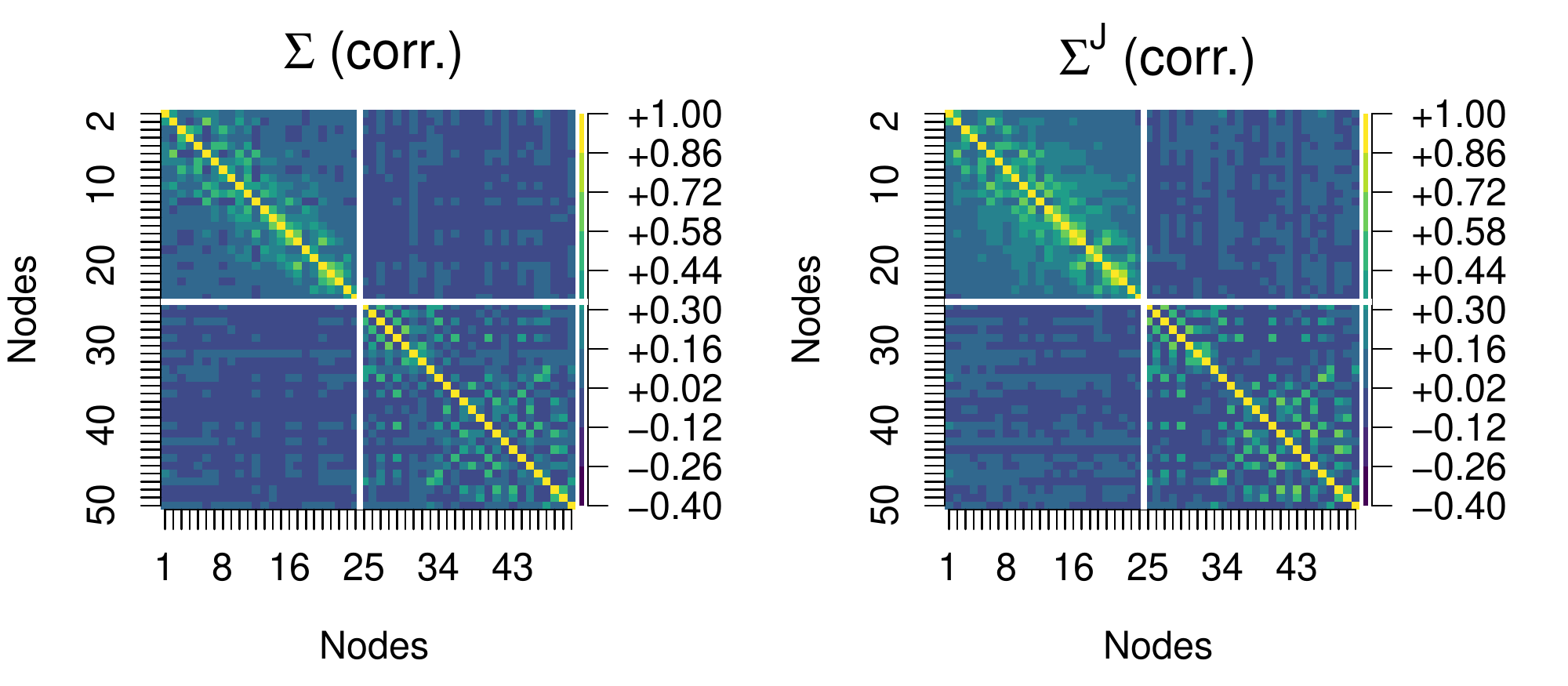}
\vspace*{-.8cm}
\caption{Estimates for $\Si$ and $\Si^J$ as correlation matrices: first and fourth quadrants are Portugal (POR) and Spain (ESP) respectively while the second and third quadrants are correlations across countries.}
\label{figure:noise_corr}
\end{center}
\end{figure}  
\newpage

As a check, we generate samples from the fitted distributions and compare their quantiles to the standardised empirical quantiles in Figure \ref{figure:qq_recovery}. It features the Q-Q plots of finite activity jump distributions of a compound Poisson with Gaussian or GHYP jumps, and the infinite activity GHYP distribution. Gaussian jumps fall short in terms of tail heaviness as shown by the $S$-shaped Q-Q plots. The GHYP distribution perform better in that respect as shown by the other two the GHYP-based samples. We choose to only use the compound Poisson with GHYP jumps in the rest of the study. We note that across the three plots, a slight bump is observed for quantiles between $0$ and $1$ which may indicate that the skewness of the fitted distributions is not sufficient to represent the data correctly. 

\begin{figure}[htp]
\begin{center}
\includegraphics[width=1.0\textwidth]{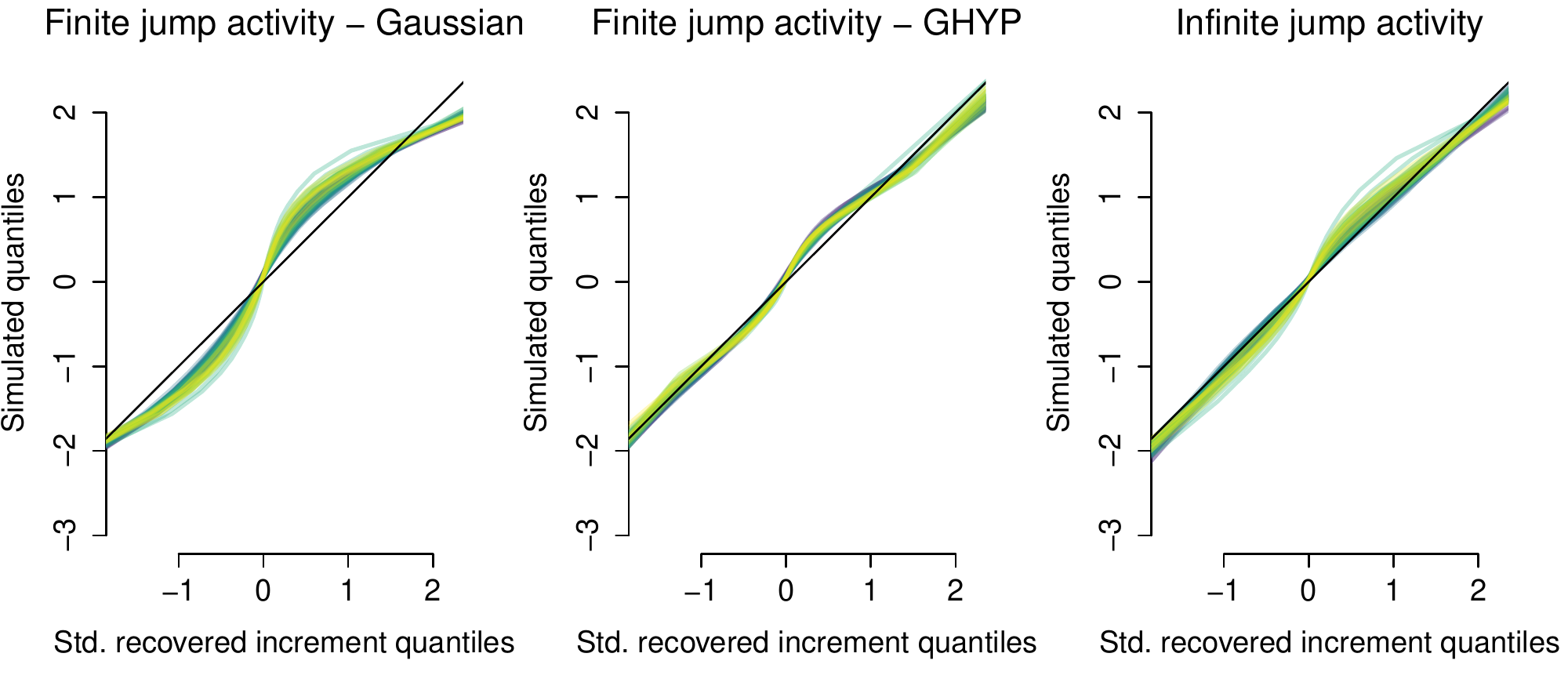}
\caption{Q-Q plot of the distributions fitted on the recovered increments: (left) Compound Poisson with Gaussian jumps, (middle) Compound Poisson with GHYP jumps and (right) GHYP-distributed increments.}
\label{figure:qq_recovery}
\end{center}
\end{figure}

\newpage
\newpage
\section{Simulation study}
\label{section:simulation-study}
In this section, we assess the reliability of the inference approach on data simulated from the model itself. By avoiding the potential model misspecification, we study the impact of the graph topology, the jump threshold parameters and we compare the MLE to the least squares (LS) estimator. 

We use the empirical parameters values  $\widehat{\boldsymbol{\theta}}_N$ found in Section \ref{section:applications} as true. We perform parametric bootstraps of our model on paths of $N=12,500$ observations each which corresponds to 50\% of the original data set length and on a uniform time grid $(k\Delta_N: \ k\in\{0,\dots,N-1\})$ where $\Delta_N = 1/12$. For the finite activity driving noise, we generate the median number of jumps across dimensions in each dimension given we work on a fixed time interval $[0,T_N]$ to ensure a reliable number of jumps are being drawn.

\subsection{Empirical validation}
We check that the $\boldTheta$-GrOU estimators converge the true value as the sample size grows and that they each follow a univariate Gaussian distribution by performing the inference on only the first $N=500; 4,500; 8,500; 12,500$ observations of each path. We restrict ourselves to the infinite jump activity case since it is the only example with a non-Gaussian driving noise; we also only work with the RE-Europe 50 network for simplicity. Gaussian Q-Q plots and histograms over the collection of $500$ independent fitted values for $N=500$ and $N=12,500$  are presented in Figure \ref{figure:estimators-gaussian}.  
\begin{figure}[htp]
  \begin{center}
\includegraphics[width=1.0\textwidth]{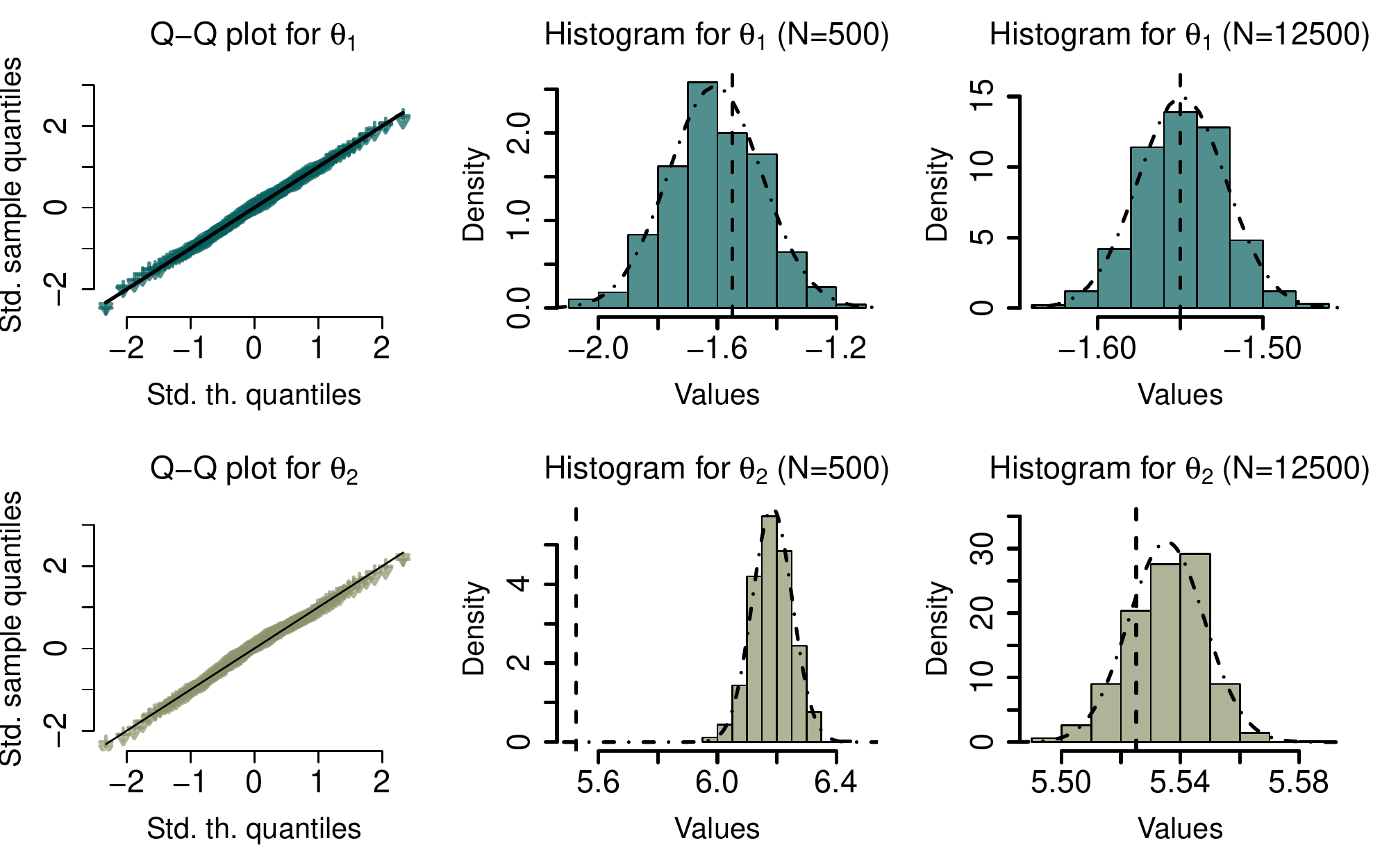}
\includegraphics[width=1.0\textwidth]{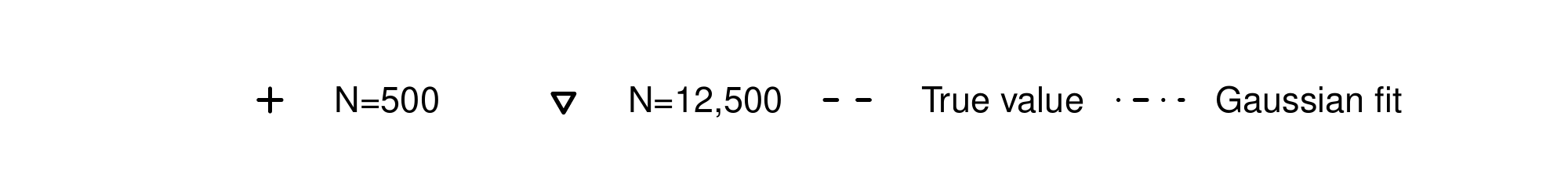}
\caption{Verification that estimators are approx.\ Gaussian for $N=500$ and $N=12,500$.}
\label{figure:estimators-gaussian}
\end{center}
\end{figure} 
The Q-Q plots for both $\theta_1$ and $\theta_2$ are straight and the Gaussian fit on the histogram has the right width. This said, with limited data $(N=500)$, the Gaussian mode is lower than the true value $\theta_1=-1.549$ and greater than the true value $\theta_2=5.525$ (by a number of standard deviations for the latter). The Gaussian fit is particularly suited for $\theta_1$ with $N=12,500$.

% TODO check this
Table \ref{table:bootstrap} features the numerical values of the experiment: as $N$ grows, the bias (and percent bias) reduces as expected although the fit with only $500$ samples is already satisfactory with $3.73\%$ and $11.94\%$ biases respectively for $\theta_1$ and $\theta_2$. Those biases boils down to a few tenths of a percent for $N=12,500$. The bias and standard deviations both shrink by a factor of 50 between the cases $N=500$ and $N=12,500$. Their rates of convergence against the sample size $N$ are estimated to be $1.23 (.106)$ and $0.55 (.006)$, respectively, using a standard linear regression on the $\log$-$\log$ scale; note that both are significant at the $5\%$ level. This corresponds to the expected rate of $1/2$ described in Section \ref{section:discretised-high-frequency-data}.
\begin{table}[htp]
\begin{tabular}{clccclccc}
\textbf{N} &  & \textbf{$\widehat{\theta}_1$} & \textbf{Bias} & \textbf{\% Bias} &  & \textbf{$\widehat{\theta}_2$} & \textbf{Bias} & \textbf{\% Bias} \\ \cline{1-1} \cline{3-5} \cline{7-9} 
500 &  & $-$1.608 (.157) & $-$0.058 & 3.73    &  & 6.185 (.067) & 0.66 & 11.94 \\
4,500 &  & $-$1.555 (.045) & $-$0.005 & 0.31    &  & 5.582 (.022) & 0.05 & 1.022 \\
8,500 &  & $-$1.551 (.032) & $-$0.001 & 0.08    &  & 5.547 (.016) & 0.02 & 0.399 \\
12,500 &  & $-$1.549 (.027) & 0.001    & $-$0.07 &  & 5.535 (.013) & 0.01 & 0.185 \\ \cline{1-1} \cline{3-5} \hline
\end{tabular}
\label{table:bootstrap}
\caption{Monte Carlo mean and standard deviations in parenthesis along with the bias and percent bias for $N=500; 4,500; 8,500; 12,500$.}
\end{table}

\subsection{Impact of the graph topology}
\label{section:impact-network-topology}
In this section, we study the impact of the adjacency matrix, or graph topology, $\A$. From Equation \eqref{eq:sde}, it seems reasonable to postulate that this said topology impacts the inference performance of the estimators---especially for the off-diagonal parameter $\theta_1$. The data corroborates this effect which we interpret as a consequence of the value dissipation across a higher number of neighbours; formally, its expression features the quadratic impact of synchronised peaks between two variables.

Figure \ref{figure:network-configurations} presents three artificial network configuration by increasing order of connectivity (here shown by some key statistics of the node degrees in parenthesis):  \texttt{Polymer} (50 nodes in a straight line, $\text{mean} = 2.96, \text{median} = 3, \min = 1,\ \max = 2$), \texttt{Lattice} (a grid of $9 \times 9$ nodes where nodes are connected to their closest neighbours and one additional node linked to only one node of the grid, $\text{mean} = 3.36, \text{median} = 4, \min = 2,\ \max = 4$), \texttt{Complete} (a 50-node complete graph, $\text{mean} = 50, \text{median} = 50, \min = 50,\ \max = 50$), along with the \texttt{RE-Europe 50} network configuration ($\text{mean} = 1.12, \text{median} = 1, \min = 0,\ \max = 4$).

\begin{figure}[htp]
  \begin{center}
\includegraphics[width=1.0\textwidth]{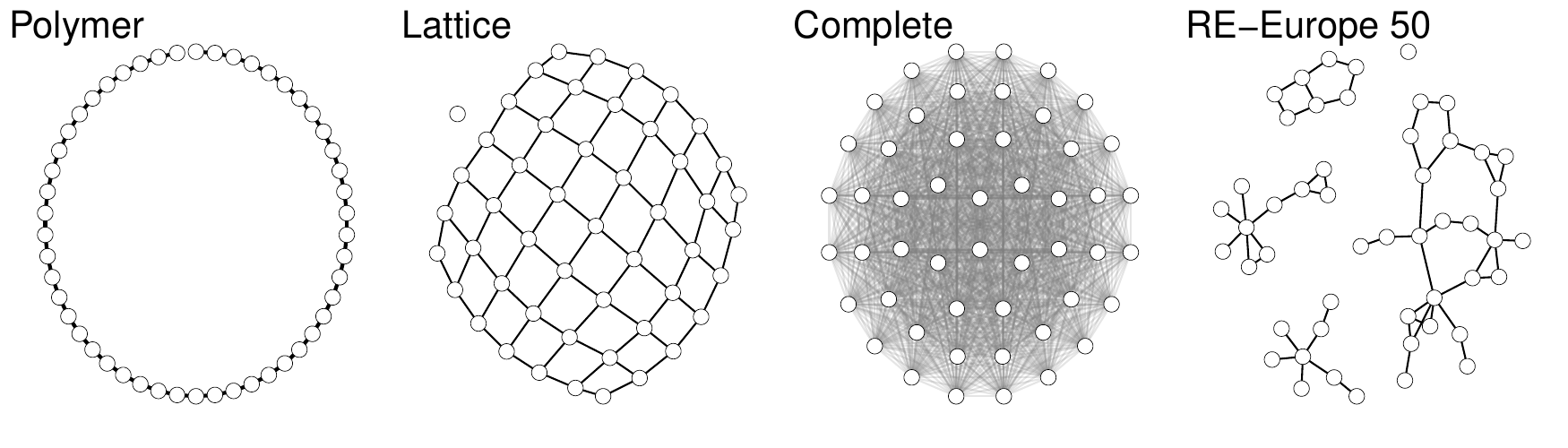}
\caption{Different network configurations used in the simulation study. Note that the Polymer configuration is not connected at the top of the circle.}
\label{figure:network-configurations}
\end{center}
\end{figure}

In Figure \ref{figure:violin-layout-configuration}, violin plots of the inference on 100 paths with 12,500 samples each are represented for those four configurations.

\begin{figure}[htp]
  \begin{center}
\includegraphics[width=1.0\textwidth]{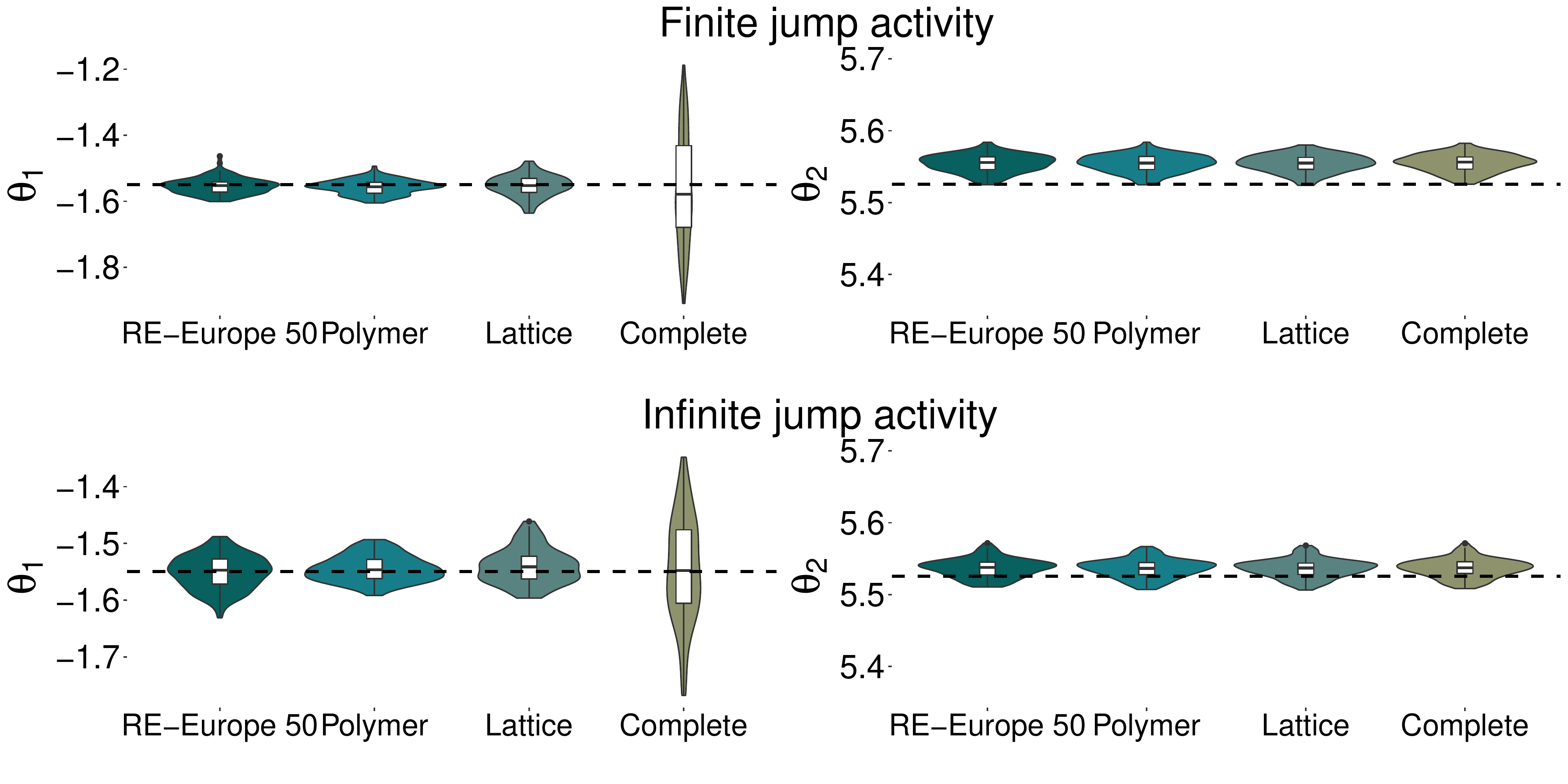}
\caption{Violon plots of the $\boldTheta$-GrOU estimation on 100 bootstrapped paths for both the finite and infinite jump activities across four graph configurations.}
\label{figure:violin-layout-configuration}
\end{center}
\end{figure}  

The MLE estimator performs well across all graph configurations except for the off-diagonal parameter $\theta_1$ for the complete graph which is substantially noisier. Quantitatively, the row-normalised parameter $\theta_1/n_i$ is low compared to $\theta_2$ where $n_i = 50$ which induces some numerical instability. Qualitatively,  the complete graph dissipates the signal at a given node to all its neighbours whilst receiving a small portion from each of its neighbours' signal which causes some instability.
 The performance for $\theta_2$ is similar across the examples since each node is equipped with a self-loop in all graph configurations.

\subsection{Tuning the jump threshold exponent $\beta$}
\label{section:beta}
An important set of parameters are the jump thresholds $(\beta^{(i)}:\ i \in \{1,\dots,d\}) \subset (0,1/2)$ which we assume to be equal across the marginals given they are centred and with similar standard deviations ($\approx 5\%$). A jump threshold too large, i.e.\ $\beta \rightarrow 0$, (resp.\ too small, i.e.\ $\beta\rightarrow 1/2$) would lead to underestimate (resp.\ overestimate) the frequency of jumps and overestimate (resp.\ underestimate) their amplitudes (Remark \ref{remark:jump-threshold-delta_n}).

Parametrically bootstrapping 500 paths using the estimated infinite activity driving noise distribution led to the bias study presented in Figure \ref{figure:noise-marginals}. As $\beta \rightarrow 1/2$, the bias for the off-diagonal term $\theta_1$ decreases in all four cases with a maximum change of $0.005$ in absolute value. It becomes negligible for \texttt{RE-Europe 50} and \texttt{Complete} topologies whilst it gets even more negative for the \texttt{Complete} and \texttt{Lattice} topologies. However, the changes in bias for $\theta_1$ are substantially smaller than the biases for $\theta_2$: the biases change by a higher order of magnitude from approximately $-0.025$ to $-0.010$ for all four topologies.

This contrasts with the increased variability in the estimation for $\theta_1$ compared to $\theta_2$ depending on the topology (Section \ref{section:impact-network-topology}). Therefore, we take a value of $\beta$ close to $1/2$ ($\beta = 0.4999$) for two reasons: (a) to limit the estimation bias for both $\theta_1$ and $\theta_2$ with the EU-Europe 50 topology (b) since $\widehat{\theta}_2$ benefits the most from this change in $\beta$ while $\widehat{\theta}_1$ is already more volatile. This is in accordance with the literature \citep[Section 4, p.\ 1749]{bollerslev2011estimation}.
\begin{figure}[htp]
  \begin{center}
\includegraphics[width=1.0\textwidth]{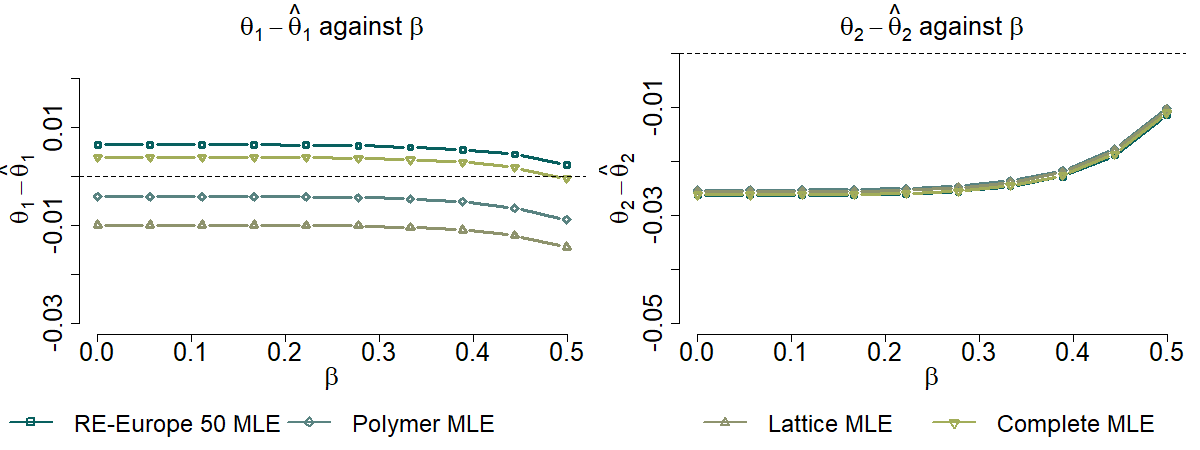}
\caption{Estimation bias of the $\boldTheta$-GrOU estimator against $\beta \in (0,1/2)$ for different graph topologies on 50 paths. Standard deviations are not included for clarity as they are smaller than $.0001$.}
\label{figure:noise-marginals}
\end{center}
\end{figure}

% TODO add 0.4999 add another  

\subsection{Comparing with the least squares estimator}
\label{section:comparing-least-squares}
In this section, we compare the high-frequency \levy-driven OU MLE estimator to the LS estimator as studied in \cite{fasen2013} which estimates $\exp(-\Q\Delta_N)$ and we extend this comparison to the multivariate context. We report the Relative $L^2$-Error Metric (REM) with respect to $Q(\boldTheta)$ which is a map from $\MdR$ onto $[0,\infty)$ defined by
$$\rmX \mapsto \|\exp(-\Q(\boldTheta)\Delta_N)\|^{-1}\cdot \|\exp(-\rmX\Delta_N) - \exp(-\Q(\boldTheta)\Delta_N)\|,$$
where $\boldTheta$ is given in Eq.\ \eqref{eq:theta-exploratory-analysis} as the true parameters for the paths simulation. The smaller the REM, the better the inference strategy is at recovery the true parameters whilst allowing to compare both the MLE and the LS estimators.

\subsubsection{Noise amplitude robustness}
To measure the robustness of estimators against the noise amplitude, we multiply our fitted driving noise distribution by a \emph{noise multiplier} scalar $\sigma \in \{0.5, 1, 5, 10, 100, 1000\}$: we generate our paths with the noise increments $(\sigma \Delta_k \LL: \ k \in\{1,\dots,N\})$. This puts the robustness of the $\boldTheta$-GrOU MLE to the test as the noise amplitude gets larger across graph topologies with both finite and infinity jumps activities. In Figure \ref{figure:fasen-sigma}, we vary the noise multiplier parameter $\sigma \in \{0.5, 1, 5, 10, 100, 1000\}$ to generate paths driven by a noise with either finite (\textnormal{MLE (fin)}) or infinite (\textnormal{MLE (inf)}) jump activity.

\begin{figure}[htp]
  \begin{center}
\includegraphics[width=1.0\textwidth]{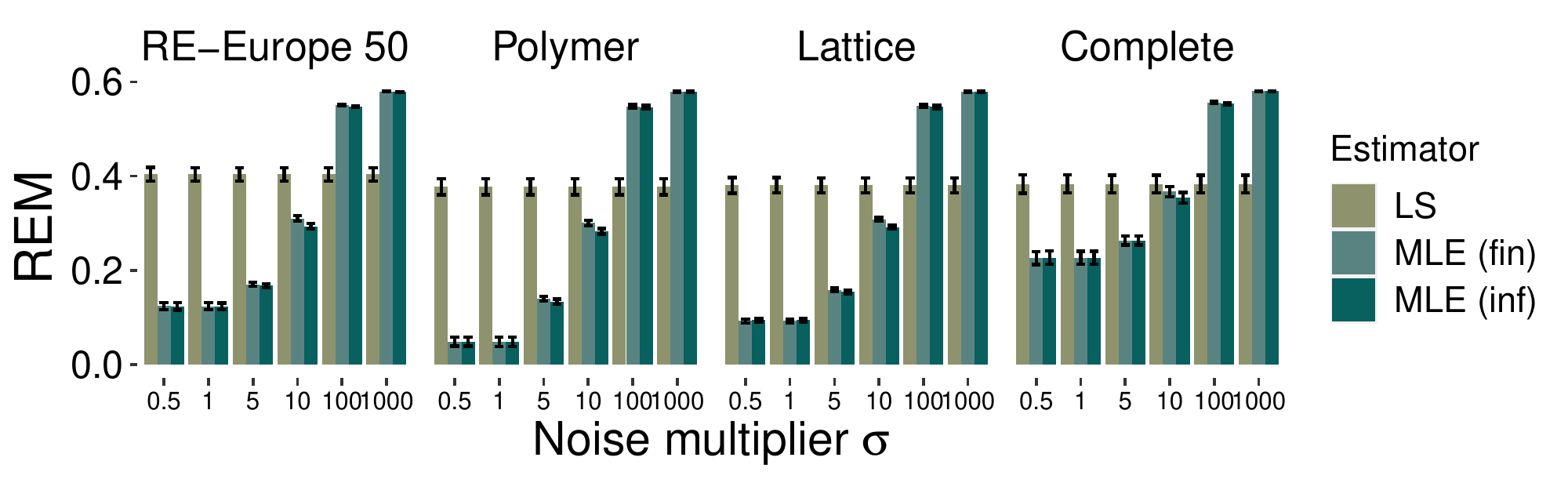}
\caption{Relative $L^2$-Error Metric (REM) comparing the $\boldTheta$-GrOU MLE estimator and the \textnormal{LS} estimator to the true parameter values (95\% standard deviations in black).}
\label{figure:fasen-sigma}
\end{center}
\end{figure} 
The LS estimator has a constant performance across the jump activities and graph topologies with the REM ranging from $0.37$ to $0.41$ in the finite and infinite cases. Thus, we only plot the finite cases. Similarly, the MLE performs equally under finite and infinite activities. 

As shown \cite{mai2014efficient} in the univariate case, the MLE estimator is superior to the LS estimator with the REM ranging from $0.05$ to $0.20$ for $\sigma \in \{0.5,1,5,10\}$. As $\sigma$ gets larger, the MLE performance gets poorer with the REM above $0.5$ for $\sigma \in \{100,1000\}$ in both the finite and infinite cases whilst the LS estimator's performance remains constant. This is due to the jump filtering term that absorbs most of the increments as the amplitude grows. Also, a finite-sample bias is present both estimators converge anyway as the REM performance is surprisingly consistent across paths for all estimators with 95\% errors equal to a thousandth of a REM unit, even for $\sigma = 1,000$.

For a given estimator, the behaviour is similar across different topologies. The MLE performs better for sparser graphs: for $\sigma = 1$, the REM is ten times as small for the MLE on \texttt{Polymer} compared the LS whilst it is only twice as small on \texttt{Complete}. 

The fact the performance of the LS estimator is not related to the underlying graph structure is difficult to interpret. We conjecture this is a side-effect of the row-normalisation applied to the dynamics matrix. However, the constant performance even for large noise amplitudes was expected as it converges even with infinite variance \citep{fasen2013}.

\subsubsection{Robustness to changes in mesh size}
\label{section:robustness-mesh-size}
	The availability of data sets quantifying the same system but different mesh sizes are common (e.g.\ meteorological and financial data) and being able to compare the GrOU parameters across data sets with different mesh sizes becomes crucial.	
	
Recall that the MLE converges as $T_N \rightarrow \infty$ and $\Delta_N \rightarrow 0$ as detailed in Assumption \ref{assumption:high-frequency-asymptotics}. The LS estimator encodes the mesh size in the estimate of $\exp(-\rmQ \Delta_N)$ as opposed to the MLE. However, for a fixed sample size $N=25,000$, we compare the estimators's performances by varying $\Delta_N$ from $0.001$ to $0.200$ whilst the true parameter values were generated with $\Delta_N=1/12$ as pictured in Figure \ref{figure:horizon}.

We observe that the parameter values are linked to a particular mesh size, may it be the MLE or LS estimators as hinted by the drop in REM around the original mesh size of $\Delta_N = 1/12$. With a larger mesh size, we see that the REM explodes exponentially (due to the log-log scale) whilst a smaller mesh size leads a larger REM which remains controlled and never exceeds $0.6$. 

We conclude that we ought to interpolate data to upsample the lower frequency data set rather than downsample the higher frequency one to make two sets of parameters comparable.

\begin{figure}[htp]
  \begin{center}
\includegraphics[width=1.0\textwidth]{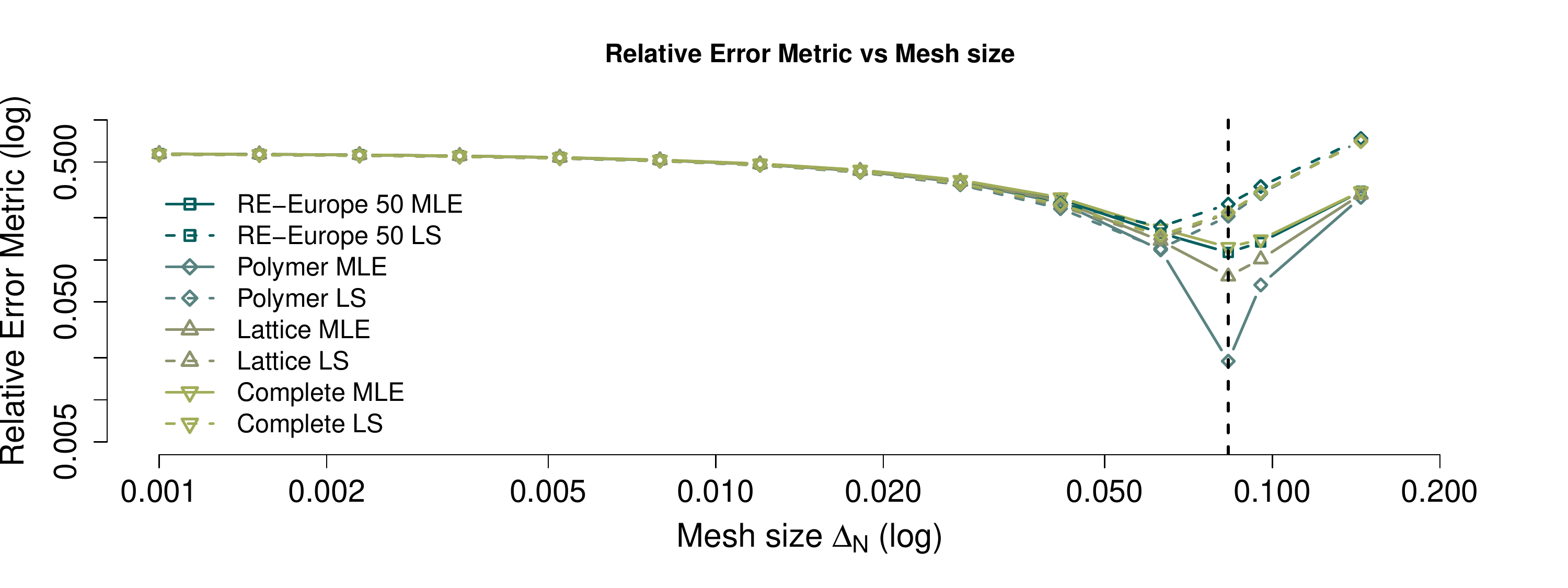}
\caption{Relative $L^2$-Error Metric comparing the $\boldTheta$-GrOU MLE and LS estimators against the true parameter values $\boldTheta =(-1.549,\ 5.525)^\top$ by varying $\Delta_N \in [0.001, 0.200]$ on a $\log$-$\log$ scale. The vertical dashed line indicated the true mesh size.} 
\label{figure:horizon}
\end{center}
\end{figure}

\newpage

\section{Conclusion}
In this article, we considered the discretely-observed Graph Ornstein-Uhlenbeck process and carried out the translations of the continuous-time estimators and asymptotic results given in \cite{courgeau2020likelihood} to a non-uniform high-frequency discrete time framework. This was done by adapting the (high-frequency) double asymptotic assumptions from \cite{mai2014efficient} to the multivariate context. Jump-filtered and discretised maximum likelihood estimators are shown to converge stably to the same distribution as in the continuously-observed case under a set of standard assumptions. This yields two different scales for interpretation: on a $d$-node graph, two-parameter $\boldTheta$-GrOU or $d$-parameter $\boldPsi$-GrOU parametrisations are available; the former detailing the average behaviour of a node whilst the latter could be useful for prediction purposes. In addition, the non-uniformity of the observation grid extends the range of applicable data sets: e.g.\ event-based data sets closer to real-time data aggregation \citep{Simonov2017EventBased, hol2018estimation}. In addition to providing the consistency of the estimators, the stable convergence form a theoretical framework to perform edge pruning in conjunction to the GrOU inference \citep{matulewicz2017statistical} if the network structure is to be stochastic or time-dependent \cite{zaccarin2010expRandomGraph}. Alternatively, if the structure is deterministic but unknown, the Adaptive Lasso scheme from \cite{courgeau2020likelihood, Gaiffas2019SparseOUProcess} extended to high-frequency observations is shown to be asymptotically normal (stably) and consistent in variable selection.

Future research comprises the assessment of the forecasting performance of such a model on high-frequency data sources (e.g.\ on tick-by-tick financial time series) for a large collection of systems or financial assets \citep{hol2018estimation}. Note that including the impact of the neighbours' neighbours, or more generally $n$-th degree neighbours as in \cite{Knight2016ModellingDA}, could prove useful in improving the model's performance. Another direction would be the extension of the discrete-time asymptotic convergence results to \levy-driven \emph{spatio-temporal} OU processes as defined in \cite{NguyenMichele2016SpatioTemporalOU} and comparing them with the coupled sparse inference on the adjacency matrix along with the GrOU inference as mentioned above.  Additionally, extending the results for the studied OU-type structure for graphs to the more general continuous time autoregressive moving average (CARMA) processes \citep{Marquardt2007MultivariateCarma, brockwell2013carma-recovery} under a high-frequency sampling scheme is a promising extension. Alternatively, the application of stochastic volatility on GrOU-type process  \cite[Section 5]{courgeau2020likelihood} would alleviate the fixed structure of the Gaussian component and introduce empirical properties such as volatility clustering. Finally, the stable convergence results may allow to extend the GrOU process to stochastic graphs, potentially by borrowing ideas from the random graph modelling literature \citep{zaccarin2010expRandomGraph, robins2011exponential}.

\appendix

\section{\levy-\Ito decomposition}
\label{appendix:levy-ito-decomposition}
In this section, we provide standard results and properties of the \levy-\Ito decomposition for \Levy processes.  We say a truncation function is any $\R^d$-valued nonnegative function. The \Levy process $\LL$ has the following characteristic function for $\LL$ at time $t \geq 0$:
\begin{align*}
	& \E\left(\exp\left(i \boldsymbol{u}^\top \LL_t \right)\right)\\
	&= \exp\left\{
t\left( i\boldsymbol{u}^\top \boldB - \frac{1}{2}\boldsymbol{u}^\top \Si \boldsymbol{u}+ \int_{\R^d\symbol{92}\{\boldsymbol{0}_d\}}\left( \exp\left[i \boldsymbol{u}^\top \boldsymbol{z} \right] - 1 - i \boldsymbol{u}^\top \boldsymbol{z} \tau(\boldsymbol{z})\right)d\nu(\boldsymbol{z})\right)\right\},
\end{align*}
 where $\boldsymbol{u},\ \boldB \in \R^d$, $\Si \in \Ss^{++}$ and $\nu$ is a \Levy measure on $\R^d$ satisfying $\int_{\R^d\symbol{92}\{\boldsymbol{0}\}} (1 \wedge \|\boldsymbol{z}\|^2)\nu(d\boldsymbol{z}) < \infty$. 
The decomposition formulated in \cite{bretagnolle2006ecole} unfolds as follows:
\begin{theorem}[adapted from Theorem 3.12, \cite{bretagnolle2006ecole}]
\label{th:levy-ito-decomposition}
Let $\boldsymbol{x} \mapsto \tau(\boldsymbol{x})$ be a truncation function (e.g. $\tau(\boldsymbol{z}) := \indicator_{\{\boldsymbol{x}\in\R^d: \|\boldsymbol{x}\| \leq 1\}}(\boldsymbol{z})$ where $\boldsymbol{z} \in \R^d$). The $\LL$ be a $d$-dimensional \Levy process with characteristic triplet $(\boldB,\Si, \nu)$ with respect to the truncation function $\tau$, then there exist
\begin{itemize}
    \item A centred Gaussian \Levy process $(\WW_t,\ t \geq 0)$ with respect to $(\mathcal{F}_t,\ t \geq 0)$ with covariance matrix $\Si$ and almost-surely continuous paths;
    \item A family $(N_t(d\boldsymbol{z}),\ t \geq 0)$ of Poisson processes, independent of $(\WW_t,\ t \geq 0)$, with $N_t(A)$ independent of $N_t(B)$ for any $t \geq 0$ if $A \cap B = \emptyset$, and with $\nu(d\boldsymbol{z}) = \E\left(N_1(d\boldsymbol{z})\right)$;
\end{itemize}
such that we have uniquely for any $t \geq 0$
$$\LL_t \overset{d}{=} t \boldB + \WW_t  + \JJ^{1}_t + \JJ^{2}_t,$$
where
\begin{align*}
	\boldB &:= \E\left(\LL_1 - \int_{\tau(\boldsymbol{z})=0}\boldsymbol{z} N_1(d\boldsymbol{z})\right)\\
		\JJ^{1}_t &:= \int_{\tau(\boldsymbol{z})=0}z N_t(d\boldsymbol{z})\\
		\JJ^{2}_t  &:= \int_{\tau(\boldsymbol{z})=1}\boldsymbol{z}\big(N_t(d\boldsymbol{z})-t\nu(d\boldsymbol{z})\big),\\
\end{align*}
and we define the process $\JJ_t := \JJ^{1}_t + \JJ^{2}_t$.
%
%\begin{align*}
%	\LL_t &= t \underbrace{\E\left\{\LL_1 - \int_{\tau(\boldsymbol{z})=0}\boldsymbol{z} N_1(d\boldsymbol{z})\right\}}_{= \boldB} + \WW_t  \\
%	&+ \underbrace{\E\left\{\int_{\tau(\boldsymbol{z})=0}z N_t(d\boldsymbol{z})\right\}}_{=: \JJ^{1}_t} + \underbrace{\E\left\{\int_{\tau(\boldsymbol{z})>0}\boldsymbol{z}(N_t(d\boldsymbol{z})-t\nu(d\boldsymbol{z}))\right\}}_{\JJ^2_t},
%\end{align*}
%that is, noting that $\JJ_t := \JJ^1_t + \JJ^2_t$ is a pure-jump \Levy process independent of $\WW_t$, we have: 
%\begin{equation}
%    \label{eq:levy-decomposition-full-pure-jump}
%\LL_t \overset{d}{=} tb + \WW_t + \JJ_t = tb + \WW_t + \JJ^1_t + \JJ^2_t.
%\end{equation}
In addition, $\nu(d\boldsymbol{z})$ is a positive measure on $\R^d\symbol{92} \{\boldsymbol{0}_{d}\}$ with $\int \min(\|\boldsymbol{z}\|^2,1) \nu (d\boldsymbol{z}) < \infty$.
\end{theorem}
\begin{remark}
	In this article, we use $\tau(\boldsymbol{z}) := \mathbb{I}_{\{\|\boldsymbol{x}\| \leq 1\}}(\boldsymbol{z})$. That is, $\tau$ is positive when $\boldsymbol{z} \leq 1$ allowing to capture small jumps in $\JJ^2$ whilst $\tau$ is zero for larger jumps (as in the definition of $\boldB$ and $\JJ^1$).
\end{remark}
Note that $N_t(d\boldsymbol{z}) = \int_0^t N(ds,d\boldsymbol{z})$ where $N(ds,d\boldsymbol{z})$ is a Poisson random measure on $[0,\infty) \times \R^d$. Under Assumption \ref{assumption:mle-convergence}, one has $\LL_t = \WW_t + \JJ_t$. Also, the stochastic integral of $\JJ^1$  is in the sense of $L^0$ and the second one, of $\JJ^2_t$, in the $L^2$ sense. As noted in \cite{ContRama2004levyitodecomp}, $(\WW_t)$, $(\JJ^1_t)$ and $(\JJ^2_t)$ are independent and the convergence in the last term is almost sure and uniform in $t$ on compact sets. 

%
%\section{Practical reformulation of the estimators}
%\label{proof:practical-reformulation}
%\label{proof:definition:discretised-estimators}
%\begin{proof}[Comments on Definition \ref{definition:discretised-estimators}]
%
%	The formulations of both $\widehat{\boldTheta}_t$ and $\widehat{\boldPsi}_t$ assume the knowledge of the diffusion matrix $\Si \in \Ss^{++}$, that is a positive definite matrix, which is generally not available and we reformulate both estimators without $\Si$. Remark that
%$$\int_0^t\langle \rmQ \YY_s , d\YY^{c}_s \rangle_\Si  = \vectorise(\Q)^\top \cdot (\Id \otimes \Si^{-1}) \cdot \int_0^t \YY_{s} \otimes d\YY^{c}_s,$$
%and
%$$\int_0^t\langle \rmQ \YY_s , \Q\YY \rangle_\Si ds  = \vectorise(\Q)^\top \cdot (\Id \otimes \Si^{-1}) \cdot (\boldsymbol{K}_t \otimes \Id) \cdot \vectorise(\Q),$$
%as given in the proof of Proposition 3.3.2, \cite{courgeau2020likelihood}.
%Hence, we can rewrite 
%$$\begin{cases}
%	\At^\Si_t &= \ -(\Id \otimes \Si^{-1})\cdot(\int_0^t \YY_s \otimes d\YY_s^c)\\ [\At^\Si]_t &= \ (\Id \otimes \Si^{-1})\cdot (\boldsymbol{K}_t \otimes \Id)
%\end{cases},$$
%such that the estimator is given by
%\begin{equation*}
%	\label{eq:reformulation-bold-psi-without-sigma}
%	\widehat{\boldsymbol{\psi}}_t := \boldsymbol{S}_t^{-1} \cdot \At_t, \quad \text{where} \quad
%\begin{cases}
%\At_t &:= \: - \int_0^t \YY_s \otimes d\YY_s^c\\
%\boldsymbol{S}_t &:= \: {\Kt}_t \otimes \Id
%\end{cases}.
%\end{equation*}
%%and we recall that ${\Kt}_t:=\int_0^t\YY_s \YY_s^\top ds$. 
%\end{proof}

\section{Proof for the discretised unfiltered estimator}
\label{proof:lemma:cv-estimator-continuous-component}
\subsection{Proof of Lemma \ref{lemma:cv-estimator-continuous-component}}
\begin{proof}[Proof of Lemma \ref{lemma:cv-estimator-continuous-component}]
From L\'evy-\ito's decomposition, there exists a d-dimensional Brownian Motion $\WW_t$ with respect to $\proba_{\YY}$ and covariance matrix $\Si$ and a pure jump L\'evy process $\mathbb{J}_t$ such that $\LL_t = \WW_t + \mathbb{J}_t$ (without drift by assumption). The continuous $\proba_{0}$-martingale part is then expressed as $\YY^{c}_t = \YY^{c}_0  + \WW_t - \int_0^t \rmQ \YY_sds$ hence $\Delta_k Y^{(i),c} = \Delta_k W^{(i)} - \int_{t_k}^{t_{k+1}}(\Q\YY)^{(i)}_sds = \Delta_k W^{(i)} - \sum_{l=1}^d\int_{t_k}^{t_{k+1}}\Q_{il}Y^{(l)}_sds$. Recall that $\boldsymbol{\psi} = \text{vec}(\Q)$ and that $$\boldsymbol{\overline{\psi}}_{N} = \Sn_N^{-1} \cdot\left(\sum_{k=0}^{N-1} \YY_k \otimes \int_{t_k}^{t_{k+1}}\Q\YY_sds\right).$$
Remark that the $d(i-1)+j$-th element of $\YY_k \otimes \int_{t_k}^{t_{k+1}}\Q\YY_sds$ is given by $Y^{(i)}_k\int_{t_k}^{t_{k+1}}\sum_{l=1}Q_{jl}Y^{(l)}_sds = \sum_{l=1}Q_{jl}\int_{t_k}^{t_{k+1}}Y^{(i)}_kY^{(l)}_sds$. The same element in $\left(\int_{t_k}^{t_{k+1}}\YY_k  \YY_{s}^\top ds \otimes \Id \right) \boldsymbol{\psi}$ is 
$$\sum_{l=1}^{d}\int_{t_k}^{t_{k+1}}Y^{(i)}_s  Y^{(l)}_sdsQ_{jl} = \sum_{l=1}^{d}Q_{jl}\int_{t_k}^{t_{k+1}}Y^{(i)}_k  Y^{(l)}_sds.$$
Since this is true for any $i,j\in\{1,\dots,d\}$, we conclude that
\begin{equation}
	\label{eq:trick-An}
	\YY_k \otimes \int_{t_k}^{t_{k+1}}\Q\YY_sds = \left(\int_{t_k}^{t_{k+1}}\YY_k  \YY^\top_{s}ds \otimes \Id \right) \cdot \boldsymbol{\psi}.
\end{equation}
Using Equation \eqref{eq:trick-An}, we obtain:
\begin{align*}
    &T_N^{1/2}(\boldsymbol{\overline{\psi}}_{N}-\boldsymbol{\psi})\\
    &=T_N^{1/2}\Sn_N^{-1} \bigg(\sum_{k=0}^{N-1} \YY_k \otimes \int_{t_k}^{t_{k+1}}\Q\YY_sds - \Sn_N \boldsymbol{\psi} \bigg) - T_N^{1/2}\Sn_N^{-1}\sum_{k=0}^{N-1}\YY_k \otimes \Delta_k\WW \\
    &= T_N^{1/2}\Sn_N^{-1} \bigg[\sum_{k=0}^{N-1} \left(\int_{t_k}^{t_{k+1}}\YY_k  \YY^\top_{s}ds \otimes \Id \right) \boldsymbol{\psi} - \Sn_N \boldsymbol{\psi} \bigg] - T_N^{1/2}\Sn_N^{-1}\sum_{k=0}^{N-1}\YY_k \otimes \Delta_k\WW\\
    &= T_N^{1/2}\Sn_N^{-1} \bigg(\int_{0}^{T_N}\YY_{\lfloor s \rfloor_N} \YY_s^\top ds \otimes \Id - \Sn_N  \bigg)\boldsymbol{\psi} - T_N^{1/2}\Sn_N^{-1}\sum_{k=0}^{N-1}\YY_k \otimes \Delta_k\WW,
    \end{align*} 
    where $\lfloor s \rfloor_N := t_k$ such that $t_k \leq s < t_{k+1}$ for $ k \in \{1,\dots, N\}$.
    Therefore, we write respectively the terms above as
    \begin{align*}
    T_N^{1/2}(\boldsymbol{\overline{\psi}}_{N}-\boldsymbol{\psi})&= \ZZ^1_N + \ZZ^2_N.
\end{align*}
As in the univariate case \citep[Section 3, p.\ 929]{mai2014efficient}, we will prove that $\ZZ^1_N \probconv 0$ and $\ZZ^2_N \xrightarrow{\ \mathcal{D} \ } \mathcal{N}\left(\boldsymbol{0}_{d^2}, 
\E\left(\YY_\infty \YY^\top_\infty\right)^{-1} \otimes \Si \right)$ as $N \rightarrow \infty$. First, we can rewrite $\ZZ^1_N$ as follows
$$\ZZ^1_N = T_N \Sn_N^{-1} \cdot T_N^{-1/2}\left(\widehat{\boldsymbol{K}}_N\otimes \Id - \Sn_N \right) \boldsymbol{\psi},$$
where $\widehat{\boldsymbol{K}}_N := \int_{0}^{T_N}\YY_{\lfloor s \rfloor_N} \YY_s^\top ds$.
Recall that $\Sn_N = \Kn_N \otimes \Id$ where we remark that $\Kn_N = \int_{0}^{T_N}\YY_{\lfloor s \rfloor_N} \YY_{\lfloor s \rfloor_N}^\top ds$. %The component $(\Kn_{N})_{ij}$ is defined as $\sum_{k=0}^{N-1} Y^{(i)}_kY^{(j)}_k(t_{k+1}-t_k)$.
%Using the same trick as before, we write
%$$(\Kn_{N})_{ij} = \sum_{k=0}^{N-1}\int_{t_k}^{t_{k+1}}Y^{(i)}_{\lfloor s \rfloor_N}Y^{(j)}_{\lfloor s \rfloor_N}ds = \int_{0}^{T_N}Y^{(i)}_{\lfloor s \rfloor_N}Y^{(j)}_{\lfloor s \rfloor_N}ds.$$
Taking the $L^1$ difference between $(\Kn_N)_{ij}$ and $(\widehat{\boldsymbol{K}}_N)_{ij}$ leads to
\begin{align*}
	T_N^{-1/2}& \E\left(\left|(\Kn_{N})_{ij} -(\widehat{\boldsymbol{K}}_N)_{ij}\right|\right)\\
&= T_N^{-1/2}\E\left[\int_{0}^{T_N}\left|Y^{(i)}_{\lfloor s \rfloor_N}Y^{(j)}_{\lfloor s \rfloor_N} - Y^{(i)}_{\lfloor s \rfloor_N}Y^{(j)}_s\right|ds\right]     \\
&=T_N^{-1/2}\int_{0}^{T_N}\E\left(\left|Y^{(i)}_{\lfloor s \rfloor_N}\right|\cdot\left|Y^{(j)}_{\lfloor s \rfloor_N} - Y^{(j)}_s\right|\right)  ds, \quad \text{by Fubini's theorem,}   \\
&\leq T_N^{-1/2}\int_{0}^{T_N}\E\left(\left|Y^{(i)}_{\lfloor s \rfloor_N}\right|^2\right)^{1/2}\cdot\E\left(\left|Y^{(j)}_{\lfloor s \rfloor_N} - Y^{(j)}_s\right|^2\right)^{1/2}  ds \\
&\leq T_N^{1/2}\E\left(\left|Y^{(i)}_{0}\right|^2\right)^{1/2}\sup_{u\in[0,\Delta_N]}\E\left(\left|Y^{(j)}_{u} - Y^{(j)}_{0}\right|^2\right)^{1/2}, \quad \text{by stationarity.}\\
&= O(T_N^{1/2}\Delta_N^{1/2}), \quad \text{by the \levy-\Ito decomposition.}
\end{align*}

We prove that $\Kn_{N}$ approximates its continuously-observed counterpart $(\boldsymbol{K}_{T_N})$ in the $L^1$ sense as follows:
\begin{align*}
T_N^{-1}& \E\left[|(\Kn_{N})_{ij} -({\boldsymbol{K}_{T_N}})_{ij}|\right]\\
&= T_N^{-1}\E\left[\left|\int_{0}^{T_N}\left(Y^{(i)}_{\lfloor s \rfloor_N}Y^{(j)}_{\lfloor s \rfloor_N} - Y^{(i)}_sY^{(j)}_s\right)ds\right|\right]     \\
&\leq T_N^{-1}\E\left(\int_{0}^{T_N}\left|Y^{(i)}_{\lfloor s \rfloor_N}Y^{(j)}_{\lfloor s \rfloor_N} - Y^{(i)}_sY^{(j)}_s\right|ds\right)     \\
&\leq T_N^{-1}\int_{0}^{T_N}\E\left(\left|Y^{(i)}_{\lfloor s \rfloor_N}Y^{(j)}_{\lfloor s \rfloor_N} - Y^{(i)}_sY^{(j)}_s\right|\right) ds,    \quad \text{by Fubini's theorem,}\\
&\leq \sup_{u\in[0,\Delta_N]}\E\left(\left|Y^{(i)}_{u}Y^{(j)}_{u} - Y^{(i)}_0 Y^{(j)}_0\right|\right), \quad \text{by stationarity,}\\
&\leq \sup_{u \in [0, \Delta_N]}\E\left(\left|(Y^{(i)}_{u} - Y^{(i)}_{0})(Y^{(j)}_{u}-Y^{(j)}_{0}) + Y^{(j)}_{0}(Y^{(i)}_{u} - Y^{(i)}_{0})  + Y^{(i)}_0(Y^{(j)}_{u}-Y^{(j)}_0)\right|\right].
\end{align*}
Again, using both the triangle and Cauchy--Schwarz's inequalities and finally the  \levy-\Ito decomposition, we see that we have 
$$
T_N^{-1}\E\left(|(\Kn_{N})_{ij} -({\boldsymbol{K}_{T_N}})_{ij}|\right)= O(\Delta_N^{1/2}), \quad \text{as $N \rightarrow \infty$.}
$$
We have $t^{-1}{\Kt}_t \rightarrow \E(\YY_\infty\YY_\infty^\top)$ $\proba_0-a.s.$ as $t\rightarrow \infty$ by ergodicity and therefore $T_N^{-1}\Kn_N \rightarrow \E(\YY_\infty\YY_\infty^\top) < \infty$ $\proba_0-a.s.$ as $N \rightarrow \infty$. 

Next, by the continuous mapping theorem, $T_N\Kn_N^{-1} \xrightarrow{\  \ } \E\left(\YY_\infty \YY_\infty^\top \right)^{-1}$ $\proba_0-a.s.$ as $N \rightarrow \infty$. We conclude that $T_N\Sn_N^{-1} \xrightarrow{\  \ } \E\left(\YY_\infty \YY_\infty^\top \right)^{-1} \otimes \Id < \infty$ $\proba_0-a.s.$ and hence, $\ZZ^1_N \probconv 0$ as $N\rightarrow \infty$.

It remains to show that $\ZZ^2_N \xrightarrow{\ \mathcal{D} \ } \mathcal{N}\left(0, \E\left(\YY_\infty \YY^\top_\infty \right)^{-1} \otimes \Si \right)$ as $N \rightarrow \infty$, which is proved in a similar fashion albeit using a multivariate martingale central limit theorem, namely Theorem 2.2, \cite{crimaldi2005multimgclt}. Again, we have that $T_N \Sn_N^{-1} \probconv  \E\left(\YY_\infty \YY_\infty^\top \right)^{-1} \otimes \Id$. We shall now focus on the remainder term $T_N^{-1/2}\sum_{k=0}^{N-1} \YY_k \otimes \Delta_k \WW$, of which the continuously-integrated equivalent is  $T_N^{-1/2} \Mt_{T_N}$ where $\Mt_t := \int_0^{t} \YY_s \otimes d\WW_s$ for $t \geq 0$. Similarly, we find the $L^1$ convergence of those quantities as $N\rightarrow \infty$. For $i,j,k,l \in\{1,\dots,d\}$, we have that (similarly as above)
$$([\Mt]_t)_{d(i-1)+k,d(j-1)+l} = \int_0^t Y^{(i)}_s Y^{(j)}_s \sum_{m=1}^d (\Si^{1/2})_{km} (\Si^{1/2})_{ml}ds = \int_0^t Y^{(i)}_s Y^{(j)}_s ds \Si_{kl} ,$$ 
therefore $[\Mt]_t = {\Kt}_t \otimes \Si$. Since $\YY$ has finite second moments, $\Mt_t$ is a $d^2$-dimensional continuous martingale under $P_{t,\YY}$, hence $\Mt_t - \lim_{s \rightarrow t^-}\Mt_s = 0$ for any $t \geq 0$. Define the matrix family $\Ubf_t := t^{-1/2}\Id \otimes \Si^{-1/2}$, we have as $t \rightarrow \infty$:
\begin{equation*}
    \begin{cases}
    \|\Ubf_t\| \rightarrow 0,\\
    \Ubf_t \cdot [\Mt]_t \cdot \Ubf_t^\top  \xrightarrow{\proba_{0}-a.s.}  \E\left(\YY_\infty \YY_\infty^\top \right) \otimes \Id,\ \text{which is positive definite,}\\
    \E\left(\sup_{0 \leq s \leq t}\left\|\Ubf_s \cdot (\Mt_s - \lim_{u\rightarrow s^-}\Mt_u)\right\|\right) = 0.
    \end{cases}
\end{equation*}
This is justified by the ergodicity of $(\YY_t, \ t \geq 0)$ and path continuity of $(\Mt_t, \ t \geq 0)$.
Therefore, by Theorem 2.2, \cite{crimaldi2005multimgclt}, $\Ubf_t \cdot  \Mt_t$ converges $\mathcal{F}$-stably %(DEFINITION p.105 Prakasa Roa Semimg and their statistical inference book) to \\
to $\mathcal{N}\left(\boldsymbol{0}_{d^2},\E\left(\YY_\infty \YY_\infty^\top \right)\otimes \Id\right)$. Thus, $T_N^{-1/2}\Mt_{T_N}$ converges in distribution to $\mathcal{N}\left(\boldsymbol{0}_{d^2},  \E\left(\YY_\infty \YY_\infty^\top \right) \otimes \Si \right)$. Similarly to $\Kn_N$ and ${\Kt}_{T_N}$, we obtain that 
$T_N^{-1/2}\int_{0}^{T_N}\YY_{\lfloor s \rfloor_N} d\WW_s$ converges $\mathcal{F}$-stably to $\mathcal{N}\left(\boldsymbol{0}_{d^2},  \E\left(\YY_\infty \YY_\infty^\top \right) \otimes \Si \right)$.

Recall that $T_N\Sn_N^{-1} \probconv \E\left(\YY_\infty \YY_\infty^\top \right)^{-1}\otimes \Id < \infty$ as $N \rightarrow \infty$. We obtain the $\mathcal{F}$-stable convergence by Slutsky's lemma:
$$\ZZ^2_N \stableconv \mathcal{N}\left(\boldsymbol{0}_{d^2}, \E\left(\YY_\infty \YY_\infty^\top \right)^{-1} \otimes \Si \right), \quad \text{as $N\rightarrow \infty$,}$$
which concludes the proof.
\end{proof}

\section{Proof for the Adaptive Lasso}
\label{proof:theorem:adaptive-lasso}
\subsection{Proof of Theorem \ref{theorem:adaptive-lasso}}
\begin{proof}
We first prove the asymptotic normality of the estimator before proving the property of consistency in variable selection. The likelihood part of the objective function is to be written as
\begin{align*}
	\ell_N(\rmQ) - \ell_N(\rmQ_0) &= - \left(\vectorise(\Q) - \vectorise(\Q_0)\right)^\top \cdot (\Id \otimes \Si^{-1}) \cdot \Atilde_N \\ 
	&\qquad - \frac{1}{2} \left(\vectorise(\Q) - \vectorise(\Q_0)\right)^\top \cdot (\widebar{\Kt}_N \otimes \Si^{-1}) \cdot \left(\vectorise(\Q) - \vectorise(\Q_0)\right)\\
	&\qquad - \left(\vectorise(\Q) - \vectorise(\Q_0)\right)^\top \cdot (\widebar{\Kt}_N \otimes \Si^{-1}) \cdot \vectorise(\Q_0),
\end{align*}
We then reformulate further the likelihood.
	In the proof of Lemma \ref{lemma:cv-estimator-continuous-component} (Appendix \ref{proof:lemma:cv-estimator-continuous-component}), we show that $T_N^{-1}\Kn_N \xrightarrow{\  \ } \E\left(\YY_\infty \YY_\infty^\top \right) < \infty$, $\proba_0-a.s.$ as $N \rightarrow \infty$ and write $\Sn_N = \Kn_N \otimes \Id$. In Theorems \ref{lemma:consistency-conservation-jump-filtering-finite} \& \ref{lemma:consistency-conservation-jump-filtering-infinite}, we show that $ \Sn_N^{-1} \cdot \left(\Atilde_N - \Abar_N\right) \probconv \boldsymbol{0}_{d^2}$ as $N\rightarrow\infty$, where $\Abar_N = \sum_{k=0}^{N-1} \YY_k \otimes \Delta_{k}  \YY^{c}$. Since $\YY^c$ satisfies $d\YY^c_t  = - \rmQ_0 \YY_{t-}dt + d\WW_t$, then $\Abar_N = - \sum_{k=0}^{N-1} \int_{t_k}^{t_{k+1}}\YY_{\lfloor s \rfloor_N} \YY_{s}^\top ds \otimes \Id \cdot \vectorise(\rmQ_0) + \sum_{k=0}^{N-1} \YY_k \otimes \Delta_{k}  \WW$, where $\lfloor s \rfloor_N := t_k$ for any $s \in [t_k, t_{k+1})$ for $k \in \{1,\dots, N\}$.

	Let $\widebar{\rmQ}_{N} := - {\widebar{\Kt}_N}^{-1} \cdot \vectorise(\Abar_N)$. We recall that $\widetilde{\rmQ}_{N} - \widebar{\rmQ}_{N} \probconv 0_{d^2}$ and that, from that same we also obtain
\begin{align*}
   T_N^{1/2}\left(\vectorise(\widebar{\rmQ}_{N})-\vectorise(\rmQ_0)\right) = \ZZ^1_N + \ZZ^2_N,
\end{align*}
where
\begin{align*}
\begin{cases}
	\ZZ^1_N &= \quad T_N^{1/2}\Sn_N^{-1} \bigg[\left(\int_{0}^{T_N}\YY_{\lfloor s \rfloor_N} \YY_s^\top ds - \Kn_N\right) \otimes \Id  \bigg]\cdot \vectorise(\rmQ_0),\\
	\ZZ^2_N &= \quad T_N^{1/2}\Sn_N^{-1}\sum_{k=0}^{N-1}\YY_k \otimes \Delta_k\WW,
\end{cases}
\end{align*}
such that $\ZZ^1_N \probconv 0_{d^2}$ and $\ZZ^2_N \stableconv \mathcal{N}(0_{d^2}, \E(\YY_\infty \YY^\top_\infty )^{-1} \otimes \Si)$ as $N\rightarrow\infty$. Using those notations, the log-likelihood is reformulated as
\begin{align*}
\ell_N(\rmQ) - \ell_N(\rmQ_0) &= - T_N^{-1/2}\left(\vectorise(\Q) - \vectorise(\Q_0)\right)^\top \cdot (\Kn_N \otimes \Si^{-1}) \cdot \ZZ^2_N\\ 
	&\quad - \frac{1}{2} \left(\vectorise(\Q) - \vectorise(\Q_0)\right)^\top \cdot (\widebar{\Kt}_N \otimes \Si^{-1}) \cdot \left(\vectorise(\Q) - \vectorise(\Q_0)\right) \\
	&\quad + T_N^{-1/2}\left(\vectorise(\Q) - \vectorise(\Q_0)\right)^\top \cdot (\Kn_N \otimes \Si^{-1}) \cdot \ZZ^1_N\\
	&\quad - \left(\vectorise(\Q) - \vectorise(\Q_0)\right)^\top \cdot (\Id\otimes \Si^{-1}) \cdot \left(\Atilde_N - \Abar_N\right).
\end{align*}
We then proceed similarly to the proof of Th.\ 5.3.1, \cite{courgeau2020likelihood} (see App.\ A.10 therein). First, without loss of generality, we change the penalty rate from $\lambda$  to $\lambda T_N$ since it does not depend on $\rmQ$. Then, by writing $\rmQ = \rmQ_0 + T_N^{-1/2}\rmM$ for some $\rmM \in \MdR$, we have
\begin{align*}
	T_N^{1/2}\left(\widetilde{\rmQ}_{\AL,N} - \rmQ_0\right)_{|\gQ_0} &= \argmax_{\rmM \in \MdR} \kappa^1_N(\rmM) + \kappa^2_N(\rmM) + \kappa^3_N(\rmM),
\end{align*}
where
\begin{align*}
	\begin{cases}
		\kappa^1_N(\rmM) &= \ -\vectorise(\rmM)^\top \cdot (T_N^{-1}\Kn_N \otimes \Si^{-1}) \cdot \ZZ^2_N - \frac{1}{2}  \vectorise(\rmM)^\top \cdot (T_N^{-1}\widebar{\Kt}_N \otimes \Si^{-1}) \cdot \vectorise(\rmM),\\ 
		\kappa^2_N(\rmM) &= \ \lambda T_N \|\rmQ_0 \odot |\widetilde{\rmQ}_N|^{-\gamma} \|_1 - \lambda T_N \|(\rmQ_0 + T_N^{-1/2}\rmM) \odot |\widetilde{\rmQ}_N|^{-\gamma} \|_1,\\
		\kappa^3_N(\rmM) &=  \vectorise(\rmM)^\top \cdot (T_N^{-1}\Kn_N \otimes \Si^{-1}) \cdot \ZZ^1_N  - T_N^{-1/2} \vectorise(\rmM)^\top \cdot (\Id\otimes \Si^{-1}) \cdot \left(\Atilde_N - \Abar_N\right),
	\end{cases}
\end{align*}
where the third term $\kappa^3_N$ represents the discretisation error which will become asymptotically negligible as shown next. We denote and define the complement of $\gQ_0$ as follows $\widebar{\gQ}_0 := \{(k,l)\not\in\gQ_0: 1\leq k,l \leq d\}$.  Then, we treat each term $\kappa^i_N$, $i\in\{1,2,3\}$, separately:
\begin{itemize}
	\item \textbf{For the first function $\kappa^1_N$:} by Slutsky's lemma (for stable convergence \citep[Th.\ 1.1]{hausler2015whystableconvergence}), the first term converges stably to a standard (centred) Gaussian distribution with covariance $K^1_\infty := \vectorise(\rmM)^\top \cdot \E(\YY_\infty \YY^\top_\infty ) \otimes \Si^{-1}\cdot \vectorise(\rmM)$. Both \cite{courgeau2020likelihood} and \cite{Gaiffas2019SparseOUProcess} use a $d \times d$ Gaussian matrix $\rmZ$ with mean zero and covariance $\Cov(\vectorise(\mZ), \vectorise(\mZ))= \E(\YY_\infty \YY^\top_\infty ) \otimes \Si^{-1}$, that is $\Cov(\mZ_{ij},\mZ_{kl}) = \E(\YY_\infty \YY^\top_\infty )_{ij} \cdot \Si^{-1}_{kl}$. Therefore, for any $\rmM \in \MdR$, $\trace(\rmM \mZ)$ is a Gaussian random variable with mean zero and variance
\begin{align*}
	\Var\left(\trace(\rmM \mZ)\right) &= \sum_{ijkl}\rmM_{ji} \rmM_{lk} \Cov(\mZ_{ij}, \mZ_{kl}) = \sum_{ijkl}\rmM_{ji} \cdot (\mK_\infty)_{ik} \cdot \Si^{-1}_{jl} \cdot \rmM_{lk} = K^1_\infty.
\end{align*}
Also, the second term converges almost-surely to $-K^1_\infty/2$ such that by Slutsky's lemma, we have for any $\rmM \in \MdR$ that $\kappa^1_N(\rmM) \stableconv -K^1_\infty/2 - \trace(\rmM \mZ)$ as $N\rightarrow \infty$.
	\item \textbf{For the second function $\kappa^2_N$:} the argument from Appendix A.10, \cite{courgeau2020likelihood} can be replicated exactly and we do not present it for brevity. In essence, by consistency of the MLE (Corollary \ref{corollary:consistency-estimator-continuous-component}) and since $\lambda(N) N^{1/2} \rightarrow 0$ and $\lambda(N) N^{(\gamma+1)/2} \rightarrow \infty$ as $N \rightarrow \infty$, the objective function diverges to $-\infty$ if $\supp(\rmM) \cap \widebar{\gG}_0 \neq \emptyset$ and \emph{zero} otherwise for any $\rmM \in \MdR$.
	\item \textbf{For the third function $\kappa^3_N$:} each term converges to zero in probability for any $\rmM \in \MdR$ as $N \rightarrow \infty$.
\end{itemize}
We summarise the resulting asymptotic properties as follows:
$$
\kappa^1_N(\rmM) + \kappa^2_N(\rmM)  + \kappa^3_N(\rmM)  \xrightarrow{\ \mathcal{D}\ } \begin{cases}
 	-\infty, & \text{if } \supp(\rmM) \cap \widebar{\gG}_0 \neq \emptyset,\\
 	- K^1_\infty/2 -  \trace\left(\rmM \mZ\right), & \text{otherwise}.
 \end{cases}
$$
This is the same asymptotic property is in \cite{courgeau2020likelihood} and therefore, we borrow their conclusion to show that the maximum in the limit $t\rightarrow \infty$, denoted $\widehat{\rmM}$, is given by
$$\vectorise(\widehat{\rmM}_{|\gG_0}):= - \left(\E(\YY_\infty \YY_\infty^\top )^{-1} \otimes \Si \right)_{|\gG_0 \times \gG_0} \cdot \vectorise(\mZ^\top_{|\gG_0}) , \quad \text{and} \quad \widehat{\rmM}_{|\widebar{\gG}_0} := 0.$$
Given the covariance structure of $\mZ$, $\vectorise(\widehat{\rmM}_{|\gG_0})$ is a centred Gaussian random vector with covariance $(\E(\YY_\infty \YY_\infty^\top )^{-1} \otimes \Si)_{|\gG_0 \times \gG_0}$ which concludes the proof of asymptotic normality. 

The consistency in variable selection is shown similarly to the continuous-time observations case \citep{courgeau2020likelihood}: the asymptotic normality of $\widetilde{\rmQ}_{\AL,N}$ on $\gG_0$  yields that $\proba\left((\widetilde{\rmQ}_{\AL,N})_{ij} \neq 0 \right) \rightarrow 1$ if $(i,j) \in \gG_0$. We then show that $\proba\left((\widetilde{\rmQ}_{\AL,N})_{ij} = 0 \right) \rightarrow 1$ if $(i,j) \in \widebar{\gG}_0$. We suppose that $(\widetilde{\rmQ}_{\AL,N})_{ij} \neq 0$ and then, by taking the derivative of the objective function with respect to the $(i,j)$-th parameter of $\widetilde{\rmQ}_{\AL,N}$. We multiply by $T_N^{-1/2}$, set it to zero before taking the absolute value on both sides. We obtain that
\begin{align*}
	&\left| \vone_{d^2}^\top \cdot (T_N^{-1}\Kn_N \otimes \Si^{-1}) \cdot \ZZ^2_N + \left(\vectorise(\widehat{\rmQ}_{\AL,N})^\top \cdot (T_N^{-1}\Kn_N) \otimes \Si^{-1}\right)_{d(j-1)+i} +o_p(1)\right|\\
	&\qquad \qquad \qquad \qquad  =  \lambda T_N^{(\gamma+1)/2}\cdot |T_N^{1/2}\widetilde{\rmQ}_{N}|^{-\gamma}_{ij}.
\end{align*}
For the right-hand side: from the first part of the proof, we have $T_N^{1/2}(\widetilde{\rmQ}_{\AL,N})_{|\widebar{\gG}_0} = o_p(1)$. Since $\lambda T_N^{(\gamma+1)/2}= O(1)$ and $\gamma > 0$, this side diverges to $\infty$ in probability as $N \rightarrow \infty$. 

For the left-hand side: since $T_N^{-1}\Kn_N \xrightarrow{\ \ } \E(\YY_\infty \YY_\infty^\top ),\ \proba_0-a.s.$, the first term is normally-distributed with finite variance as $N \rightarrow \infty$. By the asymptotic normality of $T_N^{1/2}(\widetilde{\rmQ}_{\AL,N})_{|\widebar{\gG}_0}$, the second term is also asymptotically normal with finite variance. That is, the left-hand side is the absolute value of the sum of two Gaussian random variables whose probability to be larger or equal to the right-hand side (which diverges to $\infty$) tends to zero with probability one. The right-hand side was computed under the assumption that $(\widehat{\rmQ}_{\AL,t})_{ij} \neq 0$, and the probability that event happening is upper-bounded by zero, hence we have that $\proba\left((\widehat{\rmQ}_{\AL,t})_{ij} = 0 \right) \rightarrow 1$ if $(i,j) \in \widebar{\gG}_0$. This concludes the proof of the consistency in variable selection.
\end{proof}

\section{Technical results}
\label{section:proofs}
In this section, we present a collection of intermediate results leading up to Theorems \ref{theorem:discrete-vec-finite-activity}--\ref{theorem:discrete-vec-finite-activity-2-params} \& \ref{th:discrete_clt_infinite}--\ref{theorem:discrete-vec-infinite-activity-2-params} whilst the proofs are presented in Appendix \ref{appendix:proofs-finite} and \ref{appendix:proofs-infinite}, respectively, for the finite and infinite jump activity case. 

\subsection{Intuition}
From Lemma \ref{lemma:cv-estimator-continuous-component} in Section \ref{section:jump-filtered-quantities}, the discretised unfiltered estimator $\overline{\boldPsi}_N$ converges (stably) in distribution to the limit of the continuous-time estimator with time horizon $T_N$, namely $\widehat{\boldPsi}_{T_N}$, in the sense that
$$T_N^{1/2}(\boldsymbol{\overline{\psi}}_N - \boldsymbol{\psi}) \stableconv \mathcal{N}\left(\boldsymbol{0}_{d^2},\ \mathbb{E}\left(\YY_\infty \YY_\infty^\top \right)^{-1} \otimes \Si \right), \quad \text{as $N\rightarrow\infty$.}
$$
The purpose of the following sections is to prove that we can extend this result to the jump-filtered estimator $\widetilde{\boldPsi}_N$. We show that filtering the jumps with a particular threshold gives way to a consistent estimation in both the finite and infinite jump activity cases. This is carried out by treating the different terms given by the \levy-\Ito decomposition of $\LL$ separately. Note that the decomposition differs significantly in terms of the components' properties be it in the finite or the infinite jump activity case. The latter requires intermediate filtering steps to relate the discretised jump-filtered estimator and its unfiltered counterpart.

\subsection{Finite jump activity}
\label{section:technical-lemmas}

%First, we prove that one can discretise the time horizon $T$ into $T_N$ whilst keeping the asymptotic convergence properties true, then, prove the convergence of the discretised unfiltered estimator. Finally, we prove the the consistency of the discretised filtered estimator with respect to the unfiltered one which allows to concludes the proof by Slutsky's lemma.

In this section, we prove Theorems \ref{theorem:discrete-vec-finite-activity} and \ref{theorem:discrete-vec-finite-activity-2-params}.
Recall that, under Assumption \ref{assumption:mle-convergence}, the \levy-\Ito decomposition (see Appendix \ref{appendix:levy-ito-decomposition}) yields that there exists a centred Gaussian \Levy process $\WW_t$ with covariance matrix $\Si$ and almost-surely continuous paths as well as a pure-jump \Levy process $\JJ_t$ independent of $\WW_t$ such that
$$\LL_t = \WW_t + \JJ_t.$$
Finally, under the finite jump activity assumption (Assumption \ref{assumption:jump-height-finite-activity}), remark that $\JJ_t$ is a compound Poisson process. Indeed, there are Poisson processes $N_t^{(j)}$ with elementwise \emph{finite} intensities $\lambda^{(j)}$ such that $(\lambda^{(1)},\dots,\lambda^{(d)}) := \nu(\R^d) < \infty$ such that $\JJ_t^{(j)}=\sum_{k=0}^{N_t^{(j)}}Z^{(j)}_k$ where $Z^{(j)}_k$ are i.i.d.\ jump heights with distribution $F^{(i)}$ as in Assumption \ref{assumption:jump-height-finite-activity}.

\begin{corollary}
For $\YY$ solving Equation \eqref{eq:sde}, it can be written as $\YY_t = \YY_0 -\int_0^t\Q \YY_{s-} ds + \WW_t + \JJ_t$ and its continuous part $\YY^c$ is given by $\YY^c_t = \YY^c_0 -\int_0^t\Q \YY_{s-} ds + \WW_t$ for any $t \geq 0$.
\end{corollary}

For simplicity, we introduce some notation for the drift term:
\begin{notation}
\label{notation:drift-term}
We denote by $\Delta_k\DD$ the drift term between $t_k$ and $t_{k+1}$ as expressed by
 $$\Delta_k \DD  = (\Delta_k D^{(1)},\dots, \Delta_k D^{(d)})^\top:=-\int_{t_k}^{t_{k+1}} \Q\YY_{s-} ds. $$
\end{notation}

We first introduce results proved in \cite{mai2014efficient} but adapted to multiple dimensions. A crucial result is the upper bound on the probability of the continuous part increments $\{|\Delta_k W^{(j)} + \Delta_k D^{(j)}| > \Delta_N^{1/2-\delta}\}$ for any $\delta \in (0,1/2)$ as formulated below:
\begin{lemma}{(adapted from Lemma 3.7, \cite{mai2014efficient})\\}
\label{lemma:lemma-3.7}
Suppose that for any $i\in\{1,\dots,d\}$, $$\sup_{s \geq 0}\E\left(|Y^{(i)}_s|^{l^{(i)}}\right) < \infty, \quad \text{ for some $l^{(i)} \geq 1$}.$$
Then, for any $\delta \in (0,1/2)$ and $k\in\{1,\dots,N\}$, we have:
$$ \proba\left(|\Delta_k W^{(i)} + \Delta_k D^{(i)}| > \Delta_N^{1/2-\delta} \right) = O(\Delta_N^{l^{(i)} (1/2+\delta)}), \quad \text{as $N\rightarrow\infty$.}$$
\end{lemma}
\begin{proof}
	The proof in presented on page 925, \cite{mai2014efficient}. %See Appendix \ref{proof:lemma:lemma-3.7} otherwise.
\end{proof}

%\begin{remark}
%This procedure relies on the univariate case treated in Section 3, \cite{mai2014efficient} and especially  Lemmas 3.7 and 3.8 proved therein.	
%\end{remark}

\cite{mai2014efficient} introduced a family of events when a consensus between the filtering and the absence of jumps takes place. We extend to our multivariate framework by introducing equivalent events componentwise.
\begin{definition}
\label{definition:events-small-increments-no-jumps}
	Let $(\YY, \ t \geq 0)$ be a GrOU process such that Assumptions \ref{assumption:strong-solution-vector}--\ref{assumption:jump-height-finite-activity} hold. We define the sequence of events $(A^{(i)}_{k,N}: k \in \{1,\dots,N-1\})$ for the $i$-th node where a small increments corresponds to the absence of jumps, i.e.\ define
$$A^{(i)}_{k,N} := \left\{\omega \in \Omega : \indicator_{\{|\Delta_k Y^{(i)}| \leq v_N^{(i)}\}}(\omega) = \indicator_{\{\Delta_k N^{(i)} = 0\}}(\omega)\right\}, \quad k \in \{0,\dots,N-1\}.$$
\end{definition}
\begin{notation}
	We sometimes write $\indicator \{S\}$ for $\omega \mapsto \indicator_S(\omega)$ where $S \in \mathcal{F}$.
\end{notation}
They proved that the filtering and the absence of jumps coincide all the time with probability \emph{one} as $N\rightarrow \infty$ in the following sense:
\begin{lemma}{(Lemma 3.8, \cite{mai2014efficient})\\}
\label{lemma:equivalence-filtering-no-jump}
	Suppose Assumptions \ref{assumption:strong-solution-vector}--\ref{assumption:jump-height-finite-activity} hold for  a GrOU process $(\YY, \ t \geq 0)$ and $(\beta^{(i)}: i \in \{1,\dots,d\})\subset (0,1/2)$. Let $i \in \{1,\dots,d\}$ and suppose that $v_N^{(i)} := \Delta_N^{\beta^{(i)}}$, then for $A^{(i)}_{N} := \cap_{k=1}^{N}A^{(i)}_{k,N}$, one has
$$\proba_{\YY}\left(A^{(i)}_N\right) \rightarrow 1, \quad \text{as $N\rightarrow \infty$.}$$
\end{lemma}
\begin{proof}
See Section 3.3.1, pages 925--927, \cite{mai2014efficient}.
\end{proof}

We now extend Lemma 3.9, \cite{mai2014efficient} to the multivariate setting. Let $i,j\in\{1,\dots,d\}$. We show that we can approximate (in the $L^1$ sense) the discretised integral of $Y^{(i)}$ with respect to the continuous part of $Y^{(j)}_t$ by jump-filtering $Y^{(j)}_t$ itself with a threshold $v_N^{(j)} = \Delta_N^{\beta^{(j)}}$  as follows:

\begin{lemma}{(Extension of Lemma 3.9, \cite{mai2014efficient})\\}
\label{lemma:cts_filtered_approx}
Assume Assumptions \ref{assumption:strong-solution-vector}--\ref{assumption:A-known} \& \ref{assumption:psi-grou}--\ref{assumption:jump-height-finite-activity} hold for a GrOU process $(\YY_t,\ t \geq 0)$ and $(\beta^{(i)}: i \in \{1,\dots,d\})\subset (0,1/2)$.
If $v_N^{(j)} := \Delta_N^{\beta^{(j)}}$, then we have as $N \rightarrow \infty$:
\begin{equation*}
%    \label{eq:filtering-similar-to-cts-part}
    \E\left(\left| \sum_{k=0}^{N-1} Y^{(i)}_k \left( \Delta_k Y^{(j)} \indicator{\{|\Delta_kY^{(j)}\leq v_N^{(j)}\}} - \Delta_kY^{(j),c}\right)\right|\right) = O(\Delta_N^{1/2}T_N),
\end{equation*}
for any $i,j\in\{1,\dots,d\}$.
\end{lemma}
%\begin{proof}
%	See Appendix \ref{proof:lemma:cts_filtered_approx}.
%\end{proof}

\begin{proof}[Proof of Lemma \ref{lemma:cts_filtered_approx}]
Let $i, j \in \{1,\dots, d\}$. Following Lemma \ref{lemma:equivalence-filtering-no-jump}, by conditioning on $A^{(j)}_N := \cap_{k=1}^{N}A^{(j)}_{k,N}$, i.e.\ where increments smaller than the jump thresholds coincide with the absence of jumps at all times, %that $P\left\{A^{(j)}_N\right\}\rightarrow 1$, as $N \rightarrow \infty$
we have:
\begin{align*}
	\sum_{k=0}^{N-1} Y^{(i)}_k \bigg(\Delta_k Y^{(j)}&\indicator{\{|\Delta_kY^{(j)}|\leq v_N^{(i)}\}} - \Delta_k Y^{(j),c}\bigg)\\
	&= \sum_{k=0}^{N-1} Y^{(i)}_k \left(\Delta_k Y^{(j)}\indicator{\{\Delta_k N^{(j)} = 0\}} - \Delta_k Y^{(j),c}\right).
\end{align*}
Then, $\Delta_k Y^{(j)}\indicator\{\Delta_k N^{(j)} = 0\} - \Delta_k Y^{(j),c}$ is indeed equal to zero if no jump occurs on $[t_k,t_{k+1})$, that is if $\{\Delta_k N^{(j)} = 0\}$ happens, and is equal to $-\Delta_kY^{(j),c}$ if $\{\Delta_k N^{(j)} > 0\}$ happens. Thus, we define the event $B^{(j)}_{k,N} := \{\Delta_k N^{(j)} > 0\}$. Now, we can rewrite the above equality as follows
\hspace{-0.3cm}
{\small
\begin{align*}
\indicator\{A^{(j)}_N\}\sum_{k=0}^{N-1}& \bigg|Y^{(i)}_k \left(\Delta_k Y^{(j)}\indicator{\{|\Delta_kY^{(j)}|\leq v_N^{(i)}\}} - \Delta_k Y^{(j),c}\right)\bigg|=\sum_{k=0}^{N-1} \left|Y^{(i)}_k \Delta_k Y^{(j),c}\right|\indicator\left\{A^{(j)}_N\cap B^{(j)}_{k,N}\right\}.
\end{align*}}
By the \levy-\Ito decomposition, the increments of the continuous part can be written $\Delta_k Y^{(j),c} = \Delta_k W^{(j)} + \Delta_k D^{(j)}$ where we recall that $\Delta_k D^{(j)} := -\int_{t_k}^{t_{k+1}}(\Q \YY_s)^{(j)} ds$. Hence, by the triangle inequality 
\begin{align*}
    \bigg|\sum_{k=0}^{N-1} Y^{(i)}_k \Delta_k Y^{(j),c}\indicator\{A^{(j)}_N\cap B^{(j)}_{k,N}\}\bigg| \leq \sum_{k=0}^{N-1} \left(\big|Y^{(i)}_k \Delta_k W^{(j)}\big| + \big|Y^{(i)}_k \Delta_k D^{(j)}\big|\right)\indicator\{B^{(j)}_{k,N}\}.
\end{align*}
Thus, we compute both terms of the right-hand side separately. As in the proof of Lemma 3.9, \cite{mai2014efficient}, observe that since $\Delta_k N^{(j)}$ is Poisson-distributed with rate $\lambda \times (t_{k+1,N} - t_{k,N}) > 0$, we have $\proba\left\{B^{(j)}_{k,N}\right\} \leq \lambda \Delta_N $. The independence between $W^{(j)}_t$ and $N^{(j)}_t$ as well as between $\Delta_k N^{(j)}$ and $Y^{(i)}_k$ gives way to a succinct decomposition
\begin{align*}
    \sum_{k=0}^{N-1} \E\left(\big|Y^{(i)}_k \Delta_k W^{(j)}\big|\indicator\{ B^{(j)}_{k,N}\} \right) &= \sum_{k=0}^{N-1} \E\left(\big|Y^{(i)}_k\big|\right)\E\left( \big| \Delta_k W^{(j)}\big|\right)\proba( B^{(j)}_{k,N})\\
    &\leq O(\lambda N\Delta_N^{1/2}\Delta_N)\\
    &= O(T_N \Delta_N^{1/2}), \quad \text{since $N\Delta_N = O(T_N)$}.
\end{align*} 
For the second term, we use  H{\"o}lder's inequality to obtain
\begin{align*}
    \sum_{k=0}^{N-1} \E\left(\big|Y^{(i)}_k \Delta_k D^{(j)}\big|\indicator\{ B^{(j)}_{k,N}\} \right) &\leq \sum_{k=0}^{N-1} \E\left(|Y^{(i)}_k|^2\indicator\{ B^{(j)}_{k,N}\} \right)^{1/2} \E\left( |\Delta_k D^{(j)}|^2\right)^{1/2}.
\end{align*}
Then, by independence of $Y^{(i)}_k$ and $\Delta_k N^{(j)}$, we have:
\begin{align*}
    \sum_{k=0}^{N-1} \E\left(\big|Y^{(i)}_k \Delta_k D^{(j)}\big|\indicator\{ B^{(j)}_{k,N}\}  \right) &\leq \sum_{k=0}^{N-1} \E\left(|Y^{(i)}_k|^2 \right)^{1/2}\proba\left(B^{(j)}_{k,N}\right)^{1/2}\E\left(|\Delta_k D^{(j)}|^2\right)^{1/2}.
\end{align*}
Again by H{\"o}lder's inequality, observe that:
\begin{align*}
	\left|\Delta_k D^{(j)}\right|^2 &\leq  (t_{k+1}-t_k) \int_{t_k}^{t_{k+1}} \left|(\Q \YY_s)^{(j)}\right|^2 ds\leq \|\Q\|^2 \Delta_N  \int_{t_k}^{t_{k+1}}\|\YY_s\|^2 ds.
\end{align*}
Since $\YY$ is stationary and has finite second moments, we obtain by Fubini's theorem
$$ \E\left(\left|\Delta_k D^{(j)}\right|^2\right)^{1/2} \leq \|\Q\|  \E\left(\|\YY_0\|^2\right)^{1/2} \Delta_N = O(\Delta_N).$$
We conclude that
\begin{align*}
     \sum_{k=0}^{N-1} \E\left(\big|Y^{(i)}_k \Delta_k (\Q\YY)^{(j)}\big| \right)&= O(N\Delta_N^{3/2}) = O(\Delta_N^{1/2}T_N).
\end{align*}
\end{proof}

Lemma \ref{lemma:cts_filtered_approx} is essential as it links the continuous part of $\YY$ increments to the \emph{jump-filtered} increments. We can then compute the difference $\boldsymbol{\widetilde{\psi}}_N - \boldsymbol{\overline{\psi}}_N$ element-wise and show that as given in Appendix \ref{proof:lemma:consistency-conservation-jump-filtering-finite}: 
$$T_N^{1/2}(\boldsymbol{\widetilde{\psi}}_N - \boldsymbol{\overline{\psi}}_N) \probconv\boldsymbol{0}_{d^2}, \quad \text{as $N\rightarrow\infty$.}
$$

%\begin{lemma}{(Discretised time-continuous estimator convergence)\\}
%\label{lemma:cv-estimator-continuous-component}
%In the same framework as in Theorem \ref{theorem:discrete-vec-finite-activity}, the discretised estimator $\boldsymbol{\overline{\psi}}_N$ defined in Equation \eqref{eq:discretised-estimator}
%$$T_N^{1/2}\left(\boldsymbol{\overline{\psi}}_N - \boldsymbol{\psi}\right) \xrightarrow{ \ \mathcal{D} \ } \mathcal{N}\left(\boldsymbol{0}_{d^2},\ \mathbb{E}\left\{\YY_\infty \YY_\infty^\top \right\}^{-1} \otimes \Si \right), \quad \text{as $N\rightarrow\infty$.}
%$$
%\end{lemma}
%\begin{proof}
%	See Appendix \ref{proof:lemma:cv-estimator-continuous-component}.
%\end{proof}

\subsection{Infinite jump activity}
\label{section:infinite-appendix}
Similarly to Section 4, \cite{mai2014efficient}, we put forth a collection of convergence results to prove Theorems \ref{th:discrete_clt_infinite} \& \ref{theorem:discrete-vec-infinite-activity-2-params}: we first justify the negligibility of some jump-filtered quantities (in the sense of convergence in probability to zero). We also prove that one can properly estimate various unfiltered quantities with jump-filtered equivalents with a high-frequency sampling scheme. Since we are to show the consistency of the jump-filtered estimator with respect to the unfiltered one by splitting the difference in four different terms, the proofs are relegated to Appendix \ref{appendix:proofs-infinite} for clarity.

The \emph{jump-filtered} part of the process is actually going to zero in probability as shown in the following lemma:
\begin{lemma}{(Extension of Lemma 4.14, \cite{mai2014efficient})\\}
\label{lemma:lemma-4.14-non-jump-neg}
Under the assumptions of Theorem \ref{th:discrete_clt_infinite}, we have for $i,j\in\{1,\dots,d\}$:
$$T_N^{-1/2}\sum_{k=0}^{N-1} Y^{(i)}_k\Delta_kY^{(j)} \indicator_{\{|\Delta_kW^{(j)} + \Delta_k D^{(j)}| > v_N^{(j)}\}}\probconv 0, \quad \text{as $N\rightarrow \infty$.}$$
\end{lemma}
\begin{proof}
	See Appendix \ref{proof:lemma:lemma-4.14-non-jump-neg}.
\end{proof}

The second lemma proves that filtered signals with large jumps are negligible:
\begin{lemma}{(Extension of Lemma 4.13, \cite{mai2014efficient})\\}
\label{lemma:lemma-4.13-go-big-jump}
Under the assumptions of Theorem \ref{th:discrete_clt_infinite}, we have for $i,j\in\{1,\dots,d\}$:
$$T_N^{-1/2}\sum_{k=0}^{N-1} Y^{(i)}_k\Delta_kY^{(j)} \indicator_{\left\{|\Delta_kY^{(j)}| \leq v_N^{(j)}, |\Delta_k J^{(j)}| > 2v_N^{(j)}\right\}} \probconv 0, \quad \text{as $N\rightarrow \infty$.}$$
\end{lemma}
\begin{proof}
	See Appendix \ref{proof:lemma:lemma-4.13-go-big-jump}.
\end{proof}

Finally, the last two results of this section allow to asymptotically \emph{filter out} the contribution of small jumps larger than $v_n^{(j)}/2$, for $j\in\{1,\dots,d\}$. First, one recalls an essential couple of results from \cite{mai2014efficient}:
\begin{lemma}{(Equation (38), \cite{mai2014efficient})\\}
\label{lemma:small_jump_negligeable}
Under the assumptions of Theorem \ref{th:discrete_clt_infinite} and for any $j\in\{1,\dots,d\}$,  one obtains $$\E\left( |\Delta_k J^{2,(j)}|^2\indicator \left\{|\Delta_k J^{2,(j)}| \leq 2v^{(j)}_N\right\} \right) = O(\Delta_N^{1 +  \beta^{(j)}(2-\alpha)}),\quad \text{as $N \rightarrow \infty$.}$$
Also, one has that
$$\proba\left( v^{(j)}_N/2 < |\Delta_k J^{2,(j)}| \leq 2v^{(j)}_N \right) = O(\Delta_N^{1-\alpha \beta^{(j)}}) \longrightarrow 0,\quad \text{as $N \rightarrow \infty$.}$$

\end{lemma}
\begin{proof}
	See page 949, \cite{mai2014efficient}. %See Appendix \ref{proof:lemma:small_jump_negligeable}.
\end{proof}

Secondly, we prove the negligibility of a $\JJ^{2}-$filtered discrete equivalent to $\widetilde{\At}_N$ (Definition \ref{definition:discretised-estimators}) as a convergence in probability to zero in the following sense:
\begin{lemma}{(Extension of Lemma 4.15. \cite{mai2014efficient})\\}
\label{lemma:lemma-4.15-small-jump-sandwich}
Under the assumption of Theorem \ref{th:discrete_clt_infinite} and for $i,j\in\{1,\dots,d\}$, we have:
$$T_N^{-1/2}\sum_{k=0}^{N-1}Y^{(i)}\Delta_kY^{(j)}\indicator\{v_N^{(j)}/2 < |\Delta_k J^{2}| \leq 2v^{(j)}_N\} \probconv 0, \qquad \text{as $N \rightarrow \infty$.}$$
\end{lemma}
\begin{proof}
	See Appendix \ref{proof:lemma:lemma-4.15-small-jump-sandwich}.
\end{proof}

It can be shown that in infinite jump activity case, small increments of $\YY$ are dominated by small jumps. We prove that sums of jumps increments are asymptotically zero if one filters for small enough jumps:
\begin{lemma}
	\label{lemma:small-enough-jumps}
	Under the assumptions of Theorem \ref{th:discrete_clt_infinite}, we have for $i,j\in\{1,\dots,d\}$:
	$$T_N^{-1/2}\sum_{k=0}^{N-1} Y^{(i)}_k
     \Delta_k J^{(j)} \indicator\{|\Delta_kJ^{(j)}|\leq 2v_N^{(j)}, |\Delta_k J^{1,(j)}| \neq  0\} \probconv 0, \ \text{as $N\rightarrow \infty$.}$$
\end{lemma}
\begin{proof}
	See Appendix \ref{proof:lemma:small-enough-jumps}.
\end{proof}
As an intermediary result, we prove that one can focus on the small-jump component as non-zero large jumps are asymptotically negligible as follows
\begin{lemma}
	\label{lemma:lemma-4.12-supp-large-jump-negligible}
	Under the assumptions of Theorem \ref{th:discrete_clt_infinite}, we have for $i,j\in\{1,\dots,d\}$:
	$$T_N^{-1/2} \sum_{k=0}^{N-1}Y^{(i)}_k \Delta_k Y^{(j)} \indicator\{|\Delta_k J^{(j)}|\leq 2v_N^{(j)}, |\Delta_k J^{1,(j)}| \neq 0\} \probconv 0, \quad \text{as $N \rightarrow \infty$.}$$
\end{lemma}
\begin{proof}
	See Appendix \ref{proof:lemma:lemma-4.12-supp-large-jump-negligible}.
\end{proof}

Then, we prove that filtering for small \emph{increments} is equivalent to filtering for small \emph{jumps} in the sense of the following lemma:
\begin{lemma}{(Extension of Lemma 4.12, \cite{mai2014efficient})\\}
\label{lemma:lemma-4.12-small-inc-are-small-jumps}
Under the assumptions of Theorem \ref{th:discrete_clt_infinite}, we have for $i,j\in\{1,\dots,d\}$:
$$T_N^{-1/2}\sum_{k=0}^{N-1} Y^{(i)}_k\Delta_kY^{(j)} \left(\indicator\{|\Delta_kY^{(j)}| \leq v_N^{(j)}\} - \indicator\{|\Delta_k J^{2,(j)}| \leq 2v_N^{(j)}\}\right) \probconv 0, \quad \text{as $N\rightarrow \infty$.}$$
Also, we prove that more generally
$$T_N^{-1/2}\sum_{k=0}^{N-1} Y^{(i)}_k\Delta_kY^{(j)} \left(\indicator\{|\Delta_kY^{(j)}| \leq v_N^{(j)}\} - \indicator\{|\Delta_k J^{(j)}| \leq 2v_N^{(j)}\}\right) \probconv 0, \quad \text{as $N\rightarrow \infty$.}$$
\end{lemma}
\begin{proof}
	See Appendix \ref{proof:lemma:lemma-4.12-small-inc-are-small-jumps}.
\end{proof}

\subsubsection{Filtering the drift}
The following lemma shows that under jump-filtering with threshold $v_N^{(j)} =\Delta_N^{\beta^{(j)}}$, the discrete integration with respect to the drift can be approximated regardless if the filtering is applied.
\begin{lemma}{(Extension of Lemma 4.11, \cite{mai2014efficient})\\}
\label{lemma:lemma-mai-4.11-drift}
Under the assumptions of Theorem \ref{th:discrete_clt_infinite}, then
$$T_N^{-1/2}\sum_{k=0}^{N-1} Y^{(i)}_k \cdot \left( \Delta_k D^{(j)} \cdot \indicator\{|\Delta_k Y^{(j)}| \leq v_N^{(j)}\} - \Delta_k D^{(j)} \right) \probconv 0. $$
or more explicitly:
$$T_N^{-1/2}\sum_{k=0}^{N-1} Y^{(i)}_k \cdot  \left(\int_{t_{k}}^{t_{k+1}}(\Q\YY_s)^{(j)}ds \cdot \indicator\{|\Delta_k Y^{(j)}| \leq v_N^{(j)}\} - \int_{t_{k}}^{t_{k+1}}(\Q\YY_s)^{(j)}ds \right) \probconv 0.$$
\end{lemma}
\begin{proof}
	See Appendix \ref{proof:lemma:lemma-mai-4.11-drift}.
\end{proof}

\subsubsection{Filtering the continuous martingale component}
To show that filtering out jumps is a valid option in the framework of Theorem \ref{th:discrete_clt_infinite}, we propose the lemma that follows:
\begin{lemma}{(Adapted from Lemma 4.10, \cite{mai2014efficient})\\}
\label{lemma:continuous-martingale-part-infinite}
In the framework of Theorem \ref{th:discrete_clt_infinite}, for $i,j\in\{1,\dots,d\}$ and as $N\rightarrow \infty$:
$$T_N^{-1/2}\sum_{k=0}^{N-1}Y^{(i)}_k\left(\Delta_kY^{(j)}\indicator\{|\Delta_kY^{(j)}| \leq v_N^{(j)}\}-\Delta_kY^{(j),c}\right) \probconv 0.$$
\end{lemma}
\begin{proof}
	See Appendix \ref{proof:lemma:continuous-martingale-part-infinite}.
\end{proof}

We have provided two collections of lemmas laying out the proofs for the central limit theorems for both the finite and infinite jump activities. They consist of incremental filtering and decompositions gradually linking the unfiltered to the filtered estimators. In the following section, we prove that the discretised unfiltered estimator is consistent as presented in Lemma \ref{lemma:cv-estimator-continuous-component}.

\section{Proofs for the finite activity case}
\label{appendix:proofs-finite}

In this technical section, we provide the proofs of the finite activity intermediate results presented in Appendix \ref{section:technical-lemmas}.

\subsection{Proof of Lemma \ref{lemma:consistency-conservation-jump-filtering-finite}}
\label{proof:lemma:consistency-conservation-jump-filtering-finite}

\begin{proof}[Proof of Lemma \ref{lemma:consistency-conservation-jump-filtering-finite}]
Observe that
\begin{align*}
    T_N^{1/2}(\boldsymbol{\widetilde{\psi}}_N - \boldsymbol{\overline{\psi}}_N) &= T_N \Sn_N^{-1}\left(\Abar_N - \Atilde_N\right).
\end{align*}
According to Lemma \ref{lemma:cts_filtered_approx}, for any $i,j \in\{1,\dots,d\}$, we have under $\proba_{0}$
$$T_N^{-1/2}\sum_{k=0}^{N-1} \E\left[\left| Y^{(i)}_k\left(\Delta_k Y^{(j)}_k \indicator_{\left\{|\Delta_k Y^{(j)}| \leq v_N^{(j)}\right\}} - \Delta_kY^{(j),c}_k\right)\right|\right] = O(\Delta_N^{1/2}T_N^{1/2}),$$
hence 
$$T_N^{-1/2}\sum_{k=0}^{N-1} Y^{(i)}_k\left(\Delta_k Y^{(j)}_k \indicator_{\left\{|\Delta_k Y^{(j)}| \leq v_N^{(j)}\right\}} - \Delta_k Y^{(j),c}_k\right) \probconv 0, \quad \text{as $N \rightarrow \infty$}.$$
Also, by ergodicity and the continuous mapping theorem, one has that $T_N \Sn_N^{-1} \xrightarrow{} \E\left(\YY_\infty \YY^\top_\infty\right)^{-1}\otimes \Id$ $\proba_{0}-a.s.$ which is finite componentwise and this concludes the proof.
\end{proof}

\section{Proofs for the infinite activity case}
\label{appendix:proofs-infinite}

In this technical section, we provide the proofs of the infinite activity intermediate results presented in Appendix \ref{section:infinite-appendix}.

\subsection{Proof of Lemma \ref{lemma:lemma-4.14-non-jump-neg}}
\label{proof:lemma:lemma-4.14-non-jump-neg}

	To prove this result, we introduce some new notation and define
$$\widetilde{B}^X_{N}(i,j) :=  T_N^{-1/2}\sum_{k=0}^{N-1} |Y^{(i)}_k||\Delta_k X^{(j)}| \indicator\{|\Delta_k W^{(j)} + \Delta_k D^{(j)}| > v_N^{(j)}\},$$
for $X\in\{Y, W,D,J^{1},J^{2}\}$ and $i,j\in\{1,\dots,d\}$. We prove the convergence in probability to zero of $\widetilde{B}^X_{N}(i,j)$ for each $X\in\{Y, W,D,J^{1},J^{2}\}$, starting with $\widetilde{B}^W_{N}(i,j)$ and $\widetilde{B}^D_{N}(i,j)$ in the section below.
%\end{proof}

\subsubsection{Intermediate results}
\label{section:intermediate-lemma-b-tildes}
We prove two results on the limiting behaviour of, respectively, $\widetilde{B}^W_N$ and $\widetilde{B}^D_N$ as follows
\begin{lemma}
	\label{lemma:lemma-4.14-term-in-W}
	Under the assumptions of Lemma \ref{lemma:lemma-4.14-non-jump-neg}, we have for $i,j\in\{1,\dots,d\}$
	$$\E \left(\left|\widetilde{B}^W_{N}(i,j)\right|\right) \leq O(T_N^{1/2}\Delta_N^{1/2- \beta^{(j)}}),\quad \text{as $N\rightarrow \infty$,}$$
such that $\widetilde{B}^W_{N}(i,j) \probconv 0$ as $N \rightarrow \infty$.
\end{lemma}
\begin{proof}[Proof of Lemma \ref{lemma:lemma-4.14-term-in-W}]
	$\widetilde{B}^W_N$ can be bounded using H{\"o}lder's inequality and secondly by independence between $Y^{(j)}_k$ and $\Delta_k W^{(j)}$:
\begin{align*}
    \E&\left(|\widetilde{B}^W_N(i,j)|\right)\\
    &\leq T_N^{-1/2}\sum_{k=0}^{N-1}\E\left(\left|Y^{(j)}_k \Delta_k W^{(j)}\right|^2\right)^{1/2}\E\left(\indicator_{\left\{|\Delta_k W^{(j)} + \Delta_k D^{(j)}| > v_N^{(j)}\right\}}\right)^{1/2}\\
    &\leq T_N^{-1/2}\sum_{k=0}^{N-1}\E\left(\left|Y^{(j)}_k\right|^2\right)^{1/2} \E\left( \left|\Delta_k W^{(j)}\right|^2\right)^{1/2}\proba\left(|\Delta_k W^{(j)} + \Delta_k D^{(j)}| > v_N^{(j)}\right)^{1/2}\\
    &\leq  \E\left(\left|Y^{(j)}_0\right|^2\right)^{1/2} \times T_N^{-1/2}\Delta_N^{1/2} N \Delta_N^{1-\beta^{(j)}},\quad  \text{by Lemma \ref{lemma:lemma-3.7} with $\delta = 1/2-\beta^{(j)}$ and $l^{(j)}=2$,}\\
    &\leq O( T_N^{1/2}\Delta_N^{1/2- \beta^{(j)}}), \quad  \text{since $N=O(T_N\Delta_N^{-1})$.}
\end{align*}
As we bound the $L^1$-norm of $\widetilde{B}^W_N(i,j)$, we obtain the convergence in probability to zero by Assumption \ref{assumption:cv-rate-beta}.
\end{proof}

\begin{lemma}
	\label{lemma:lemma-4.14-term-in-D-drift}
	Under the assumptions of Lemma \ref{lemma:lemma-4.14-non-jump-neg}, we have for $i,j\in\{1,\dots,d\}$
	$$\E \left(\left|\widetilde{B}^D_{N}(i,j)\right|\right)  \leq O(T_N^{1/2}\Delta_N^{1 - \beta^{(j)}}),\quad \text{as $N\rightarrow \infty$,}$$
such that $\widetilde{B}^W_{N}(i,j) \probconv 0$ as $N \rightarrow \infty$.
\end{lemma}
\begin{proof}[Proof of \ref{lemma:lemma-4.14-term-in-D-drift}]
	We prove this result similarly to Lemma \ref{lemma:lemma-4.14-term-in-W}. First, we \emph{recreate} increments as follows:
\begin{align*}
    \E & \left(\left|\widetilde{B}^D_N(i,j)\right|\right)\\
    &\leq T_N^{-1/2}\max_{l\in\{1,\dots,d\}}\left|Q_{jl}\right| \\
    &\qquad \times \left[ \sum_{l=1}^d \sum_{k=0}^{N-1} \E\left( \left|Y^{(i)}_k  \int_{t_k}^{t_{k+1}} (Y^{(l)}_s - Y^{(l)}_k)ds \right|\indicator\{|\Delta_kW^{(j)} + \Delta_k D^{(j)}| > v_N^{(j)}\} \right) \right.\\
      &\qquad \qquad \left.  + \sum_{l=1}^d \sum_{k=0}^{N-1}(t_{k+1}-t_k) \E\left( \left|Y^{(i)}_k Y^{(l)}_k \right| \indicator\{|\Delta_kW^{(j)} + \Delta_k D^{(j)}| > v_N^{(j)}\}    \right) \right] .
\end{align*}
	
%\begin{align*}
%    \E & \left\{|S^D_N(i,j)|\right\}\\
%    &\leq T_N^{-1/2}\max\left\{|Q_{jl}|:l\in\{1,\dots,d\}\right\} \\
%    &\qquad \times \left( \sum_{l=1}^d \sum_{k=0}^{N-1} \E\left\{ \left|(Y^{(i)}_k - Y^{(i)}_0) \int_{t_k}^{t_{k+1}} (Y^{(l)}_s - Y^{(l)}_k)ds \right|\indicator\{|\Delta_kW^{(j)} + \Delta_k D^{(j)}| > v_N^{(j)}\} \right\} \right.\\
%    &\qquad \qquad \left. + \sum_{l=1}^d \sum_{k=0}^{N-1} \E\left\{ \left| Y^{(i)}_0 \int_{t_k}^{t_{k+1}}(Y^{(l)}_s - Y^{(l)}_k) ds \right|\indicator\{|\Delta_k W^{(j)} + \Delta_k D^{(j)}| > v_N^{(j)}\}    \right\} \right.\\
%      &\qquad \qquad \left.  + \sum_{l=1}^d \sum_{k=0}^{N-1}(t_{k+1}-t_k) \E\left\{ \left|Y^{(i)}_0 Y^{(i)}_k \right| \indicator\{|\Delta_kW^{(j)} + \Delta_k D^{(j)}| > v_N^{(j)}\}    \right\} \right) .
%\end{align*}
%Remark that for $i,l\in\{1,\dots,d\}$ and $u \geq 0$, we have by stationarity that
%$$\E\left\{(Y^{(i)}_0)^2\right\}^{1/2} \sup_{s\in (t_k, t_k+u)}\E\left\{ \left(Y^{(l)}_s - Y^{(l)}_k\right)^2\right\}^{1/2} = \E\left\{(Y^{(i)}_0)^2\right\}^{1/2} \sup_{s\in (0, u)}\E\left\{ \left(Y^{(l)}_s - Y^{(l)}_0\right)^2\right\}^{1/2},$$
%and both sides are finite.
%Also, recall that $\YY_0$ is independent from $(\LL_t,\ t \geq 0 )$.

The first double sum is upper-bounded by H{\"o}lder's inequality and Fubini's theorem as follows:
\begin{align*}
	\E&\left( \left|Y^{(i)}_k  \int_{t_k}^{t_{k+1}} (Y^{(l)}_s - Y^{(l)}_k)ds \right|\indicator\{|\Delta_kW^{(j)} + \Delta_k D^{(j)}| > v_N^{(j)}\} \right) \\
	&\leq \E\left( \left|Y^{(i)}_k  \int_{t_k}^{t_{k+1}} (Y^{(l)}_s - Y^{(l)}_k)ds \right|^2\right)^{1/2} \proba\left(|\Delta_kW^{(j)} + \Delta_k D^{(j)}| > v_N^{(j)}\right)^{1/2}  \\
	&\leq \Delta_N^{1/2}\left( \int_{t_k}^{t_{k+1}} \E\left( \left|Y^{(i)}_k (Y^{(l)}_s - Y^{(l)}_k) \right|^2\right)ds\right)^{1/2} \Delta_N^{1-\beta^{(j)}}, \quad \text{since $v^{(i)}_N = \Delta^{\beta^{(i)}}_N$,}\\
	&\leq \Delta_N^{3/2-\beta^{(j)}}\left[ \int_{0}^{\Delta_N} \E\left( \left|Y^{(i)}_0 (Y^{(l)}_s - Y^{(l)}_0) \right|^2\right)ds\right]^{1/2}\\
	&\leq \Delta_N^{2-\beta^{(j)}} \sup_{s\in[0,\Delta_N]} \E\left( \left|Y^{(i)}_0 (Y^{(l)}_s - Y^{(l)}_0) \right|^2\right)^{1/2}.
\end{align*}
%\textbf{------------- IGNORE THIS -------------}\\
%Remark that as $t \rightarrow 0$:
%\begin{align*}
% 	\YY_t - \YY_0 &= \left(e^{-t\Q}-\Id\right)\YY_0  + \int_0^te^{-(t-u)\Q}d\LL_u\\
% 	&=(-t\Q+o(t^2))\YY_0 + e^{-t\Q}\int_0^te^{u\Q}d\LL_u.
% \end{align*}
%Since by Assumption \ref{assumption:infinite_activity} \eqref{assumption:infinite_activity:cross-moments} $\YY$ has finite second moments and $\YY_0$ is independent from $(\LL_t,\ t \geq 0)$, we have as $s\rightarrow 0$:
%\begin{equation}
%\label{eq:fourth_moments_as_second_moments}
%\begin{cases}
%\E\left\{\left|Y^{(i)}_0(Y^{(l)}_s-Y^{(l)}_0)\right|\right\} &= O(s), \\
% \E\left\{\left(Y^{(i)}_0(Y^{(j)}_s-Y^{(j)}_0)\right)^2\right\} &= O(s^2).
% \end{cases}
%\end{equation}
% 
% \textbf{--------------------------}\\
%$\YY_t = e^{-t\Q}\YY_0 + \int_{0}^{t}e^{-(t - s)\Q}d\LL_s$x
 Define
 \begin{equation*}
 		C_i(\Delta_N) := \max_{l\in\{1,\dots,d\}}\sup_{s\in(0,\Delta_N)} \E\left( \left|Y^{(i)}_0 (Y^{(l)}_s - Y^{(l)}_0) \right|^2\right)^{1/2} < \infty,
 \end{equation*}
 which is finite by Assumption \ref{assumption:infinite_activity}-\eqref{assumption:infinite_activity:cross-moments} such that $C_i(\Delta_N) = O(\Delta_N), \ \text{as $N\rightarrow \infty$.}$
Hence, we have:
\begin{align*}
    T_N^{-1/2} \sum_{l=1}^d &\sum_{k=0}^{N-1} \E\left(\left|Y^{(i)}_k \int_{t_k}^{t_{k+1}} (Y^{(l)}_s - Y^{(l)}_k)ds \right| \indicator\{|\Delta_k W^{(j)} + \Delta_k D^{(j)}| > v_N^{(j)}\} \right)\\
    &\leq C_i(\Delta_N)  T_N^{-1/2}N\Delta_N^{2-\beta^{(j)}} \\
    &\leq O(T_N^{1/2}\Delta_N^{1-\beta^{(j)}}), \quad \text{since $N\Delta_N=O(T_N)$ as $N \rightarrow \infty$.}
\end{align*}
%\begin{align*}
%    \E &\left\{\left| (Y^{(i)}_k-Y^{(i)}_0)  \int_{t_k}^{t_{k+1}}(Y^{(l)}_s - Y^{(l)}_k)ds\right| \indicator\left\{|\Delta_kW^{(j)} + \Delta_k D^{(j)}| > v_N^{(j)}\right\}\right\}\\
%     &\leq \E\left\{(Y^{(i)}_k-Y^{(i)}_0)^2\right\}^{1/2} \E\left\{\left(\int_{t_k}^{t_{k+1}}(Y^{(l)}_s - Y^{(l)}_k)ds\right)^2\right\}^{1/2} \proba\left\{|\Delta_kW^{(j)} + \Delta_k D^{(j)}| > v_N^{(j)}
%    \right\}^{1/2}\\
%    &\leq \E\left\{(Y^{(i)}_k-Y^{(i)}_0)^2\right\}^{1/2} (t_{k+1}-t_{k})^{1/2}\sup_{s\in (t_k,t_{k+1})}\E\left\{\left(Y^{(l)}_s - Y^{(l)}_k\right)^2\right\}^{1/2} \Delta_N^{1-\beta^{(j)}}\\
%    &\leq \E\left\{(Y^{(i)}_k-Y^{(i)}_0)^2\right\}^{1/2} \sup_{s\in (t_k, t_{k+1})}\E\left\{\left(Y^{(l)}_s - Y^{(l)}_k\right)^2\right\}^{1/2} \Delta_N^{3/2-\beta^{(j)}}\\
%        &\leq \E\left\{(Y^{(i)}_k)^2\right\}^{1/2}\E\left\{(Y^{(l)}_k)^2\right\}^{1/2} \Delta_N^{3/2-\beta^{(j)}}, \quad \text{by stationarity.}
% \end{align*}
%Hence, we have:
%\begin{align*}
%    T_N^{-1/2} \sum_{l=1}^d &\sum_{k=0}^{N-1} \E\left\{\left|(Y^{(i)}_k-Y^{(i)}_0) \int_{t_k}^{t_{k+1}} (Y^{(l)}_s - Y^{(l)}_k)ds \right| \indicator\{|\Delta_k W^{(j)} + \Delta_k D^{(j)}| > v_N^{(j)}\} \right\}\\
%    &\leq \left[\sum_{l=1}^d \E\left\{(Y^{(i)}_k)^2\right\}^{1/2}\E\left\{(Y^{(l)}_k)^2\right\}^{1/2} \right]\times T_N^{-1/2} \Delta_N^{3/2-\beta^{(j)}}N\\
%    &\leq O\left(T_N^{1/2}\Delta_N^{1-\beta^{(j)}}\right)   = O\left(T_N^{1/2}\Delta_N^{1/2-\beta^{(j)}}\right).
%\end{align*}
The second term can be processed in a similar fashion. Using H{\"o}lder's inequality and Assumption \ref{assumption:infinite_activity} \eqref{assumption:infinite_activity:cross-moments}, we obtain:
\begin{align*}
     T_N^{-1/2}&\sum_{l=1}^d \sum_{k=0}^{N-1}(t_{k+1}-t_k)\cdot \E\left( \left|Y^{(i)}_k Y^{(l)}_k\right| \indicator\{|\Delta_kW^{(j)} + \Delta_k D^{(j)}| > v_N^{(j)}\}    \right)\\
     &\leq  T_N^{-1/2}\Delta_N \sum_{l=1}^d \sum_{k=0}^{N-1} \E\left(\left|Y^{(i)}_k Y^{(l)}_k\right|^2\right)^{1/2} \proba\left(|\Delta_kW^{(j)} + \Delta_k D^{(j)}| > v_N^{(j)}\right)^{1/2}\\
     &\leq \max_{l\in\{1,\dots,d\}}E\left(\left|Y^{(i)}_0 Y^{(l)}_0\right|^2\right)^{1/2} \times T_N^{-1/2}N\Delta_N^{2-\beta^{(j)}},\\
     &\leq O(T_N^{1/2} \Delta_N^{1-\beta^{(j)}})
\end{align*}
Therefore, we have that:
$$ \E\left(\left|\widetilde{B}^D_N\right|\right) \leq  O(T_N^{1/2}\Delta_N^{1-\beta^{(j)}}).$$
Hence, as we bound the $L^1$-norm of $\widetilde{B}^D_N(i,j)$, we obtain the convergence in probability to zero by Assumption \ref{assumption:cv-rate-beta}.
\end{proof}
\subsubsection{Main argument}
\label{section:intermediate-main-argument}
We now formally prove Lemma \ref{lemma:lemma-4.14-non-jump-neg}.
\begin{proof}[Proof of Lemma \ref{lemma:lemma-4.14-non-jump-neg}]
 Recall that $\Delta_k Y^{(j)} = \Delta_k W^{(j)} + \Delta_k D^{(j)} + \Delta_k J^{(j)}$ and that $v_N^{(j)} = \Delta_N^{\beta^{(j)}}$. 
% According to Lemma \ref{lemma:lemma-3.7} with $\delta = 1/2 - \beta^{(j)}$ and $l^{(i)}=2$, we obtain:
%$$\proba\left(|\Delta_k W^{(j)} + \Delta_k D^{(j)}| > v_N^{(j)}\right) = O\left(\Delta_N^{2-2\beta^{(j)}}\right).$$
By the triangle inequality, we then split the quantity of interest $\widetilde{B}^Y_{N}(i,j)$ in four different parts
\begin{align*}
    \widetilde{B}^Y_N(i,j) \leq \widetilde{B}^W_N(i,j) + \widetilde{B}^D_N(i,j) + \widetilde{B}^{J^{1}}_N(i,j) + \widetilde{B}^{J^{2}}_N(i,j).
\end{align*}
First, Lemmas \ref{lemma:lemma-4.14-term-in-W} and \ref{lemma:lemma-4.14-term-in-D-drift} yield that
$$ \begin{cases}
 	\E\left(|\widetilde{B}^W_N|\right) &\leq  \: \ O(T_N^{1/2}\Delta_N^{1/2-\beta^{(j)}}),\\
 	\E\left(|\widetilde{B}^D_N|\right) &\leq  \: \ O(T_N^{1/2}\Delta_N^{1-\beta^{(j)}}).
 \end{cases}
$$
Finally, remark that by independence and again by H{\"older}'s inequality, $\widetilde{B}^{J^{1}}_N(i,j)$ and $\widetilde{B}^{J^{2}}_N(i,j)$ can be treated similarly to $\widetilde{B}^W_N(i,j)$. All four part are then dominated by $T_N^{1/2} \Delta_N^{1/2-\beta^{(j)}}$ which converges to $0$ as $N \rightarrow \infty$ since $T_N \Delta_N^{(1-2\beta^{(j)}) \wedge 1/2}=o(1)$ by Assumption \ref{assumption:cv-rate-beta}. We conclude that $\widetilde{B}^Y_N(i,j) \probconv 0$ for any $i,j \in\{1,\dots,d\}$.
\end{proof}

\subsection{Proof of Lemma \ref{lemma:lemma-4.13-go-big-jump}}
\label{proof:lemma:lemma-4.13-go-big-jump}

\begin{proof}[Proof of Lemma \ref{lemma:lemma-4.13-go-big-jump}]
The key idea of the proof is to prove that the indicator $\indicator\{|\Delta_kY^{(j)}| \leq v_N^{(j)}, |\Delta_k J^{(j)}| > 2v_N^{(j)}\}$ is triggered at most as much as $\indicator\{|\Delta_k W^{(j)}+ \Delta_k D^{(j)}| \leq v_N^{(j)}\}$. Recall that $\Delta_k Y^{(j)} = \Delta_k W^{(j)} + \Delta_k D^{(j)} + \Delta_k J^{(j)}$. Hence, by the triangle inequality, on $\{|\Delta_kY^{(j)}| \leq v_N^{(j)}, |\Delta_k J^{(j)}| > 2v_N^{(j)}\}$ we have
\begin{equation}
	\label{eq:triangle-inverted}
	\left||\Delta_k W^{(j)} + \Delta_k D^{(j)}| - |\Delta_k J^{(j)}|\right| \leq |\Delta_k Y^{(j)}| \leq v_N^{(j)},
\end{equation}
by adding and subtracting $\Delta_k Y^{(j)}$ in the first term.
That implies that, as given in \cite{mai2014efficient}:
\begin{equation}
\label{eq:proof-lemma-3.7-inclusion}
\{|\Delta_kY^{(j)}| \leq v_N^{(j)}, |\Delta_k J^{(j)}| > 2v_N^{(j)}\} \subseteq \{|\Delta_k W^{(j)} + \Delta_k D^{(j)}| > v_N^{(j)}\}.  
\end{equation}
Indeed, if that were not the case, then one would have $|\Delta_k W^{(j)} + \Delta_k D^{(j)}| - |\Delta_k J^{(j)}| < -v_N^{(i)}$ leading to a contradiction in Equation \eqref{eq:triangle-inverted} when taking the absolute value.
Therefore,
$$ \indicator\{|\Delta_kY^{(j)}| \leq v_N^{(j)}, |\Delta_k J^{(j)}| > 2v_N^{(j)}\} \leq\indicator \{|\Delta_k W^{(j)} + \Delta_k D^{(j)}| > v_N^{(j)}\},$$
and Lemma \ref{lemma:lemma-4.14-non-jump-neg} allows to conclude directly.
\end{proof}

%
%\subsection{Proof of Lemma \ref{lemma:small_jump_negligeable}}
%\label{proof:lemma:small_jump_negligeable}
%
%\begin{proof}[Proof of Lemma \ref{lemma:small_jump_negligeable}]
% This proof was given on page 949, \cite{mai2014efficient} and is given for completeness. Recall that $v_N^{(j)} = \Delta_N^{\beta^{(j)}}$. Then, Proposition 4.8, \cite{mai2014efficient} yields that for any $j\in\{1,\dots,d\}$:
%$$\E\left(|\Delta_k J^{2,(j)}|^2\indicator\{|\Delta_k J^{2,(j)}| \leq 2v^{(j)}_N\} \right)=O\left(\Delta_{N}^{1+\beta^{(i)}(2-\alpha)}\right).$$
% Using Markov's inequality:
% \begin{align*}
%    \proba & \left( v^{(j)}_N/2 < |\Delta_k J^{2,(j)}| \leq 2v^{(j)}_N \right) \\
%    &= \proba\left(|\Delta_k J^{2,(j)}|\indicator \{|\Delta_k J^{2,(j)}| \leq 2v^{(j)}_N \} > v^{(j)}_N/2 \right)\\
%    &\leq 4\left(v^{(j)}_N\right)^{-2} \E\left(|\Delta_k J^{2,(j)}|^2\indicator\{|\Delta_k J^{2,(j)}| \leq 2v^{(j)}_N\} \right)\\
%    &= O\left(\Delta_N \Delta_N^{-2\beta^{(j)} + \beta^{(j)}(2-\alpha)}\right), \qquad \text{by Assumption \ref{assumption:infinite_activity}-\eqref{assumption:infinite_activity:jumps_2_epsilon},}\\
%    &= O\left(\Delta_N^{1-\alpha \beta^{(j)}}\right).
%\end{align*}
%\end{proof}

\subsection{Proof of Lemma \ref{lemma:lemma-4.15-small-jump-sandwich}}
\label{proof:lemma:lemma-4.15-small-jump-sandwich}

\begin{proof}[Proof of Lemma \ref{lemma:lemma-4.15-small-jump-sandwich}]
Recall that by Lemma \ref{lemma:small_jump_negligeable}, we know that 
$$\proba\left( v^{(j)}_N/2 < |\Delta_k J^{2,(j)}| \leq 2v^{(j)}_N \right) = O(\Delta_N^{1-\alpha \beta^{(j)}}).$$
Therefore, since $\beta^{(j)} \in (0,1/2)$ and $\alpha \in (0,2)$:
$$\proba\left(v^{(j)}_N/2 < |\Delta_k J^{2,(j)}| \leq 2v^{(j)}_N \right) \longrightarrow 0, \quad \text{as $N \rightarrow \infty$.}$$
For $X\in\{Y,D, W,J^{1}, J^{2}\}$, define for any $i,j\in\{1,\dots,d\}$:
$$\widetilde{U}^X_N(i,j) := T_N^{-1/2}\sum_{k=0}^{N-1}Y^{(i)}\Delta_k X^{(j)}\indicator\{v_N^{(j)}/2 < |\Delta_k J^{2}| \leq 2v^{(j)}_N\}.$$
Again, recall that $\Delta_k Y^{(j)}$ as $\Delta_k D^{(j)} + \Delta_k W^{(j)} + \Delta_k J^{1,(j)} +\Delta_k J^{2,(j)}$ such that we focus on four different terms:
\begin{align*}
    \widetilde{U}^Y_N(i,j)&= \widetilde{U}^{D}_N(i,j) + \widetilde{U}^{W}_N(i,j) + \widetilde{U}^{J^{1}}_N(i,j) + \widetilde{U}^{J^{2}}_N(i,j).
\end{align*}
%Recall that by Lemma \ref{lemma:small_jump_negligeable}:
%\begin{align*}
%    \proba\left(|\Delta_k J^{2,(j)}| \right. &\leq \left. v^{(j)}_N/2 < |\Delta_k J^{2,(j)}| \leq 2v^{(j)}_N \right) =  O\left(\Delta_N^{1-\alpha \beta^{(j)}}\right).
%\end{align*}

By the independence between $Y^{(i)}_k$, $\Delta_k W^{(j)}$ and $\Delta_k J^{2,(j)}$, one obtains that
$$\E\left(\widetilde{U}^W_N(i,j)\right) = \sum_{k=0}^{N-1} \E\left(Y^{(i)}_k\right)\E\left(\Delta_k W^{(j)}\right)\proba\left(v^{(j)}_N/2 < |\Delta_k J^{2,(j)}| \leq 2v^{(j)}_N \right) = 0.$$
Again, by independence and by H{\"o}lder's inequality, one has:
\begin{align*}
    \E\Big[&\left.\left(\widetilde{U}^W_N(i,j)\right)^2\right]\\
    &\leq 2T_N^{-1}\sum_{k=0}^{N-1}\E\left(\left|Y^{(i)}_k\right|^2\right)^{1/2}\E\left(\left|\Delta_k W^{(i)}\right|^2\right)^{1/2}\proba\left(v^{(j)}_N/2 < |\Delta_k J^{2,(j)}| \leq 2v^{(j)}_N \right)^{1/2}\\
    &\leq O\left(T_N^{-1} N \Delta_N \Delta_N^{(1-\alpha\beta^{(j)})/2}\right), \quad \text{by stationarity and finite second moments,}\\
    &= O\left(\Delta_N^{(1-\alpha\beta^{(j)})/2}\right), \quad \text{as per Section \ref{section:high-frequency-setup}},
\end{align*}
which yields convergence in probability to zero for $U^W_N(i,j)$. 

For the drift element $\widetilde{U}^D_N(i,j)$, we have
\begin{align*}
    \E & \left(|\widetilde{U}^D_N(i,j)|\right)\\
    &\leq T_N^{-1/2}\max_{l\in\{1,\dots,d\}}|Q_{jl}| \\
    &\qquad \times \left[ \sum_{l=1}^d \sum_{k=0}^{N-1} \E\left( \left|Y^{(i)}_k  \int_{t_k}^{t_{k+1}} (Y^{(l)}_s - Y^{(l)}_k)ds \right|\indicator\left\{v^{(j)}_N/2 < |\Delta_k J^{2,(j)}| \leq 2v^{(j)}_N \right\} \right) \right.\\
      &\qquad \qquad \left.  + \sum_{l=1}^d \sum_{k=0}^{N-1}(t_{k+1}-t_k) \E\left( \left| Y^{(i)}_k Y^{(l)}_k \right| \indicator\left\{ v^{(j)}_N/2 < |\Delta_k J^{2,(j)}| \leq 2v^{(j)}_N \right\}  \right) \right] .
\end{align*}
For the first term on the right-hand side, we repeat a similar argument to the proof of Lemma \ref{lemma:lemma-4.14-term-in-D-drift}, presented in \ref{proof:lemma:lemma-4.14-non-jump-neg}: by Assumption \ref{assumption:infinite_activity}-\eqref{assumption:infinite_activity:cross-moments} and H{\"o}lder's inequality, we obtain that
$$\E\left( \left|Y^{(i)}_k  \int_{t_k}^{t_{k+1}} (Y^{(l)}_s - Y^{(l)}_k)ds \right|\indicator\left\{v^{(j)}_N/2 < |\Delta_k J^{2,(j)}| \leq 2v^{(j)}_N \right\} \right) \leq O(\Delta_N^{2-\alpha \beta^{(j)}}).
$$
Next, for the second term, the independence between $Y^{(i)}_k$ and $\Delta_k J^{2,(j)}$ yields that
$$\sum_{l=1}^d \sum_{k=0}^{N-1}(t_{k+1}-t_k) \E\left( \left|Y^{(i)}_k Y^{(l)}_k \right| \indicator\left\{ v^{(j)}_N/2 < |\Delta_k J^{2,(j)}| \leq 2v^{(j)}_N \right\}\right) \leq O(N\Delta_N^{2-\alpha\beta^{(j)}}).$$
Since $N\Delta_N = O(T_N)$, we have $\widetilde{U}^D_N(i,j) \leq O(T_N^{1/2}\Delta_N^{1-\alpha \beta^{(j)}})$.

The fourth element $\widetilde{U}_N^{J^{2}}(i,j)$ is bounded using Assumption \ref{assumption:infinite_activity}-\eqref{assumption:infinite_activity:jumps_2_epsilon} in the sense that for some $\epsilon_0 > 0$, for $N$ large enough $v_N^{(j)} = \Delta_N^{\beta^{(j)}} < \epsilon_0$, therefore
$$\E\left(\Delta_k\JJ^{2,(i)}\indicator\{|\Delta_k\JJ^{2,(i)}| < 2v_N^{(j)}\}\right) = 0.$$
Then, the independence between $Y^{(i)}_k$ and $\Delta_k\JJ^{2,(i)}\indicator\{|\Delta_k\JJ^{2,(i)}| < 2v_N^{(j)}\}$ allows to conclude that 
$$\E\left(\widetilde{U}^{J^{2}}_N\right)=0.$$
 Also, recall that Lemma \ref{lemma:small_jump_negligeable} gives that $\E\left(|\Delta_k J^{2,(j)}|^2\indicator\{|\Delta_k J^{2,(j)}| \leq 2v^{(j)}_N\} \right) = O(\Delta_N^{1 +  \beta^{(j)}(2-\alpha)})$. Therefore, we obtain:
\begin{align*}
    \E\left(\left|\widetilde{U}^{J^2}_N(i,j)\right|^2\right)&\leq 2T_N^{-1}\sum_{k=0}^{N-1} \E\left(\left|Y^{(i)}_k\right|^2\right)\E\left(|\Delta_k J^{2,(j)}|^2\indicator\{|\Delta_k J^{2,(j)}| \leq 2v^{(j)}_N\} \right)\\
    &\leq O\left(T_N^{-1}N\Delta_N^{1 + \beta^{(j)}(2-\alpha)}\right)\\
    &= O(\Delta_N^{\beta^{(j)}(2-\alpha)}).
\end{align*}
Hence, by Markov's inequality and since $\alpha \in (0,2)$, $\widetilde{U}^{J^{2}}_N(i,j)$ converges to zero in probability as $N \rightarrow \infty$. 

Finally, for the third term $\widetilde{U}^{J^{1}}_N$ we use the independence between $Y_k^{(i)}$ and $\Delta_k J^{1,(j)}$. By definition $\E\left(|\Delta_k J^{1,(i)}|\right) = O(\Delta_N)$, thus:
\begin{align*}
    \E\left(|\widetilde{U}^{J^1}_N|\right) &= T_N^{-1/2}\sum_{k=0}^{N-1}  \E\left(|Y^{(i)}_k|\right) \E\left(|\Delta_k J^{1,(j)}|\right) \proba\left(v^{(j)}_N/2 < |\Delta_k J^{2,(j)}| \leq 2v^{(j)}_N\right)\\
    &= O\left(T_N^{-1/2} N \Delta_N \Delta_N^{1-\alpha \beta^{(j)}}\right)\\
    &= O\left(T_N^{1/2} \Delta_N^{1-\alpha \beta^{(j)}}\right).
\end{align*}
Since $T_N^{1/2}\Delta_N^{1-\alpha\beta^{(i)}} \leq T_N\Delta_N^{1-2\beta^{(i)}} = o(1)$ as $N \rightarrow \infty$, all four terms converge to zero in $L^1$ hence in probability as $N \rightarrow \infty$ which concludes the proof.
\end{proof}
\subsection{Proof of Lemma \ref{lemma:small-enough-jumps}}
\label{proof:lemma:small-enough-jumps}

\begin{proof}[Proof of Lemma \ref{lemma:small-enough-jumps}]
We prove that
$$T_N^{-1/2}\sum_{k=0}^{N-1}  Y^{(i)}_k
     \Delta_k J^{(j)}\indicator\{|\Delta_kJ^{(j)}|\leq 2v_N^{(j)}, |\Delta_k J^{1,(j)}| \neq  0\} \xrightarrow{\ L^1\ } 0, \ \text{as $N\rightarrow \infty$.}$$
Since $|\Delta_k J^{(j)}| \leq |\Delta_k J^{1,(j)}| + |\Delta_k J^{2,(j)}|$, the quantity above is split into two terms. Also, recall that $\JJ^1$ is a compound Poisson process with finite intensity such that $\proba\left(|\Delta_k J^{1,(j)}| \neq  0\right)=O(\Delta_N)$ (cf.\ page 943, \cite{mai2014efficient}).

Using H{\"o}lder's inequality and the independence between $Y^{(i)}_k$, $\Delta_k J^{1,(j)}$ and $\Delta_k J^{2,(j)}$, the part with respect to $\Delta_k J^{1,(j)}$ is upper-bounded as follows:
\begin{align*}
T_N^{-1/2}&\sum_{k=0}^{N-1} \E\left( \left|Y^{(i)}_k\right| \right)
    \E\left(\left| \Delta_k J^{1,(j)}\right|\indicator\{|\Delta_kJ^{(j)}|\leq 2v_N^{(j)}, |\Delta_k J^{1,(j)}| \neq  0\}\right) \\
    &\leq T_N^{-1/2}\sum_{k=0}^{N-1} \E\left( \left|Y^{(i)}_k\right| \right)
    \E\left(\left| \Delta_k J^{1,(j)}\right|^2 \right)^{1/2} \proba\left(|\Delta_k J^{1,(j)}| \neq  0\right)^{1/2}\\
     &\leq T_N^{-1/2} \E\left( \left|Y^{(i)}_0\right| \right)
    \sum_{k=0}^{N-1} \left( \Delta_N^{2}\right)^{1/2} O(\Delta_N^{1/2}) \\
%    &\leq O\left(T_N^{-1/2} N \Delta_N^{\frac{3}{2} + \frac{\beta^{(j)}}{2}\left(2-\alpha\right)}\right) \\
	&\leq O(T_N^{-1/2}N\Delta_N^{3/2})\\
        &\leq  O(T_N^{1/2}\Delta_N^{1/2}),
\end{align*}
where we used the assumption $N \Delta_N = O(T_N)$ as $N\rightarrow \infty$. For the term with respect to $\Delta_k J^{2}$, we proceed using Lemma \ref{lemma:small_jump_negligeable}:
\begin{align*}
	T_N^{-1/2}&\sum_{k=0}^{N-1} \E\left( \left|Y^{(i)}_k\right| \right)
    \E\left(\left| \Delta_k J^{2,(j)}\right|\indicator\{|\Delta_k J^{(j)}|\leq 2v_N^{(j)}, |\Delta_k J^{1,(j)}| \neq  0\}\right)\\
    &\leq T_N^{-1/2}\sum_{k=0}^{N-1} \E\left( \left|Y^{(i)}_k\right| \right)
    \E\left(\left| \Delta_k J^{2,(j)}\right|^2\indicator\{|\Delta_kJ^{2,(j)}|\leq 2v_N^{(j)}\}\right)^{1/2}\proba\left( |\Delta_k J^{1,(j)}| \neq  0\right)^{1/2}\\
    &\leq O\left(T_N^{-1/2} N \Delta_N^{\frac{3}{2} + \beta^{(j)}(1-\frac{\alpha}{2})}\right)\\
    &\leq O(T_N^{1/2}\Delta_N^{1/2}),
\end{align*}

%QQQQQQQQQQQQ
%For the second term, for a set $A\subseteq [-1,1]^d\symbol{92}\{\boldsymbol{0}_d\}$ such that $\int_A\nu(dy)<\infty$ it is known that the compensated compound Poisson process $\int_A (J^{2,(i)}_{t+\delta_t}-J^{2,(i)}_t)(dy) =  \int_A \int_{t}^{t+\delta_t} J^{2,(i)}(ds,dy)$ is a martingale for some $\delta_t > 0$.
%Thus, $$\E\left(|\Delta_k J^{2,(i)}|\indicator\{\Delta_k J^{2,(i)} \in A\}\right) = O(\Delta_N^{1/2}).$$ Therefore, we ought to consider small jumps that are gradually smaller as N increases, e.g. on $[ -2v_N^{(j)}, 2v_N^{(j)}]$ we obtain for any small $\epsilon >0$
%$$ \E\left\{ |\Delta_k J^{2,(i)}|\indicator\{\epsilon \leq |\Delta_k J^{2,(j)}|\leq 2 v_N^{(j)} \} \right\} = O(\Delta_N^{1/2}).$$ 
%In particular, we take $\epsilon = \epsilon_N = v^{(j)}_N/2$ and we conclude that
%\begin{align*}
%	T_N^{-1/2} \sum_{k=0}^{N-1} & \E\left\{ \left|Y^{(i)}_k\right| \right\} \E\left\{\left| \Delta_k J^{2,(j)}\right|\indicator\{|\Delta_k J|\leq 2v_N^{(j)}, |\Delta_k J^{1,(j)}| \neq  0\}\right\}\\
%	&\leq T_N^{-1/2}\sum_{k=0}^{N-1} \E\left\{ \left|Y^{(i)}_k\right| \right\}  \E\left\{\left| \Delta_k J^{2,(j)}\right|\indicator\{v_N^{(j)}/2 \leq |\Delta_k J^{2,(j)}| \leq 2v_N^{(j)}\}%
%	\right\} \proba\left\{|\Delta_k J^{1,(j)}| \neq  0\right\}\\
%	&\quad + T_N^{-1/2}\sum_{k=0}^{N-1} \E\left\{ \left|Y^{(i)}_k\right| \right\}  \E\left\{\left| \Delta_k J^{2,(j)}\right|\indicator\{|\Delta_k J^{2,(j)}| \leq v_N^{(j)}/2\}%
%	\right\} \proba\left\{|\Delta_k J^{1,(j)}| \neq  0\right\}\\
%	&=  O(T_N^{-1/2} N \Delta_N^{3/2}) = O( T_N^{1/2} \Delta_N^{1/2}),
%\end{align*}

which concludes the proof.
\end{proof}

\subsection{Proof of Lemma \ref{lemma:lemma-4.12-supp-large-jump-negligible}}
\label{proof:lemma:lemma-4.12-supp-large-jump-negligible}

\begin{proof}[Proof of Lemma \ref{lemma:lemma-4.12-supp-large-jump-negligible}]
We prove that the pure jump part $\JJ = \JJ^1 + \JJ^2$ is dominated by small jumps in the sense that the $\JJ^1$-jump contribution is asymptotically zero for filtered jumps as follows
$$T_N^{-1}\E\left(\left|\sum_{k=0}^{N-1}Y^{(i)}_k \Delta_k Y^{(j)} \indicator\{|\Delta_kJ^{(j)}|\leq 2v_N^{(j)}, |\Delta_k J^{1,(j)}|\neq 0\}\right|\right) \rightarrow 0, \quad \text{as $N \rightarrow \infty$.}$$
Again, using $\Delta_k Y^{(j)} = \Delta_k D^{(j)} + \Delta_k W^{(j)} + \Delta_k J^{(j)}$, one has:
\begin{align*}
    \E &\Bigg(\Big|\sum_{k=0}^{N-1}Y^{(i)}_k \Delta_k Y^{(j)} \indicator\{|\Delta_kJ^{(j)}|\leq 2v_N^{(j)}, |\Delta_k J^{1,(j)}\neq 0|\}\Big|\Bigg)\\
    &\leq \sum_{k=0}^{N-1} \E\left( \left|Y^{(i)}_k\right| \right)
    \E\left(\left| \Delta_k W^{(j)}\right|\right) \proba\left(|\Delta_k J^{1,(j)}| \neq  0\right)\\
    &\quad + \sum_{k=0}^{N-1} \E\left( \left|Y^{(i)}_k\Delta_k D^{(j)}\right|\right) \indicator\left(|\Delta_k J^{1,(j)}| \neq  0\right)\\
    &\quad + \sum_{k=0}^{N-1} \E\left( \left|Y^{(i)}_k\right| \right)
    \E\left(\left| \Delta_k J^{(j)}\right|\indicator\{|\Delta_kJ^{(j)}|\leq 2v_N^{(j)}, |\Delta_k J^{1,(j)}| \neq  0\}\right).
\end{align*}
As remarked on page 943, \cite{mai2014efficient}, $J^{1,(j)}$ is a compound Poisson with a finite intensity given by $\nu(\R \symbol{92} [-1,1]) < \infty$ such that one can write $\proba\left\{|\Delta_k J^{1,(j)}| \neq  0\right\} = O(\Delta_N)$. This implies that 
$$T_N^{-1/2}\sum_{k=0}^{N-1} \E\left( \left|Y^{(i)}_k\right| \right)
    \E\left(\left| \Delta_k W^{(j)}\right|\right) \proba\left(|\Delta_k J^{1,(j)}| \neq  0\right) = O(T_N^{1/2}\Delta_N^{1/2}) \xrightarrow{N \rightarrow \infty} 0.$$
The second item can be processed in a similar way to the term $\widetilde{U}^D_N(i,j)$ in the proof of Lemma \ref{lemma:lemma-4.14-non-jump-neg} (as in Lemma \ref{lemma:lemma-4.14-term-in-D-drift}, Appendix \ref{proof:lemma:lemma-4.14-non-jump-neg}). Hence
$$T_N^{-1/2}\sum_{k=0}^{N-1} \E\left( \left|Y^{(i)}_k\Delta_k D^{(j)}\right|\indicator\{|\Delta_k J^{1,(j)}| \neq  0\}\right) = O(T_N^{-1/2}N\Delta_N^2) = O(T_N^{1/2}\Delta_N),$$
since $P\left\{|\Delta_k J^{1,(j)}| \neq  0\right\} = O(\Delta_N)$. Finally, by Lemma \ref{lemma:small-enough-jumps}, the third term is asymptotically zero w.r.t.\ to the $L^1$-norm as $N\rightarrow \infty$. Therefore:
$$T_N^{-1/2} \E\left\{\left|\sum_{k=0}^{N-1}Y^{(i)}_k \Delta_k Y^{(j)} \indicator\{|\Delta_k J^{(j)}|\leq 2v_N^{(j)}, |\Delta_k J^{1,(j)}| \neq 0\} \right|\right\} \rightarrow 0, \quad \text{as $N \rightarrow \infty$.}$$	
\end{proof}

\subsection{Proof of Lemma \ref{lemma:lemma-4.12-small-inc-are-small-jumps}}
\label{proof:lemma:lemma-4.12-small-inc-are-small-jumps}

 \begin{proof}[Proof of Lemma \ref{lemma:lemma-4.12-small-inc-are-small-jumps}]
Consider the whole pure jump process $J^{(j)}$ instead of $J^{2,(j)}$. We split this quantity of interest into two parts as follows:
\begin{align*}
    T_N^{-1/2}&\sum_{k=0}^{N-1} Y^{(i)}_k \Delta_k Y^{(j)} \left(\indicator\{|\Delta_kY^{(j)}| \leq v_N^{(j)}\} - \indicator\{|\Delta_k J^{(j)}| \leq 2v_N^{(j)}\}\right) \\
    &=  T_N^{-1/2}\sum_{k=0}^{N-1} Y^{(i)}_k \Delta_k Y^{(j)} \left(\indicator\{|\Delta_kY^{(j)}| \leq v_N^{(j)}, |\Delta_k J^{(j)}| > 2v_N^{(j)}\}\right) \\
    &\quad -  T_N^{-1/2}\sum_{k=0}^{N-1} Y^{(i)}_k \Delta_k Y^{(j)} \left(\indicator\{|\Delta_kY^{(j)}| > v_N^{(j)}, |\Delta_k J^{(j)}| \leq 2v_N^{(j)}\}\right).
\end{align*}
We note that the first term is exactly the quantity converging in probability to zero from Lemma \ref{lemma:lemma-4.13-go-big-jump} as $N \rightarrow \infty$. By Lemma \ref{lemma:lemma-4.12-supp-large-jump-negligible}, we know that
\begin{equation}
	\label{eq:neq-co-integrated}
\end{equation}
$$T_N^{-1/2}\sum_{k=0}^{N-1} Y^{(i)}_k
  \Delta_k Y^{(j)}\indicator\left\{|\Delta_kY^{(j)}|>v_N^{(j)}, |\Delta_k J^{(j)}| \leq 2v^{(j)}_N, |\Delta_k J^{1,(j)}| \neq 0\right\} \probconv 0, \ \text{as $N\rightarrow \infty$.}$$
Next, we show that
$$ T_N^{-1/2} \sum_{k=0}^{N-1} Y^{(i)}_k \Delta_k Y^{(j)} \indicator\{ |\Delta_kY^{(j)}|>v_N^{(j)}, |\Delta_k J^{(j)}| \leq 2v^{(j)}_N, |\Delta_k J^{1,(j)}| = 0\} \xrightarrow{\ p\ } 0.$$
Indeed, observe that if $\Delta J^{1,(j)} = 0$, then $\Delta_k J^{(j)} = \Delta_k J^{2,(j)}$; thus 
\begin{align*}
 \indicator\{&|\Delta_kY^{(j)}|> v_N^{(j)}, |\Delta_k J^{(j)}| \leq 2v^{(j)}_N, |\Delta_k J^{1,(j)}|=0\}\\
 &= \indicator\{|\Delta_kY^{(j)}|>v_N^{(j)}, |\Delta_k J^{2,(j)}| \leq 2v^{(j)}_N, |\Delta_k J^{1,(j)}|=0\}.   
\end{align*}
Similarly to the proof of Lemma 4.12, \cite{mai2014efficient}, we split the event in the indicator function as follows:
\begin{align*}
\{&|\Delta_kY^{(j)}|>v_N^{(j)}, |\Delta_k J^{2,(j)}| \leq 2v^{(j)}_N, |\Delta_k J^{1,(j)}|=0\}\\
&\subset \{|\Delta_kW^{(j)} + \Delta_kD^{(j)}|>v_N^{(j)}/2\} \cup \{v^{(j)}_N/2 < |\Delta_k J^{2,(j)}| \leq 2v^{(j)}_N\},
\end{align*}
which redirects the convergence proof to the results of Lemmas \ref{lemma:lemma-4.14-non-jump-neg} and \ref{lemma:lemma-4.15-small-jump-sandwich}. This inclusion is true since either we have $|\Delta_k J^{2,(j)}| < v_N^{(j)}/2$ and in this case
$$|\Delta_k W^{(j)} + \Delta_k D^{(j)}| \geq \left| |\Delta_k Y^{(i)}| - |\Delta_k J^{2,(i)}|\right| > v_N^{(j)}/2,$$
or we have that $v_N^{(j)}/2 < \Delta_k J^{2,(j)} \leq 2v_N^{(j)}$.
Given this convergence result and \eqref{eq:neq-co-integrated}, we obtain that
\begin{align*}
 T_N^{-1/2}\sum_{k=0}^{N-1} Y^{(i)}_k\Delta_kY^{(j)} \left(\indicator\{|\Delta_kY^{(j)}| \leq v_N^{(j)}\} - \indicator\{|\Delta_k J^{2,(j)}| \leq 2v_N^{(j)}\}\right)
  \probconv 0, \quad \text{as $N\rightarrow\infty$,}
\end{align*}
which proves the desired result.
\end{proof}

\subsection{Proof of Lemma \ref{lemma:lemma-mai-4.11-drift}}
\label{proof:lemma:lemma-mai-4.11-drift}

\begin{proof}[Proof of Lemma \ref{lemma:lemma-mai-4.11-drift}] 
Remark that
\begin{align*}
T_N^{-1/2}&\sum_{k=0}^{N-1} Y^{(i)}_k \cdot  \left(\Delta_k D^{(j)} \cdot \indicator\{|\Delta_k Y^{(j)}| \leq v_N^{(j)}\} - \Delta_k D^{(j)} \right) = T_N^{-1/2}\sum_{k=0}^{N-1} Y^{(i)}_k \cdot  \Delta_k D^{(j)} \cdot \indicator\{|\Delta_k Y^{(j)}| > v_N^{(j)}\}.
\end{align*}
Lemma \ref{lemma:lemma-4.12-small-inc-are-small-jumps} yields that we can prove the result for the quantity above where $\indicator\{|\Delta_k Y^{(j)}| > v_N^{(j)}\}$ is replaced by $\indicator\{|\Delta_k J^{2,(j)}| > 2v_N^{(j)}\}$ or $\indicator\{|\Delta_k J^{(j)}| > 2v_N^{(j)}\}$. Similarly to the proof of Lemma \ref{lemma:lemma-4.14-non-jump-neg} (Lemma \ref{lemma:lemma-4.14-term-in-D-drift}, presented in Appendix \ref{proof:lemma:lemma-4.14-non-jump-neg}), we write the following up to a term that converges to zero in probability as $N\rightarrow \infty$:
\begin{align*}
      T_N^{-1/2}&\sum_{k=0}^{N-1} \left|Y^{(i)}_k \cdot  \Delta_k D^{(j)} \cdot \indicator\{|\Delta_k Y^{(j)}| > v_N^{(j)}\}\right|\\
    &\leq T_N^{-1/2}\max_{l\in\{1,\dots,d\}}|Q_{jl}| \\
    &\qquad \times \left( \sum_{l=1}^d \sum_{k=0}^{N-1} \left|Y^{(i)}_k  \int_{t_k}^{t_{k+1}} (Y^{(l)}_s - Y^{(l)}_k)ds \right|\indicator\{|\Delta_k J^{(j)}| > 2v_N^{(j)}\} \right.\\
      &\qquad \qquad \left.  + \sum_{l=1}^d \sum_{k=0}^{N-1}(t_{k+1}-t_k)  \left|Y^{(i)}_k Y^{(l)}_k \right| \indicator\{|\Delta_k J^{2,(j)}| > 2v_N^{(j)}\}   \right) + o_p(1).
\end{align*}
Since $v_N^{(j)}=\Delta_N^{\beta^{(j)}}$, Markov's inequality yields that $\proba\left(|\Delta_k J^{2,(j)}| > 2v_N^{(j)}\right) = O(\Delta_N^{1-\beta^{(j)}})$ and by Chebychev's inequality, we have $\proba\left(|\Delta_k J^{(j)}| > 2v_N^{(j)}\right) = O(\Delta_N^{1-2\beta^{(j)}})$. By independence of $Y^{(i)}_k$, $Y^{(l)}_k$ and $\Delta_k J^{(j)}$, we obtain directly for the second term:
$$ \E\left(T_N^{-1/2}\sum_{k=0}^{N-1}(t_{k+1}-t_k)  \left|Y^{(i)}_k Y^{(l)}_k \right| \indicator\{|\Delta_k J^{(j)}| > 2v_N^{(j)}\}\right) = O(T_N^{1/2} \Delta_N^{1-\beta^{(j)}}).$$
and then for the first term, according to H{\"o}lder's inequality, we obtain
\begin{align*}
	\E&\left(\left|T_N^{-1/2}\sum_{k=0}^{N-1} \left| \int_{t_k}^{t_{k+1}} (Y^{(l)}_s - Y^{(l)}_k)ds \right|\times \left|Y^{(i)}_k \right|\indicator\{|\Delta_k J^{(j)}| > 2v_N^{(j)}\}\right|\right) \\
	&= T_N^{-1/2}\sum_{k=0}^{N-1} \E\left(\left|\int_{t_k}^{t_{k+1}} (Y^{(l)}_s - Y^{(l)}_k)ds\right|^2\right)^{1/2} \E\left(|Y^{(i)}_k|^2\right)^{1/2}\proba\left(|\Delta_k J^{(j)}| > 2v_N^{(j)}\right)^{1/2}\\
	&\leq T_N^{-1/2} \E\left(|Y^{(i)}_0|^2\right)^{1/2}  N \Delta_N^{1/2} \sup_{s\in[0,\Delta_N]}\E\left(\left|Y^{(l)}_s - Y^{(l)}_0\right|^2\right)^{1/2}\proba\left(|\Delta_k J^{(j)}| > 2v_N^{(j)}\right)^{1/2}\\
	&= O(T^{1/2}_N\Delta_N^{1/2-\beta^{(j)}}).
\end{align*}
Therefore, we have proved the convergence in probability since $T_N \Delta_N^{1-2\beta^{(j)}} = o(1)$ thus $T_N^{1/2} \Delta_N^{1/2-\beta^{(j)}} = o(1)$ as $N\rightarrow \infty$ by Assumption \ref{assumption:cv-rate-beta}.
\end{proof}

\subsection{Proof of Lemma \ref{lemma:continuous-martingale-part-infinite}}
\label{proof:lemma:continuous-martingale-part-infinite}

\begin{proof}[Proof of Lemma \ref{lemma:continuous-martingale-part-infinite}]
%The beginning of this proof is necessary given the multivariate nature of the objects involved. However, we make two comments: the univariate proof can be found in the proof of Lemma 4.10, \cite{mai2014efficient}. Also, the strategy is to use the said proof as soon as we remove the multivariate component as presented below.\\
%\textbf{very similar to the proof of Lemma 4.10, \cite{mai2014efficient}}
We define the process containing all components $\YY$ except the infinite-activity pure jump process $\JJ^2$, namely define $\widehat{\YY}_t := \YY_0 - \int_0^t\Q\YY_sds + \WW_t + \JJ^1_t$. Consider splitting $\YY_t$ into $\YY_t - \widehat{\YY}_t$ and $\JJ^2_t$ as follows:
\begin{align*}
   &Y^{(i)}_k\left(\Delta_kY^{(j)}\indicator\{|\Delta_kY^{(j)}| \leq v_N^{(j)}\}-\Delta_kY^{(j),c}\right)\\
    &=Y^{(i)}_k\left(\Delta_k\widehat{Y}^{(j)}\indicator\{|\Delta_kY^{(j)}| \leq v_N^{(j)}\}-\Delta_kY^{(j),c}\right) + Y^{(i)}_k\Delta_k J^{2,(j)}\cdot\indicator\{|\Delta_kY^{(j)}| \leq v_N^{(j)}\}\\
    &=Y^{(i)}_k\left(\Delta_k\widehat{Y}^{(j)}\indicator\{|\Delta_k\widehat{Y}^{(j)}| \leq 2v_N^{(j)}\}-\Delta_kY^{(j),c}\right)\\
    & \qquad  + Y^{(i)}_k\Delta_k\widehat{Y}^{(j)}\left(\indicator\{|\Delta_kY^{(j)}| \leq v_N^{(j)}\}-\indicator\{|\Delta_k\widehat{Y}^{(j)}| \leq 2v_N^{(j)}\}\right)\\
    & \qquad  + Y^{(i)}_k\Delta_k J^{2,(j)}\cdot \indicator\{|\Delta_kY^{(j)}| \leq v_N^{(j)}\}.
\end{align*}
Define 
\begin{align*}
	\widetilde{V}_N^{1}(i,j) &:= T_N^{1/2}\sum_{k=0}^{N-1} Y^{(i)}_k\left(\Delta_k\widehat{Y}^{(j)}\indicator\{|\Delta_k\widehat{Y}^{(j)}| \leq 2v_N^{(j)}\}-\Delta_kY^{(j),c}\right)\\
	\widetilde{V}_N^{2}(i,j) &:= T_N^{1/2}\sum_{k=0}^{N-1} Y^{(i)}_k\Delta_k\widehat{Y}^{(j)}\left(\indicator\{|\Delta_kY^{(j)}| \leq v_N^{(j)}\}-\indicator\{|\Delta_k\widehat{Y}^{(j)}| \leq 2v_N^{(j)}\}\right)\\
	\widetilde{V}_N^{3}(i,j) &:= T_N^{1/2}\sum_{k=0}^{N-1} Y^{(i)}_k\Delta_k J^{2,(j)}\cdot \indicator\{|\Delta_kY^{(j)}| \leq v_N^{(j)}\}.
\end{align*}    
The first term $V^{1}_N(i,j)$ is similar quantity to Lemma \ref{lemma:cts_filtered_approx} and by definition $V_N^1(i,j)$ has finite activity jump. By applying this result, we have that $V^{1}_N(i,j) = O_p(\Delta_N^{1/2}T_N)$: it converges in probability to zero as $N \rightarrow \infty$. 

The second term can be rewritten if one remarks that $\indicator\{|\Delta_kY^{(j)}| \leq v_N^{(j)}\}-\indicator\{|\Delta_k\widehat{Y}^{(j)}| \leq 2v_N^{(j)}\}$ can be split into
$$\indicator\{|\Delta_kY^{(j)}| \leq v_N^{(j)}, |\Delta_k\widehat{Y}^{(j)}| > 2v_N^{(j)}\} - \indicator\{|\Delta_kY^{(j)}| > v_N^{(j)}, |\Delta_k\widehat{Y}^{(j)}| \leq 2v_N^{(j)}\}.$$
Hence
\begin{align*}
    \widetilde{V}^{2}_N(i,j) &= T_N^{-1/2}\sum_{k=0}^{N-1}Y^{(i)}_k\Delta_k\widehat{Y}^{(j)}\indicator\{|\Delta_kY^{(j)}| \leq v_N^{(j)}, |\Delta_k\widehat{Y}^{(j)}| > 2v_N^{(j)}\} \\
    &- T_N^{-1/2}\sum_{k=0}^{N-1}Y^{(i)}_k\Delta_k\widehat{Y}^{(j)}\indicator\{|\Delta_kY^{(j)}| > v_N^{(j)}, |\Delta_k\widehat{Y}^{(j)}| \leq 2v_N^{(j)}\}\\
    &=: \widetilde{V}^{2,1}_N(i,j) - \widetilde{V}^{2,2}_N(i,j),
\end{align*}
with obvious definitions for $\{V^{2,l}_N(i,j), : \ l\in\{1,2\}\}$.
The first term $\widetilde{V}^{2,1}_N(i,j)$ can be shown to converge to zero in probability since, informally, it implies that the process $\JJ^{2}$ is negligible overall as below $v_N^{(j)}=\Delta_N^{\beta^{(j)}}$ at every time step. Indeed, by sub-additivity:
\begin{align*}
    \proba\left(\left|\widetilde{V}^{2,1}_N(i,j)\right| > 0\right) &\leq \sum_{k=0}^{N-1}\proba\left(|Y^{(i)}_k \Delta_k\widehat{Y}^{(j)}| \indicator\{|\Delta_kY^{(j)}| \leq v_N^{(j)}, |\Delta_k\widehat{Y}^{(j)}| > 2v_N^{(j)}\} > 0\right)\\
    &\leq \sum_{k=0}^{N-1}\proba\left(|\Delta_kY^{(j)}| \leq v_N^{(j)}, |\Delta_k\widehat{Y}^{(j)}| > 2v_N^{(j)}\right)
\end{align*}
Since we now have a univariate right-hand side, we now use the proof of Lemma 4.10, \cite{mai2014efficient}. Remark that by the triangle inequality %lemma:lemma-4.13-go-big-jump
\begin{equation}
	\label{eq:triangle-y-hat-j1-w-drift}
	\left||\Delta_k\widehat{Y}^{(j)}| - |\Delta_k {J^{1,(j)}}|\right| \leq |\Delta_k{W^{(j)}} +\Delta_k{D^{(j)}}|
\end{equation}
If $|\Delta_k\widehat{Y}^{(j)}| > 2v_N^{(j)}$ and $|\Delta_k {J^{1,(j)}}| = 0$, then $|\Delta_k{W^{(j)}} +\Delta_k{D^{(j)}}| > 2v_N^{(j)}$, thus
$$\proba\left(|\Delta_k\widehat{Y}^{(j)}| > 2v_N^{(j)}, |\Delta_k {J^{1,(j)}}|= 0\right) \leq \proba\left(|\Delta_k{W^{(j)}} +\Delta_k{D^{(j)}} | > 2v_N^{(j)}\right).$$
Recall that, by Lemma \ref{lemma:lemma-3.7} with $\delta = 1/2 -\beta^{(j)}$ and $l^{(j)}=2$, we have
\begin{align*}
    \sum_{k=0}^{N-1}\proba\left(|\Delta_k{W^{(j)}} +\Delta_k{D^{(j)}} | > 2v_N^{(j)}\right) = O\left(N\Delta_N^{2-2\beta^{(j)}}\right) = O(T_N\Delta_N^{1-2\beta^{(j)}}).
\end{align*}
Since this sum goes to zero as $N\rightarrow\infty$, when a large increment in $\widehat{Y}$ occurs, it is highly likely that a large jump occurred (i.e. $|\Delta_k {J^{1,(j)}}| \neq 0$). Also, again by the triangle inequality $\left||\Delta_kY^{(j)}| - |\Delta_k\widehat{Y}^{(j)}|\right| \leq |\Delta_k J^{2,(j)}|$. Hence, on $\{|\Delta_kY^{(j)}| \leq v_N^{(j)}, |\Delta_k\widehat{Y}^{(j)}| > 2v_N^{(j)}\}$ we have $\{|\Delta_k J^{2,(j)}|>v_N^{(j)}\}$. Using the independence between $\Delta_k J^{1,(j)}$ and $\Delta_k J^{2,(j)}$, we prove that $\widetilde{V}_N^{2,1}(i,j)$ is negligible as follows:  
\begin{align*}
    \proba\left(\left|\widetilde{V}^{2,1}_N(i,j)\right| > 0\right) &\leq  \sum_{k=0}^{N-1} \proba\left(|\Delta_k {J^{1,(j)}}| \neq 0,  |\Delta_kY^{(j)}| \leq v_N^{(j)}, |\Delta_k\widehat{Y}^{(j)}| > 2v_N^{(j)}\right) \\
    &\  +\sum_{k=0}^{N-1}\proba\left(|\Delta_k {J^{1,(j)}}| = 0, |\Delta_kY^{(j)}| \leq v_N^{(j)}, |\Delta_k\widehat{Y}^{(j)}| > 2v_N^{(j)} \right)\\
    &\leq \sum_{k=0}^{N-1} \proba\left(|\Delta_k {J^{1,(j)}}| \neq 0, |\Delta_k {J^{2,(j)}}| > v_N^{(j)}\right) + O\left(T_N \Delta_N^{1-2\beta^{(j)}} \right)\\
    &\leq  \sum_{k=0}^{N-1} \proba\left(|\Delta_k {J^{1,(j)}}| \neq 0\right) \proba\left(|\Delta_k {J^{2,(j)}}| > v_N^{(j)}\right)  + O\left(T_N \Delta_N^{1-2\beta^{(j)}} \right)\\
     &= O\left(T_N \Delta_N^{1-2\beta^{(j)}} \right).
\end{align*}
Note $\proba\left(|\Delta_k {J^{2,(j)}}| > v_N^{(j)}\right) = O(\Delta_N^{1-2\beta^{(j)}})$ is given by the Chebychev's inequality. The term $\widetilde{V}_N^{2,2}(i,j)$ can undergo the same treatment and we have
$$\proba\left(\left|\widetilde{V}^{2,2}_N(i,j)\right|>0\right) \leq \sum_{k=0}^{N-1}\proba\left(|\Delta_kY^{(j)}| > v_N^{(j)}, |\Delta_k\widehat{Y}^{(j)}| \leq 2v_N^{(j)}\right).$$
Again, we fall back to the univariate case and prove the convergence as presented in \cite{mai2014efficient}. Recall that $\JJ^1$ is a compound Poisson process with $(N_t^{(i)}, \ i \in\{1,\dots,d\})$ as a collection of $d$ Poisson processes with counting processes $(N^{(i)}: i \in \{1,\dots,d \})$. According to \eqref{eq:triangle-y-hat-j1-w-drift}, we have that
$$\{|\Delta_k \widehat{Y}^{(j)}| \leq 2v_N^{(j)}, \Delta_k N^{(j)} = 1\} \subset \{|\Delta_k{W^{(j)}} +\Delta_k{D^{(j)}} | > 1-2v_N^{(j)}\}.$$
 With $l^{(j)}=2$ and
 $$\delta:=1/2-\ln(1-2\Delta_N^{\beta^{(j)}})/\ln(\Delta_N) < 1/2, \quad \text{for $N$ large enough,}$$
  such that $1-2v_N^{(j)} = \Delta_N^{1/2-\delta}$, according to Lemma \ref{lemma:lemma-3.7} we obtain:
  $$\proba\left(|\Delta_k{W^{(j)}} +\Delta_k{D^{(j)}} | > 1-2v_n^{(j)}\right) = O(\Delta_N^{2})$$ 
  Since $\JJ^{1}$ is defined as the \emph{large-jump} component, we show that it has a negligible impact when $\{|\Delta_k \widehat{Y}^{(j)}| \leq 2v_N^{(j)}\}$ as follows
\begin{align*}
	T_N^{-1/2}&\sum_{k=0}^{N-1}\proba\left(|\Delta_k \widehat{Y}^{(j)}| \leq 2v_N^{(j)}, \Delta_k J^{1,(j)} > 0\right)\\
	&= T_N^{-1/2}\sum_{k=0}^{N-1}\proba\left(|\Delta_k \widehat{Y}^{(j)}| \leq 2v_N^{(j)}, \Delta_k J^{1,(j)} > 0\right)\\
	&= T_N^{-1/2}\sum_{k=0}^{N-1}\left[\proba\left(|\Delta_k \widehat{Y}^{(j)} |\leq 2v_N^{(j)}, \Delta_k N^{(j)} = 1\right) + \proba\left(|\Delta_k \widehat{Y}^{(j)}| \leq 2v_N^{(j)}, \Delta_k N^{(j)} > 1\right)\right]\\
	&\leq T_N^{-1/2}\sum_{k=0}^{N-1}\left[\proba\left(|\Delta_k \widehat{Y}^{(j)}| \leq 2v_N^{(j)}, \Delta_k N^{(j)} = 1\right) + O(\Delta_N^2)\right]\\
	&\leq T_N^{-1/2}\sum_{k=0}^{N-1}\left[\proba\left(|\Delta_k{W^{(j)}} +\Delta_k{D^{(j)}} | > 1-2v_n^{(j)}\right) + O(\Delta_N^2)\right]\\
	&\leq T_N^{-1/2}\sum_{k=0}^{N-1} O(\Delta_N^2) \leq O(T_N^{1/2}\Delta_N), \quad \text{as $N\rightarrow \infty$.}
\end{align*}
which converges to zero in the limit. This shows that we have  $\Delta_k \widehat{Y}^{(j)} = \Delta_k{W^{(j)}} + \Delta_k {D^{(j)}} \ a.s.$ on $\{|\Delta_k \widehat{Y}^{(j)}| \leq 2v_N^{(j)}\}$ as $N\rightarrow \infty$. 
 
Next, according to Lemma \ref{lemma:lemma-mai-4.11-drift} (and its proof), the part of $\widetilde{V}^{2,2}_N(i,j)$ with respect to $\Delta_k D^{(j)}$ is such that
\begin{align*}
	 \E & \left(T_N^{-1/2}\sum_{k=0}^{N-1}\left|Y^{(i)}_k\right| \left|\Delta_k D^{(j)}\right|\indicator\{|\Delta_kY^{(j)}| > v_N^{(j)}, |\Delta_k\widehat{Y}^{(j)}| \leq 2v_N^{(j)}\}\right)\\
	 &\leq \E\left(T_N^{-1/2}\sum_{k=0}^{N-1}\left|Y^{(i)}_k\right| \left|\Delta_k D^{(j)}\right|\indicator\{|\Delta_k J^{2,(j)}| > v_N^{(j)}\}\right) \longrightarrow 0, \quad \text{as $N\rightarrow \infty$.}
\end{align*}
The second part of $\widetilde{V}_N^{2,2}(i,j)$ is given by 
$$T_N^{-1/2}\sum_{k=0}^{N-1}Y^{(i)}_k\Delta_k W^{(j)}\indicator\{|\Delta_kY^{(j)}| > v_N^{(j)}, |\Delta_k\widehat{Y}^{(j)}| \leq 2v_N^{(j)}\}.$$
Recall the independence between $Y^{(i)}_k$, $\Delta_k W^{(j)}$ and $\Delta_k J^{2,(j)}$. Its expectation is then zero and its $L^2$-norm we can bound as follows by :
\begin{align*}
\E &  \left(\left|T_N^{-1/2}\sum_{k=0}^{N-1} Y^{(i)}_k\Delta_k W^{(j)}\indicator\{|\Delta_kY^{(j)}| > v_N^{(j)}, |\Delta_k\widehat{Y}^{(j)}| \leq 2v_N^{(j)}\}\right|^2\right)\\
	&\leq 2 T_N^{-1}\sum_{k=0}^{N-1} \E  \left(\left|Y^{(i)}_k\right|^2\left|\Delta_k W^{(j)}\right|^2\indicator\{|\Delta_kY^{(j)}| > v_N^{(j)}, |\Delta_k\widehat{Y}^{(j)}| \leq 2v_N^{(j)}\}\right)\\
	&\leq 2T_N^{-1}\sum_{k=0}^{N-1} \E\left(\left|Y^{(i)}_k\right|^2\left|\Delta_k W^{(j)}\right|^2 \indicator\{|\Delta_k J^{2,(j)}| > v_N^{(j)}\}\right)\\
	&\leq  2T_N^{-1}\sum_{k=0}^{N-1} \E\left(\left|Y^{(i)}_k\right|^2\right)\E\left(\left|\Delta_k W^{(j)} \right|^2\right)\proba\left(|\Delta_k J^{2,(j)}| > v_N^{(j)}\right)\\
	&\leq O(T_N^{-1} N \Delta_N \Delta_N^{1/2-\beta^{(j)}}), \quad \text{by Markov's inequality,}\\
	&\leq O(\Delta_N^{1/2-\beta^{(j)}})\rightarrow 0, \ \text{as $N\rightarrow\infty$.}.
\end{align*}
Therefore, we have also proved that this second term of $\widetilde{V}^{2,2}_N(i,j)$ converges to zero in probability and, as a result, so does $\widetilde{V}^{2,2}_N(i,j)$ and, in turn, $\widetilde{V}^2_N(i,j)$ as $N\rightarrow \infty$. 

Finally, for the third term, recall Lemma \ref{lemma:lemma-4.12-supp-large-jump-negligible} such that we consider
$$\widetilde{V}_N^{3}(i,j) := T_N^{1/2}\sum_{k=0}^{N-1} Y^{(i)}_k\Delta_k J^{2,(j)}\cdot \indicator\{|\Delta_kY^{(j)}| \leq v_N^{(j)}, |\Delta_k J^{2,(j)}| \leq 2v_N^{(j)}\} + o_p(1).$$
Then, by the proof of Lemma \ref{lemma:lemma-4.12-small-inc-are-small-jumps}, we know that 
\begin{align*}
	T_N^{1/2}&\sum_{k=0}^{N-1} Y^{(i)}_k\Delta_k J^{2,(j)}\cdot \left(\indicator\{|\Delta_k J^{2,(j)}| \leq 2v_N^{(j)}\} - \indicator\{|\Delta_kY^{(j)}| \leq v_N^{(j)}, |\Delta_k J^{2,(j)}| \leq 2v_N^{(j)}\}\right)\\
	&=  T_N^{1/2}\sum_{k=0}^{N-1} Y^{(i)}_k\Delta_k J^{2,(j)}\cdot \indicator\{|\Delta_kY^{(j)}| > v_N^{(j)}, |\Delta_k J^{2,(j)}| \leq 2v_N^{(j)}\}\probconv 0, \quad \text{as $N\rightarrow\infty$.}
\end{align*}
Since $v_N^{(j)} \rightarrow 0$ as $N\rightarrow \infty$, then by Assumption \ref{assumption:infinite_activity}-\eqref{assumption:infinite_activity:jumps_2_epsilon}, we have that 
$$\E\left(\Delta_k J^{2,(j)} \indicator\{|\Delta_k J^{2,(j)}| \leq 2v_N^{(j)}\} \right) = 0, \ \text{for $N$ large enough.}$$
Then, by independence of $ Y^{(i)}_k$ and $\Delta_k J^{2,(j)}$, we obtain that 
$$\E\left(T_N^{1/2}\sum_{k=0}^{N-1} Y^{(i)}_k\Delta_k J^{2,(j)}\indicator\{|\Delta_k J^{2,(j)}| \leq 2v_N^{(j)}\} \right) = 0.$$
Finally, again by independence and by the triangle inequality, the second moment of this quantity is bounded by
\begin{align*}
	T_N^{-1} &\sum_{k=0}^{N-1} \E\left(\left|Y^{(i)}_k\right|^2\right) \E\left(\left|\Delta_k J^{2,(j)}\right|^2\indicator\{|\Delta_k J^{2,(j)}| \leq 2v_N^{(j)}\} \right)\\
	&= O(T_N^{-1}N \Delta_N^{1+\beta^{(i)}(2-\alpha)}), \qquad \text{by Lemma \ref{lemma:small_jump_negligeable},}\\
	&= O( \Delta_N^{\beta^{(i)}(2-\alpha)}) \rightarrow 0, \quad \text{as $N \rightarrow \infty$.}
\end{align*}
This concludes the proof.
\end{proof}

\subsection{Proof of Lemma \ref{lemma:consistency-conservation-jump-filtering-infinite}}
\label{proof:lemma:consistency-conservation-jump-filtering-infinite}

\begin{proof}[Proof of Lemma \ref{lemma:consistency-conservation-jump-filtering-infinite}]
	Observe that
\begin{align*}
    T_N^{1/2}(\boldsymbol{\widetilde{\psi}}_N - \boldsymbol{\overline{\psi}}_N) &= T_N\Sn_N^{-1}\left(\Abar_N - \Atilde_N\right).
\end{align*}
Now, by Lemma \ref{lemma:continuous-martingale-part-infinite},
We know that for any $i,j \in\{1,\dots,d\}$, we have that
$$T_N^{-1/2}\sum_{k=0}^{N-1}Y^{(i)}_k\left(\Delta_kY^{(j)}\indicator\{|\Delta_kY^{(j)}| \leq v_N^{(j)}\}-\Delta_kY^{(j),c}\right) \probconv 0.$$
Also, recall that $\Sn_N := \Kn_N \otimes \Id$ where $\Kn_N = \left(\sum_{k=0}^{N-1} Y^{(i)}_{k} Y^{(j)}_{k} (t_{k+1} - t_{k})\right)_{1\leq i,j \leq d}$ with $Y^{(j)}_{k} := Y^{(j)}_{t_k}$ and that $T_N \Sn_N^{-1} \xrightarrow{} \E\left(\YY_\infty \YY^\top_\infty\right)^{-1}\otimes \Id \ \proba_{0}-a.s.$ as $N \rightarrow \infty$ which is finite componentwise. 
Therefore, similarly to Lemma \ref{lemma:consistency-conservation-jump-filtering-finite}, we obtain that
$$T_N^{1/2}(\boldsymbol{\widetilde{\psi}}_N - \boldsymbol{\overline{\psi}}_N) \probconv \boldsymbol{0}_{d^2}, \qquad \text{as $N\rightarrow \infty$.}$$
\end{proof}

\subsection{Proof of Theorem \ref{th:discrete_clt_infinite}}
\label{proof:th:discrete_clt_infinite}

We can finally prove Theorem \ref{th:discrete_clt_infinite} using the lemmas above.
\begin{proof}[Proof of Theorem \ref{th:discrete_clt_infinite}]
By Lemma \ref{lemma:cv-estimator-continuous-component}, we know that the discretised unfiltered estimator $\boldsymbol{\overline{\psi}}_N$  converges in distribution as follows:
$$T_N^{1/2}(\boldsymbol{\overline{\psi}}_N - \boldsymbol{\psi}) \xrightarrow{ \ \mathcal{D} \ } \mathcal{N}\left(\boldsymbol{0}_{d^2},\ \mathbb{E}\left\{\YY_\infty \YY_\infty^\top \right\}^{-1} \otimes \Sigma \right).$$
By Lemma \ref{lemma:consistency-conservation-jump-filtering-infinite} we have that
$$T_N^{1/2}(\boldsymbol{\overline{\psi}}_N- \boldsymbol{\widetilde{\psi}}_N) \probconv 0, \quad \text{as $N\rightarrow\infty$.}$$
We conclude by Slutsky's lemma.
\end{proof}

%%%%%%%%%%%%%%%%%%%%%%%%%%%%%%%%%%%%%%%%%%%%%%
%% Support information (funding), if any,   %%
%% should be provided in the                %%
%% Acknowledgements section.                %%
%%%%%%%%%%%%%%%%%%%%%%%%%%%%%%%%%%%%%%%%%%%%%%
\section*{Acknowledgments}
The authors would like to thank the Isaac Newton Institute for Mathematical Sciences for support and hospitality during the programme  \emph{The Mathematics of Energy Systems} when work on this paper was undertaken. This work was supported by: EPSRC grant number EP/R014604/1. AV would also like to acknowledge funding by the Simons Foundation.

\bibliographystyle{imsart-number} % Style BST file
\bibliography{NetworkOU-hf-arxiv-v2-merge}       % Bibliography file (usually '*.bib')

\begin{thebibliography}{53}
% BibTex style file: imsart-number.bst, 2017-11-03
% Default style options (sort=1,type=number).
% Used options (sort=1,type=number).

\bibitem{abdelrazeq2018paramboostrap}
\begin{barticle}[author]
\bauthor{\bsnm{Abdelrazeq},~\bfnm{Ibrahim}\binits{I.}}, \bauthor{\bsnm{Ivanoff},~\bfnm{B.~Gail}\binits{B.~G.}} \AND \bauthor{\bsnm{Kulik},~\bfnm{Rafal}\binits{R.}}
(\byear{2018}).
\btitle{{Goodness-of-fit tests for L{\'e}vy-driven Ornstein-Uhlenbeck processes}}.
\bjournal{Canadian Journal of Statistics}
\bvolume{46}
\bpages{355-376}.
\bdoi{https://doi.org/10.1002/cjs.11352}
\end{barticle}
\endbibitem

\bibitem{AitSahalia2010CovarianceAsynchronous}
\begin{barticle}[author]
\bauthor{\bsnm{A{\"i}t-Sahalia},~\bfnm{Yacine}\binits{Y.}}, \bauthor{\bsnm{Fan},~\bfnm{Jianqing}\binits{J.}} \AND \bauthor{\bsnm{Xiu},~\bfnm{Dacheng}\binits{D.}}
(\byear{2010}).
\btitle{{High-Frequency Covariance Estimates With Noisy and Asynchronous Financial Data}}.
\bjournal{Journal of the American Statistical Association}
\bvolume{105}
\bpages{1504-1517}.
\bdoi{10.1198/jasa.2010.tm10163}
\end{barticle}
\endbibitem

\bibitem{aitsahalia2011testingFiniteInfinite}
\begin{barticle}[author]
\bauthor{\bsnm{A{\"i}t-Sahalia},~\bfnm{Yacine}\binits{Y.}} \AND \bauthor{\bsnm{Jacod},~\bfnm{Jean}\binits{J.}}
(\byear{2011}).
\btitle{Testing whether jumps have finite or infinite activity}.
\bjournal{Ann. Statist.}
\bvolume{39}
\bpages{1689--1719}.
\bdoi{10.1214/11-AOS873}
\end{barticle}
\endbibitem

\bibitem{Ait2008OutSampleQV}
\begin{barticle}[author]
\bauthor{\bsnm{A{\"i}t-Sahalia},~\bfnm{Yacine}\binits{Y.}} \AND \bauthor{\bsnm{Mancini},~\bfnm{Loriano}\binits{L.}}
(\byear{2008}).
\btitle{Out of sample forecasts of quadratic variation}.
\bjournal{Journal of Econometrics}
\bvolume{147}
\bpages{17 - 33}.
\bnote{Econometric modelling in finance and risk management: An overview}.
\bdoi{https://doi.org/10.1016/j.jeconom.2008.09.015}
\end{barticle}
\endbibitem

\bibitem{ambroise2012new}
\begin{barticle}[author]
\bauthor{\bsnm{Ambroise},~\bfnm{Christophe}\binits{C.}} \AND \bauthor{\bsnm{Matias},~\bfnm{Catherine}\binits{C.}}
(\byear{2012}).
\btitle{New consistent and asymptotically normal parameter estimates for random-graph mixture models}.
\bjournal{Journal of the Royal Statistical Society: Series B (Statistical Methodology)}
\bvolume{74}
\bpages{3--35}.
\end{barticle}
\endbibitem

\bibitem{barndorff2002bnsmodel}
\begin{barticle}[author]
\bauthor{\bsnm{Barndorff-Nielsen},~\bfnm{Ole~E.}\binits{O.~E.}} \AND \bauthor{\bsnm{Shephard},~\bfnm{Neil}\binits{N.}}
(\byear{2001}).
\btitle{{N}on-{G}aussian {O}rnstein-{U}hlenbeck-based models and some of their uses in financial economics}.
\bjournal{Journal of the Royal Statistical Society: Series B (Statistical Methodology)}
\bvolume{63}
\bpages{167-241}.
\bdoi{10.1111/1467-9868.00282}
\end{barticle}
\endbibitem

\bibitem{barndorff2004BiPowerVariation}
\begin{barticle}[author]
\bauthor{\bsnm{Barndorff-Nielsen},~\bfnm{Ole~E.}\binits{O.~E.}} \AND \bauthor{\bsnm{Shephard},~\bfnm{Neil}\binits{N.}}
(\byear{2004}).
\btitle{{Power and Bipower Variation with Stochastic Volatility and Jumps}}.
\bjournal{Journal of Financial Econometrics}
\bvolume{2}
\bpages{1-37}.
\bdoi{10.1093/jjfinec/nbh001}
\end{barticle}
\endbibitem

\bibitem{bikowski2017autoregressive}
\begin{bmisc}[author]
\bauthor{\bsnm{Bi{\'n}kowski},~\bfnm{Mikołaj}\binits{M.}}, \bauthor{\bsnm{Marti},~\bfnm{Gautier}\binits{G.}} \AND \bauthor{\bsnm{Donnat},~\bfnm{Philippe}\binits{P.}}
(\byear{2017}).
\btitle{{Autoregressive Convolutional Neural Networks for Asynchronous Time Series}}.
\end{bmisc}
\endbibitem

\bibitem{BlumenthalGetoor1961}
\begin{barticle}[author]
\bauthor{\bsnm{Blumenthal},~\bfnm{R.~M.}\binits{R.~M.}} \AND \bauthor{\bsnm{Getoor},~\bfnm{R.~K.}\binits{R.~K.}}
(\byear{1961}).
\btitle{{Sample Functions of Stochastic Processes with Stationary Independent Increments}}.
\bjournal{Journal of Mathematics and Mechanics}
\bvolume{10}
\bpages{493--516}.
\end{barticle}
\endbibitem

\bibitem{bollerslev2011estimation}
\begin{barticle}[author]
\bauthor{\bsnm{Bollerslev},~\bfnm{Tim}\binits{T.}} \AND \bauthor{\bsnm{Todorov},~\bfnm{Viktor}\binits{V.}}
(\byear{2011}).
\btitle{Estimation of jump tails}.
\bjournal{Econometrica}
\bvolume{79}
\bpages{1727--1783}.
\end{barticle}
\endbibitem

\bibitem{bollmeyer2015CosmoRea6}
\begin{barticle}[author]
\bauthor{\bsnm{Bollmeyer},~\bfnm{C.}\binits{C.}}, \bauthor{\bsnm{Keller},~\bfnm{J.~D.}\binits{J.~D.}}, \bauthor{\bsnm{Ohlwein},~\bfnm{C.}\binits{C.}}, \bauthor{\bsnm{Wahl},~\bfnm{S.}\binits{S.}}, \bauthor{\bsnm{Crewell},~\bfnm{S.}\binits{S.}}, \bauthor{\bsnm{Friederichs},~\bfnm{P.}\binits{P.}}, \bauthor{\bsnm{Hense},~\bfnm{A.}\binits{A.}}, \bauthor{\bsnm{Keune},~\bfnm{J.}\binits{J.}}, \bauthor{\bsnm{Kneifel},~\bfnm{S.}\binits{S.}}, \bauthor{\bsnm{Pscheidt},~\bfnm{I.}\binits{I.}}, \bauthor{\bsnm{Redl},~\bfnm{S.}\binits{S.}} \AND \bauthor{\bsnm{Steinke},~\bfnm{S.}\binits{S.}}
(\byear{2015}).
\btitle{Towards a high-resolution regional reanalysis for the European CORDEX domain}.
\bjournal{Quarterly Journal of the Royal Meteorological Society}
\bvolume{141}
\bpages{1-15}.
\bdoi{10.1002/qj.2486}
\end{barticle}
\endbibitem

\bibitem{Boninsegna2018SparseLearningDynamical}
\begin{barticle}[author]
\bauthor{\bsnm{Boninsegna},~\bfnm{Lorenzo}\binits{L.}}, \bauthor{\bsnm{N{\"u}ske},~\bfnm{Feliks}\binits{F.}} \AND \bauthor{\bsnm{Clementi},~\bfnm{Cecilia}\binits{C.}}
(\byear{2018}).
\btitle{Sparse learning of stochastic dynamical equations}.
\bjournal{The Journal of Chemical Physics}
\bvolume{148}
\bpages{241723}.
\bdoi{10.1063/1.5018409}
\end{barticle}
\endbibitem

\bibitem{bretagnolle2006ecole}
\begin{bbook}[author]
\bauthor{\bsnm{Bretagnolle},~\bfnm{Jean~L}\binits{J.~L.}}, \bauthor{\bsnm{Chatterji},~\bfnm{Shrishti~Dhar}\binits{S.~D.}} \AND \bauthor{\bsnm{Meyer},~\bfnm{P-A}\binits{P.-A.}}
(\byear{2006}).
\btitle{{\'E}cole d'{\'e}t{\'e} de probabilit{\'e}s: processus stochastiques}
\bvolume{307}.
\bpublisher{Springer}.
\end{bbook}
\endbibitem

\bibitem{Brockwell2009}
\begin{binbook}[author]
\bauthor{\bsnm{Brockwell},~\bfnm{Peter~J.}\binits{P.~J.}}
(\byear{2009}).
\btitle{L{\'e}vy--{D}riven {C}ontinuous--{T}ime {ARMA} {P}rocesses}
In \bbooktitle{Handbook of Financial Time Series}
\bpages{457--480}.
\bpublisher{Springer Berlin Heidelberg}, \baddress{Berlin, Heidelberg}.
\bdoi{10.1007/978-3-540-71297-8_20}
\end{binbook}
\endbibitem

\bibitem{Brockwell2007}
\begin{barticle}[author]
\bauthor{\bsnm{Brockwell},~\bfnm{Peter~J.}\binits{P.~J.}}, \bauthor{\bsnm{Davis},~\bfnm{Richard~A.}\binits{R.~A.}} \AND \bauthor{\bsnm{Yang},~\bfnm{Yu}\binits{Y.}}
(\byear{2007}).
\btitle{Estimation for Nonnegative {L}\'{e}vy-Driven {O}rnstein-{U}hlenbeck Processes}.
\bjournal{Journal of Applied Probability}
\bvolume{44}
\bpages{977--989}.
\end{barticle}
\endbibitem

\bibitem{brockwell2013carma-recovery}
\begin{barticle}[author]
\bauthor{\bsnm{Brockwell},~\bfnm{Peter~J.}\binits{P.~J.}} \AND \bauthor{\bsnm{Schlemm},~\bfnm{Eckhard}\binits{E.}}
(\byear{2013}).
\btitle{Parametric estimation of the driving {L}{\'e}vy process of multivariate CARMA processes from discrete observations}.
\bjournal{Journal of Multivariate Analysis}
\bvolume{115}
\bpages{217 - 251}.
\bdoi{https://doi.org/10.1016/j.jmva.2012.09.004}
\end{barticle}
\endbibitem

\bibitem{carr2002}
\begin{barticle}[author]
\bauthor{\bsnm{Carr},~\bfnm{Peter}\binits{P.}} \AND \bauthor{\bsnm{Geman},~\bfnm{Helyette}\binits{H.}}
(\byear{2002}).
\btitle{The Fine Structure of Asset Returns: An Empirical Investigation}.
\bjournal{The Journal of Business}
\bvolume{75}
\bpages{305-332}.
\end{barticle}
\endbibitem

\bibitem{cheang2018three}
\begin{bphdthesis}[author]
\bauthor{\bsnm{Cheang},~\bfnm{Chi~Wan}\binits{C.~W.}}
(\byear{2018}).
\btitle{{Three Essays in Financial Econometrics: Fractional Cointegration, Nonlinearities and Asynchronicities}},
\btype{PhD thesis},
\bpublisher{University of Southampton}.
\end{bphdthesis}
\endbibitem

\bibitem{chen2020community}
\begin{bmisc}[author]
\bauthor{\bsnm{Chen},~\bfnm{Elynn~Y.}\binits{E.~Y.}}, \bauthor{\bsnm{Fan},~\bfnm{Jianqing}\binits{J.}} \AND \bauthor{\bsnm{Zhu},~\bfnm{Xuening}\binits{X.}}
(\byear{2020}).
\btitle{{Community Network Auto-Regression for High-Dimensional Time Series}}.
\end{bmisc}
\endbibitem

\bibitem{ContRama2004levyitodecomp}
\begin{binbook}[author]
\bauthor{\bsnm{Cont},~\bfnm{Rama}\binits{R.}} \AND \bauthor{\bsnm{Tankov},~\bfnm{Peter}\binits{P.}}
(\byear{2004}).
\btitle{{Financial Modelling with Jump Processes}}.
\bseries{{Chapman \& Hall/CRC Financial Mathematics Series}}
\bchapter{3},
\bpages{69--80}.
\end{binbook}
\endbibitem

\bibitem{courgeau2020likelihood}
\begin{barticle}[author]
\bauthor{\bsnm{Courgeau},~\bfnm{Valentin}\binits{V.}} \AND \bauthor{\bsnm{Veraart},~\bfnm{Almut~ED}\binits{A.~E.}}
(\byear{2020}).
\btitle{Likelihood theory for the {G}raph {O}rnstein-{U}hlenbeck process}.
\bjournal{arXiv preprint arXiv:2005.12720}.
\end{barticle}
\endbibitem

\bibitem{crimaldi2005multimgclt}
\begin{barticle}[author]
\bauthor{\bsnm{Crimaldi},~\bfnm{Irene}\binits{I.}} \AND \bauthor{\bsnm{Pratelli},~\bfnm{Luca}\binits{L.}}
(\byear{2005}).
\btitle{Convergence results for multivariate martingales}.
\bjournal{Stochastic Processes and their Applications}
\bvolume{115}
\bpages{571 - 577}.
\bdoi{https://doi.org/10.1016/j.spa.2004.10.004}
\end{barticle}
\endbibitem

\bibitem{Eberlein2001Ghyp}
\begin{binbook}[author]
\bauthor{\bsnm{Eberlein},~\bfnm{Ernst}\binits{E.}}
(\byear{2001}).
\btitle{{A}pplication of {G}eneralized {H}yperbolic {L}{\'e}vy {M}otions to {F}inance}
In \bbooktitle{L{\'e}vy Processes: {T}heory and {A}pplications}
\bpages{319--336}.
\bpublisher{Birkh{\"a}user Boston}, \baddress{Boston, MA}.
\bdoi{10.1007/978-1-4612-0197-7_14}
\end{binbook}
\endbibitem

\bibitem{endres2019tradingLevyDrivenOU}
\begin{barticle}[author]
\bauthor{\bsnm{Endres},~\bfnm{S.}\binits{S.}} \AND \bauthor{\bsnm{St{\"u}binger},~\bfnm{J.}\binits{J.}}
(\byear{2019}).
\btitle{{Optimal trading strategies for L{\'e}vy-driven Ornstein--Uhlenbeck processes}}.
\bjournal{Applied Economics}
\bvolume{51}
\bpages{3153-3169}.
\bdoi{10.1080/00036846.2019.1566688}
\end{barticle}
\endbibitem

\bibitem{fasen2013}
\begin{barticle}[author]
\bauthor{\bsnm{Fasen},~\bfnm{Vicky}\binits{V.}}
(\byear{2013}).
\btitle{Statistical estimation of multivariate {O}rnstein-{U}hlenbeck processes and applications to co-integration}.
\bjournal{Journal of Econometrics}
\bvolume{172}
\bpages{325--337}.
\end{barticle}
\endbibitem

\bibitem{Gaiffas2019SparseOUProcess}
\begin{barticle}[author]
\bauthor{\bsnm{Ga{\"i}ffas},~\bfnm{St{\'e}phane}\binits{S.}} \AND \bauthor{\bsnm{Matulewicz},~\bfnm{Gustaw}\binits{G.}}
(\byear{2019}).
\btitle{Sparse inference of the drift of a high-dimensional {O}rnstein-{U}hlenbeck process}.
\bjournal{Journal of Multivariate Analysis}
\bvolume{169}
\bpages{1--20}.
\bdoi{https://doi.org/10.1016/j.jmva.2018.08.005}
\end{barticle}
\endbibitem

\bibitem{Hanser2005RealizedVarianceHighFrequency}
\begin{barticle}[author]
\bauthor{\bsnm{Hansen},~\bfnm{Peter}\binits{P.}} \AND \bauthor{\bsnm{Lunde},~\bfnm{Asger}\binits{A.}}
(\byear{2005}).
\btitle{A {R}ealized {V}ariance for the {W}hole {D}ay {B}ased on {I}ntermittent {H}igh-{F}requency {D}ata}.
\bjournal{Journal of Financial Econometrics}
\bvolume{3}
\bpages{525-554}.
\end{barticle}
\endbibitem

\bibitem{hausler2015whystableconvergence}
\begin{binbook}[author]
\bauthor{\bsnm{H{\"a}usler},~\bfnm{Erich}\binits{E.}} \AND \bauthor{\bsnm{Luschgy},~\bfnm{Harald}\binits{H.}}
(\byear{2015}).
\btitle{Why Stable Convergence?}
In \bbooktitle{Stable Convergence and Stable Limit Theorems}
\bpages{1--9}.
\bpublisher{Springer International Publishing}, \baddress{Cham}.
\bdoi{10.1007/978-3-319-18329-9_1}
\end{binbook}
\endbibitem

\bibitem{hol2018estimation}
\begin{bmisc}[author]
\bauthor{\bsnm{Hol{\'y}},~\bfnm{Vladimír}\binits{V.}} \AND \bauthor{\bsnm{Tomanov{\'a}},~\bfnm{Petra}\binits{P.}}
(\byear{2018}).
\btitle{{E}stimation of {O}rnstein-{U}hlenbeck {P}rocess {U}sing {U}ltra-{H}igh-{F}requency {D}ata with {A}pplication to {I}ntraday {P}airs {T}rading {S}trategy}.
\end{bmisc}
\endbibitem

\bibitem{Hu2009LeastSquares}
\begin{barticle}[author]
\bauthor{\bsnm{Hu},~\bfnm{Yaozhong}\binits{Y.}} \AND \bauthor{\bsnm{Long},~\bfnm{Hongwei}\binits{H.}}
(\byear{2009}).
\btitle{{Least squares estimator for Ornstein--Uhlenbeck processes driven by $\alpha$-stable motions}}.
\bjournal{Stochastic Processes and their Applications}
\bvolume{119}
\bpages{2465 - 2480}.
\bdoi{https://doi.org/10.1016/j.spa.2008.12.006}
\end{barticle}
\endbibitem

\bibitem{jacod1997continuous}
\begin{bincollection}[author]
\bauthor{\bsnm{Jacod},~\bfnm{Jean}\binits{J.}}
(\byear{1997}).
\btitle{On continuous conditional {G}aussian martingales and stable convergence in law}.
In \bbooktitle{S{\'e}minaire de Probabilit{\'e}s XXXI}
(\beditor{\bfnm{Jacques}\binits{J.}~\bsnm{Az{\'e}ma}}, \beditor{\bfnm{Marc}\binits{M.}~\bsnm{Yor}} \AND \beditor{\bfnm{Michel}\binits{M.}~\bsnm{Emery}}, eds.)
\bpages{232--246}.
\bpublisher{Springer Berlin Heidelberg}, \baddress{Berlin, Heidelberg}.
\end{bincollection}
\endbibitem

\bibitem{jensen2017re}
\begin{barticle}[author]
\bauthor{\bsnm{Jensen},~\bfnm{Tue~V}\binits{T.~V.}} \AND \bauthor{\bsnm{Pinson},~\bfnm{Pierre}\binits{P.}}
(\byear{2017}).
\btitle{RE-{E}urope, a large-scale dataset for modeling a highly renewable European electricity system}.
\bjournal{Scientific {D}ata}
\bvolume{4}
\bpages{170175}.
\end{barticle}
\endbibitem

\bibitem{Knight2016ModellingDA}
\begin{barticle}[author]
\bauthor{\bsnm{Knight},~\bfnm{Marina~I.}\binits{M.~I.}}, \bauthor{\bsnm{Nunes},~\bfnm{M.~A.}\binits{M.~A.}} \AND \bauthor{\bsnm{Nason},~\bfnm{Guy~P.}\binits{G.~P.}}
(\byear{2016}).
\btitle{{Modelling, Detrending and Decorrelation of Network Time Series}}.
\bjournal{arXiv: Methodology}.
\end{barticle}
\endbibitem

\bibitem{Longoria2019OULevyElectricityPortfolios}
\begin{barticle}[author]
\bauthor{\bsnm{Longoria},~\bfnm{Genaro}\binits{G.}}, \bauthor{\bsnm{Davy},~\bfnm{Alan}\binits{A.}} \AND \bauthor{\bsnm{Shi},~\bfnm{Lei}\binits{L.}}
(\byear{2018}).
\btitle{{{O}rnstein-{U}hlenbeck-{L}{\'e}vy {E}lectricity {P}ortfolios with {W}ind {E}nergy {C}ontracting}}.
\bjournal{Technology and Economics of Smart Grids and Sustainable Energy}
\bvolume{3}
\bpages{16}.
\bdoi{10.1007/s40866-018-0054-9}
\end{barticle}
\endbibitem

\bibitem{ma2021sparseAutoregressions}
\begin{barticle}[author]
\bauthor{\bsnm{Ma},~\bfnm{Yingying}\binits{Y.}}, \bauthor{\bsnm{Guo},~\bfnm{Shaojun}\binits{S.}} \AND \bauthor{\bsnm{Wang},~\bfnm{Hansheng}\binits{H.}}
(\byear{2021}).
\btitle{{Sparse spatio-temporal autoregressions by profiling and bagging}}.
\bjournal{Journal of Econometrics (to appear)}.
\bdoi{https://doi.org/10.1016/j.jeconom.2020.10.010}
\end{barticle}
\endbibitem

\bibitem{mai2014efficient}
\begin{barticle}[author]
\bauthor{\bsnm{Mai},~\bfnm{Hilmar}\binits{H.}}
(\byear{2014}).
\btitle{Efficient maximum likelihood estimation for {L}{\'e}vy-driven {O}rnstein-{U}hlenbeck processes}.
\bjournal{Bernoulli}
\bvolume{20}
\bpages{919--957}.
\end{barticle}
\endbibitem

\bibitem{Marquardt2007MultivariateCarma}
\begin{barticle}[author]
\bauthor{\bsnm{Marquardt},~\bfnm{Tina}\binits{T.}} \AND \bauthor{\bsnm{Stelzer},~\bfnm{Robert}\binits{R.}}
(\byear{2007}).
\btitle{Multivariate CARMA processes}.
\bjournal{Stochastic Processes and their Applications}
\bvolume{117}
\bpages{96 - 120}.
\bdoi{https://doi.org/10.1016/j.spa.2006.05.014}
\end{barticle}
\endbibitem

\bibitem{Masuda2004}
\begin{barticle}[author]
\bauthor{\bsnm{Masuda},~\bfnm{Hiroki}\binits{H.}}
(\byear{2004}).
\btitle{On multidimensional {O}rnstein-{U}hlenbeck processes driven by a general {L}\'{e}vy process}.
\bjournal{Bernoulli}
\bvolume{10}
\bpages{97--120}.
\end{barticle}
\endbibitem

\bibitem{masuda2007ergodicity}
\begin{barticle}[author]
\bauthor{\bsnm{Masuda},~\bfnm{Hiroki}\binits{H.}}
(\byear{2007}).
\btitle{Ergodicity and exponential $\beta$-mixing bounds for multidimensional diffusions with jumps}.
\bjournal{Stochastic processes and their applications}
\bvolume{117}
\bpages{35--56}.
\end{barticle}
\endbibitem

\bibitem{matulewicz2017statistical}
\begin{bphdthesis}[author]
\bauthor{\bsnm{Matulewicz},~\bfnm{Gustaw}\binits{G.}}
(\byear{2017}).
\btitle{Statistical inference of {O}rnstein-{U}hlenbeck processes: generation of stochastic graphs, sparsity, applications in finance},
\btype{PhD thesis},
\bpublisher{Universit{\'e} Paris-Saclay}.
\end{bphdthesis}
\endbibitem

\bibitem{McNeil2005QuantRiskManagementEM}
\begin{bbook}[author]
\bauthor{\bsnm{McNeil},~\bfnm{Alexander~J.}\binits{A.~J.}}, \bauthor{\bsnm{Frey},~\bfnm{R{\"u}diger}\binits{R.}} \AND \bauthor{\bsnm{Embrechts},~\bfnm{Paul}\binits{P.}}
(\byear{2015}).
\btitle{{Quantitative Risk Management: Concepts, Techniques and Tools Revised edition}}.
\bseries{Economics Books}
\bvolume{10496}.
\bpublisher{Princeton University Press}.
\end{bbook}
\endbibitem

\bibitem{Melanson2019DataDrivenStationaryJumpDiffusions}
\begin{barticle}[author]
\bauthor{\bsnm{Melanson},~\bfnm{Alexandre}\binits{A.}} \AND \bauthor{\bsnm{Longtin},~\bfnm{Andr{\'e}}\binits{A.}}
(\byear{2019}).
\btitle{Data-driven inference for stationary jump-diffusion processes with application to membrane voltage fluctuations in pyramidal neurons}.
\bjournal{The Journal of Mathematical Neuroscience}
\bvolume{9}
\bpages{6}.
\bdoi{10.1186/s13408-019-0074-3}
\end{barticle}
\endbibitem

\bibitem{NguyenMichele2016SpatioTemporalOU}
\begin{barticle}[author]
\bauthor{\bsnm{Nguyen},~\bfnm{Michele}\binits{M.}} \AND \bauthor{\bsnm{Veraart},~\bfnm{Almut E.~D.}\binits{A.~E.~D.}}
(\byear{2017}).
\btitle{{Spatio-temporal Ornstein--Uhlenbeck Processes: Theory, Simulation and Statistical Inference}}.
\bjournal{Scandinavian Journal of Statistics}
\bvolume{44}
\bpages{46-80}.
\bdoi{10.1111/sjos.12241}
\end{barticle}
\endbibitem

\bibitem{pigorsch2009DefinitionSemiPositiveMultOU}
\begin{barticle}[author]
\bauthor{\bsnm{Pigorsch},~\bfnm{Christian}\binits{C.}} \AND \bauthor{\bsnm{Stelzer},~\bfnm{Robert}\binits{R.}}
(\byear{2009}).
\btitle{On the definition, stationary distribution and second order structure of positive semidefinite {O}rnstein-{U}hlenbeck type processes}.
\bjournal{Bernoulli}
\bvolume{15}
\bpages{754--773}.
\bdoi{10.3150/08-BEJ175}
\end{barticle}
\endbibitem

\bibitem{robins2011exponential}
\begin{barticle}[author]
\bauthor{\bsnm{Robins},~\bfnm{Garry}\binits{G.}}
(\byear{2011}).
\btitle{Exponential random graph models for social networks}.
\bjournal{Handbook of Social Network Analysis. Sage}.
\end{barticle}
\endbibitem

\bibitem{Shojaie2010EstimationAcyclicGraph}
\begin{barticle}[author]
\bauthor{\bsnm{Shojaie},~\bfnm{Ali}\binits{A.}} \AND \bauthor{\bsnm{Michalidis},~\bfnm{George}\binits{G.}}
(\byear{2010}).
\btitle{{Penalized likelihood methods for estimation of sparse high-dimensional directed acyclic graphs}}.
\bjournal{Biometrika}
\bvolume{97}
\bpages{519--538}.
\end{barticle}
\endbibitem

\bibitem{Simonov2017EventBased}
\begin{barticle}[author]
\bauthor{\bsnm{{Simonov}},~\bfnm{M.}\binits{M.}}, \bauthor{\bsnm{{Chicco}},~\bfnm{G.}\binits{G.}} \AND \bauthor{\bsnm{{Zanetto}},~\bfnm{G.}\binits{G.}}
(\byear{2017}).
\btitle{Real-Time Event-Based Energy Metering}.
\bjournal{IEEE Transactions on Industrial Informatics}
\bvolume{13}
\bpages{2813-2823}.
\end{barticle}
\endbibitem

\bibitem{stute1993paramBootstrap}
\begin{barticle}[author]
\bauthor{\bsnm{Stute},~\bfnm{Winfried}\binits{W.}}, \bauthor{\bsnm{Manteiga},~\bfnm{Wenceslao}\binits{W.}} \AND \bauthor{\bsnm{Quindimil},~\bfnm{Manuel}\binits{M.}}
(\byear{1993}).
\btitle{{Bootstrap based goodness-of-fit-tests}}.
\bjournal{Metrika: International Journal for Theoretical and Applied Statistics}
\bvolume{40}
\bpages{243--256}.
\end{barticle}
\endbibitem

\bibitem{Wu2019MoMCompoundPoissonOU}
\begin{barticle}[author]
\bauthor{\bsnm{Wu},~\bfnm{Yanfeng}\binits{Y.}}, \bauthor{\bsnm{Hu},~\bfnm{Jianqiang}\binits{J.}} \AND \bauthor{\bsnm{Zhang},~\bfnm{Xinsheng}\binits{X.}}
(\byear{2019}).
\btitle{{Moment estimators for the parameters of Ornstein-Uhlenbeck processes driven by compound Poisson processes}}.
\bjournal{Discrete Event Dynamic Systems}
\bvolume{29}
\bpages{57--77}.
\bdoi{10.1007/s10626-019-00276-y}
\end{barticle}
\endbibitem

\bibitem{zaccarin2010expRandomGraph}
\begin{binproceedings}[author]
\bauthor{\bsnm{Zaccarin},~\bfnm{Susanna}\binits{S.}} \AND \bauthor{\bsnm{Rivellini},~\bfnm{Giulia}\binits{G.}}
(\byear{2010}).
\btitle{Modelling Network Data: An Introduction to Exponential Random Graph Models}.
In \bbooktitle{Data Analysis and Classification}
(\beditor{\bfnm{Francesco}\binits{F.}~\bsnm{Palumbo}}, \beditor{\bfnm{Carlo~Natale}\binits{C.~N.}~\bsnm{Lauro}} \AND \beditor{\bfnm{Michael~J.}\binits{M.~J.}~\bsnm{Greenacre}}, eds.)
\bpages{297--305}.
\bpublisher{Springer Berlin Heidelberg}, \baddress{Berlin, Heidelberg}.
\end{binproceedings}
\endbibitem

\bibitem{zambon2019autoregressiveGraphs}
\begin{binproceedings}[author]
\bauthor{\bsnm{Zambon},~\bfnm{Daniele}\binits{D.}}, \bauthor{\bsnm{Grattarola},~\bfnm{Daniele}\binits{D.}}, \bauthor{\bsnm{Livi},~\bfnm{Lorenzo}\binits{L.}} \AND \bauthor{\bsnm{Alippi},~\bfnm{Cesare}\binits{C.}}
(\byear{2019}).
\btitle{{Autoregressive Models for Sequences of Graphs}}.
In \bbooktitle{2019 International Joint Conference on Neural Networks (IJCNN)}
\bvolume{In proceedings}
\bpages{1-8}.
\bdoi{10.1109/IJCNN.2019.8852131}
\end{binproceedings}
\endbibitem

\bibitem{zhu2017network}
\begin{barticle}[author]
\bauthor{\bsnm{Zhu},~\bfnm{Xuening}\binits{X.}}, \bauthor{\bsnm{Pan},~\bfnm{Rui}\binits{R.}}, \bauthor{\bsnm{Li},~\bfnm{Guodong}\binits{G.}}, \bauthor{\bsnm{Liu},~\bfnm{Yuewen}\binits{Y.}}, \bauthor{\bsnm{Wang},~\bfnm{Hansheng}\binits{H.}} \betal{et~al.}
(\byear{2017}).
\btitle{Network vector autoregression}.
\bjournal{The Annals of Statistics}
\bvolume{45}
\bpages{1096--1123}.
\end{barticle}
\endbibitem

\bibitem{zou2006adaptive}
\begin{barticle}[author]
\bauthor{\bsnm{Zou},~\bfnm{Hui}\binits{H.}}
(\byear{2006}).
\btitle{{The Adaptive Lasso and Its Oracle Properties}}.
\bjournal{Journal of the American Statistical Association}
\bvolume{101}
\bpages{1418--1429}.
\end{barticle}
\endbibitem

\end{thebibliography}

\end{document}